\documentclass[reqno]{amsart}
\usepackage{amssymb}
\usepackage{enumerate}
\usepackage[usenames,dvipsnames]{color}
\usepackage[bookmarks,colorlinks,breaklinks]{hyperref}  % PDF hyperlinks, with coloured links

\definecolor{dullmagenta}{rgb}{0.4,0,0.4}   % #660066
\definecolor{darkblue}{rgb}{0,0,0.4}
\hypersetup{linkcolor=red,citecolor=blue,filecolor=dullmagenta,urlcolor=darkblue} 

\newcommand{\eq}[1]{\eqref{#1}}

\theoremstyle{remark}
\theoremstyle{definition}
\theoremstyle{plain}

\newtheorem{theorem}{Theorem}[section]
\newtheorem{proposition}[theorem]{Proposition}
\newtheorem{lemma}[theorem]{Lemma}
\newtheorem{corollary}[theorem]{Corollary}

\newtheorem{definition}[theorem]{Definition}
\newtheorem{remark}[theorem]{Remark}

\newtheorem*{induction}{Induction hypothesis}
\newtheorem*{remark*}{Remark}

\numberwithin{equation}{section}

\DeclareMathOperator{\tr}{tr}
\DeclareMathOperator{\Ran}{Ran}
\DeclareMathOperator{\dist}{dist}

\DeclareMathOperator{\diam}{diam}

\newcommand{\pr}{\prime}

\newcommand\R{\mathbb R}
\newcommand\N{\mathbb N}

\newcommand\Z{\mathbb Z}

\renewcommand\P{\mathbb P}
\newcommand\E{\mathbb E}

\renewcommand\L{\mathrm{L}}

\newcommand{\cA}{\mathcal{A}}
\newcommand{\cB}{\mathcal{B}}
\newcommand{\cC}{\mathcal{C}}
\newcommand{\cD}{\mathcal{D}}
\newcommand{\cE}{\mathcal{E}}
\newcommand{\cF}{\mathcal{F}}

\renewcommand{\H}{\mathcal{H}}

\newcommand{\cJ}{\mathcal{J}}
\newcommand{\cK}{\mathcal{K}}
\newcommand{\cL}{\mathcal{L}}
\newcommand{\cM}{\mathcal{M}}

\newcommand{\cR}{\mathcal{R}}
\newcommand{\cS}{\mathcal{S}}
\newcommand{\cT}{\mathcal{T}}
\newcommand{\cU}{\mathcal{U}}
\newcommand{\cW}{\mathcal{W}}

\newcommand{\cZ}{\mathcal{Z}}

\newcommand{\Pb}{\boldsymbol{P}}

\newcommand{\bom}{\boldsymbol{\omega}}

\newcommand\eps{\varepsilon}

\newcommand\Chi{\raisebox{.2ex}{$\chi$}}

\newcommand{\abs}[1]{\left\lvert #1 \right\rvert}
\newcommand{\norm}[1]{\left\lVert #1 \right\rVert}
\newcommand{\scal}[1]{\left\langle #1 \right\rangle}
\newcommand{\set}[1]{\left\{ #1 \right\}}
\newcommand{\pa}[1]{\left( #1 \right)}
\newcommand{\br}[1]{\left[ #1 \right]}

\newcommand{\up}[1]{^{(#1)}}

\newcommand{\x}{\boldsymbol{x}}
\newcommand{\bolda}{\boldsymbol{a}}
\newcommand{\boldx}{\boldsymbol{x}}
\newcommand{\boldb}{\boldsymbol{b}}
\newcommand{\boldu}{\boldsymbol{u}}
\newcommand{\boldv}{\boldsymbol{v}}
\newcommand{\y}{\boldsymbol{y}}
\newcommand{\boldy}{\boldsymbol{y}}

\newcommand{\boldlambda}{\mathbf{\Lambda}}

\newcommand{\NboxLa}{\mathbf{\Lambda}_{L}^{(N)}(\bolda)}
\newcommand{\NboxLb}{\mathbf{\Lambda}_{L}^{(N)}(\boldb)}
\newcommand{\NboxLx}{\mathbf{\Lambda}_{L}^{(N)}(\boldx)}
\newcommand{\NboxLy}{\mathbf{\Lambda}_{L}^{(N)}(\boldy)}
\newcommand{\NboxLu}{\mathbf{\Lambda}_{L}^{(N)}(\boldu)}

\newcommand{\Nboxlx}{\mathbf{\Lambda}_{\ell}^{(N)}(\boldx)}
\newcommand{\Nboxly}{\mathbf{\Lambda}_{\ell}^{(N)}(\boldy)}
\newcommand{\Nboxlu}{\mathbf{\Lambda}_{\ell}^{(N)}(\boldu)}

\newcommand{\nboxLa}{\mathbf{\Lambda}_{L}^{(n)}(\bolda)}
\newcommand{\nboxLb}{\mathbf{\Lambda}_{L}^{(n)}(\boldb)}
\newcommand{\nboxLx}{\mathbf{\Lambda}_{L}^{(n)}(\boldx)}
\newcommand{\nboxLy}{\mathbf{\Lambda}_{L}^{(n)}(\boldy)}
\newcommand{\nboxLu}{\mathbf{\Lambda}_{L}^{(n)}(\boldu)}

\newcommand{\nboxx}{\mathbf{\Lambda}^{(n)}(\boldx)}
\newcommand{\nboxy}{\mathbf{\Lambda}^{(n)}(\boldy)}

\newcommand{\nboxlu}{\mathbf{\Lambda}_{\ell}^{(n)}(\boldu)}

\newcommand{\naboxLa}{ \boldlambda^{(n)}(\bolda) = \prod_{i = 1}^{n}  \,\, \Lambda_{L_i}(a_i)   }

\newcommand{\setn}{\bigl\{1, \, ... , \, n \bigr\}}
\newcommand{\setN}{\bigl\{1, \, ... , \, N \bigr\}}

\newcommand{\suitcx}{\Xi_{L,l}(\boldx)}
\newcommand{\suitc}{\Xi_{L,\ell}}

\newcommand{\Ndspace}{\mathbb{Z}^{Nd}}
\newcommand{\ndspace}{\mathbb{Z}^{nd}}

\newcommand{\dspace}{\mathbb{Z}^{d}}

\newcommand{\Bl}{\Bigl}
\newcommand{\Br}{\Bigr}

\newcommand\beq{\begin{equation}}
\newcommand\eeq{\end{equation}}

\newcommand{\qtx}[1]{\quad\text{#1}\quad}
\newcommand{\sqtx}[1]{\;\text{#1}\;}

\begin{document}

\title[Multi-particle Anderson model]
{The bootstrap multiscale analysis  for the  multi-particle Anderson model}

\author{Abel Klein and Son T. Nguyen}
\address{University of California, Irvine,
Department of Mathematics,
Irvine, CA 92697-3875,  USA}
 \email[A. Klein]{aklein@uci.edu}
  \email[S.T. Nguyen]{sondgnguyen1@gmail.com}

\thanks{A.K. was  supported in part by the NSF under grant DMS-1001509.}

%\date{Version of \today}

\begin{abstract}
 We extend the bootstrap multi-scale analysis developed by Germinet and Klein   to the multi-particle Anderson model, obtaining  Anderson localization, dynamical localization, and decay of eigenfunction correlations. \end{abstract}

\maketitle

\tableofcontents

\section{Introduction}

Localization is by now well understood for the Anderson model, a random Schr\" odinger  operator that describes an electron moving in a medium with random  impurities (e.g., the review 
\cite{Ki}).  More recently, localization has been proved for a multi-particle Anderson model with a  finite range  interaction potential by Chulaevsky and Suhov \cite{CS1,CS2,CS3} and Aizenman and Warzel \cite{AWmp}. Chulaevsky and Suhov  used a multiscale analysis based on  \cite{vDK} and  Aizenman and Warzel \cite{AWmp}  employed the fractional moment method as in  \cite{ASFH}.  Chulaevsky, Boutet de Monvel, and Suhov \cite{CBS} extended the results of Chulaevsky and Suhov to 
the  continuum multi-particle Anderson model. 
  
In this article we extend the bootstrap multi-scale analysis developed by Germinet and Klein \cite{GK1,Kl}  to the multi-particle Anderson model, obtaining  Anderson localization, dynamical localization, and decay of eigenfunction correlations. The advantage of our method is that it extends to the 
 continuum multi-particle Anderson model, yielding the strong localization results 
proven in  \cite{GK1,GKjsp,Kl} for the one particle continuum Anderson model.  This extension will appear in a sequel to this paper.

We start by defining the $n$-particle  Anderson model.

\begin{definition}\label{defAndmodel}
The $n$-particle Anderson model is  the  random Schr\"odinger 
operator on 
$\ell^{2}(\mathbb{Z}^{nd})$ given by
\beq \label{AndH}
H_{\bom}^{(n)}: =  -\Delta^{(n)} + \,V_{\bom}^{(n)} + U ,
\eeq
where:
\begin{enumerate}
\item 
$\Delta^{(n)}$ is the discrete $nd$-dimensional  Laplacian operator.
\item 
$\bom=\{ \omega_x \}_{x\in
\Z^d}$ is a family of independent 
identically distributed random
variables  whose  common probability 
distribution $\mu$ has a bounded density $\rho$ with compact support.
\item 
$V_{\bom}^{(n)}$ is the random potential  given by
\beq 
V_{\bom}^{(n)}(\x)= \sum_{i = 1, ..., n} V_{\bom}^{(1)}(x_i),\,\, \x =(x_1, ..., x_n) \in \mathbb{Z}^{nd} ,
\eeq   
where $V_{\bom}^{(1)}(x)=\omega_{x}$ for every $x \in \mathbb{Z}^{d}$.
\item
$U$ is a potential governing the short range interaction between the $n$ particles.  We take 
\beq\label{2body}
U(\x) = \sum_{1 \leq i < j \leq n } \widetilde{U}(x_i - x_j) ,
\eeq
where $\widetilde{U}\colon \Z^d  \to \R$, $\widetilde{U}(y)= \widetilde{U}(-y)$, and  $\widetilde{U}(y) = 0$ for $\norm{y}_\infty > r_0$ for some $0<r_0 < \infty$.
\end{enumerate}
\end{definition}

\begin{remark*} We took  a two-body interaction potential in \eq{2body} for simplicity, but our  results would still be valid with a more general finite range interaction potential as in \cite{AWmp}.
\end{remark*}

We will generally omit ${\bom}$ from the notation, and 
use the following notation:
\begin{enumerate}
\item  Given $x = (x_1, \ldots, x_d) \in \R^{d}$, we set $\norm{x}=\norm{x}_\infty  := \max \{\abs{x_1}, \ldots, \abs{x_d} \}$.  If  $\bolda = (a_1, \ldots, a_n) \in \R^{nd}$, we let
  $\norm{\bolda} := \max \{ \norm{a_1}, \ldots, \norm{a_n} \}$, $\scal{\bolda} := \sqrt{ 1+ \norm{\bolda}^{2}}$, and $\cS_{ \bolda} = \bigl \{a_1, \,...,\, a_n   \bigr\}$.

\item Given $\bolda, \boldb  \in \R^{nd}$, we set $d_{H} ( \bolda, \, \boldb):= d_{H} ( \cS_{ \bolda}, \, \cS_{ \boldb})$, 
where  $d_{H}(S_1, \,S_2)$ denotes the      the Hausdorff distance between  finite subsets   $S_1, \,\, S_2 \subseteq \R^{d}$,  given by
\begin{align}
d_{H}(S_1, \,S_2)& := \max \Bl\{ \max_{x \in S_1} \, \min_{y \in S_2} \norm{x - y} \, , \,   \max_{y \in S_2} \, \min_{x \in S_1} \norm{x - y}  \Br\}\\
&= \max \Bl\{ \max_{x \in S_1} \, \dist(x, \,\,S_2) \, , \,   \max_{y \in S_2} \, \dist(y, \,\,S_1)  \Br\}.
\notag
\end{align}

\item We use $n$-particle boxes in $\ndspace$ centered at points in $\R^{nd}$.   The $n$-particle box  of side $L\ge 1$  centered at $\boldx=(x_1,x_2,\ldots,x_n) \in \R^{nd}$ is given by
\beq \label{defbox}
\nboxLx= \set{\y \in \ndspace; \norm{\y-\x} \le  \tfrac{L}{2}}= \prod_{i = 1}^{n}  \,\, \Lambda_L(x_i)\subseteq \ndspace.
\eeq
 By a box $\boldlambda_L$ in $\ndspace$ we mean an $n$-particle box  $\nboxLx$.   Note that  
 \beq (L-2)^{nd}<\abs  {\nboxLx}\le (L+1)^{nd}.
 \eeq
Since we always work with $L$ large, we will use $\abs  {\nboxLx}\le L^{nd}$ and ignore the small error.
\item  
 We will occasionally use boxes in  $\R^{nd}$.  We set
 \beq \label{defboxR}
\widehat{\boldlambda}_L^{(n)}(\x)= \set{\y \in \R^{nd}; \norm{\y-\x} \le  \tfrac{L}{2}}; \qtx{note}  
\nboxLx= \widehat{\boldlambda}_L^{(n)}\cap \Z^{nd}.
\eeq

\item Given a box  $\boldlambda_{t}^{(n)}\subseteq \boldlambda_L^{(n)}= \nboxLx$, we let
\begin{align}
\partial^{ \boldlambda_L^{(n)}} \boldlambda_{t}^{(n)}&=\set{(\boldu,\boldv) \in \boldlambda_{t}^{(n)}\times\pa{\boldlambda_L^{n)}\setminus \boldlambda_{t}^{(n)}}\, |\, \norm{u-v}_1=1},
\\
\partial_+^{ \boldlambda_L^{(n)}} \boldlambda_{t}^{(n)}&=\set{\boldv \in\boldlambda_L^{(n)}\setminus \boldlambda_{t}^{(n)}\, |\, (\boldu, \, \boldv) \in \partial \boldlambda_{t}^{(n)}\qtx{for some} \boldu \in \boldlambda_{t}^{(n)}}.\notag
\end{align}
Note that there exists a constant $s_{Nd} $ such that for $t\ge 1$ we have
\beq
\abs{\partial_+^{ \boldlambda_L^{(n)}} \boldlambda_{t}^{(n)}}\le \abs{\partial^{ \boldlambda_L^{(n)}} \boldlambda_{t}^{(n)}}\le s_{nd} t^{nd}.
\eeq
When it is clear that  $\boldlambda_{t}^{(n)}\subseteq \boldlambda_L^{(n)}$ we will simply write 
$\partial \boldlambda_{t}^{(n)}$ and $\partial_+\boldlambda_{t}^{(n)}$.

\item Given an $n$-particle box $\boldlambda=\nboxLx$, we define the finite volume operator  $H_{\boldlambda} = H_{\nboxLx} ^{(n)}$  as the self-adjoint operator on  $\ell^2 \left( \boldlambda \right)$  obtained by restricting  $H^{(n)}$  to  $ \boldlambda $ with Dirichlet (simple)  boundary condition:  $H_{\boldlambda} = \Chi_{\boldlambda} H^{(n)} \Chi_{\boldlambda}$  restricted to   $\ell^2 \left( \boldlambda \right)$.  If 
$z \notin \sigma \left( H_{\boldlambda}   \right)$, we set 
\beq  G_{\boldlambda}(z)= (H_{\boldlambda}  -z)^{-1}, \quad
G_{\boldlambda}(z;\boldu,\y) = \scal{ \delta_{\boldu}, (H_{\boldlambda}  -z)^{-1}\delta_{\y}} \;\; \text{for} \;\;  \boldu, \, \y \in \mathbf{\boldlambda}.
\eeq
\end{enumerate}

We will use several types of good boxes. Note that they are defined for a fixed   $\bom$ (omitted from the notation).

\begin{definition} Let $\boldlambda=\nboxLx$ be an $n$-particle box  and let $E\in \R$. Then:
\begin{enumerate}
\item 
Given  $\theta > 0$, the  $n$-particle box $\boldlambda$ is said to be $(\theta, E)$-suitable if, and only if, 
$E \notin \sigma \Bl(H_{\boldlambda}  \Br)$  and
\beq
\left|G_{\boldlambda}(E; \bolda, \boldb)\right| \leq L^{-\theta}\qtx{for all}
 \bolda, \boldb \in \boldlambda \qtx{with}  \norm{\bolda - \boldb} \geq \tfrac{L}{100} .\eeq
Otherwise, $\boldlambda$ is called $(\theta, E)$-nonsuitable.

\item Given  $\zeta \in (0,1)$,  the  $n$-particle box $\boldlambda$ is said to be $(\zeta, E)$-subexponentially suitable (SES) if, and only if, 
$E \notin \sigma \Bl(H_{\boldlambda}  \Br)$  and
\beq
 \left|G_{\boldlambda}(E; \bolda, \boldb)\right| \leq e^{-L^{\zeta}}
\qtx{for all}\bolda, \boldb \in \boldlambda \qtx{with}  \norm{\bolda - \boldb} \geq \tfrac{L}{100} .
\eeq
Otherwise,  $\nboxLx$ is called  $(\zeta, E)$-nonsubexponentially suitable (nonSES).

\item  Given  $m > 0$, the  $n$-particle box $\boldlambda$ is said to be  $(m, E)$-regular if, and only if, 
$E \notin \sigma \Bl(H_{\boldlambda}  \Br)$  and
\beq
 \left|G_{\boldlambda}(E; \bolda, \boldb)\right| \leq e^{-m\norm{\bolda - \boldb}}
\qtx{for all}\bolda, \boldb \in \boldlambda \qtx{with}  \norm{\bolda - \boldb} \geq \tfrac{L}{100} .
\eeq
Otherwise, $\boldlambda$ is called $(m, E)$-nonregular.

\end{enumerate}
\end{definition}

\begin{remark}  \label{goodbox}
It follows immediately  from the definitions that:
\begin{enumerate}
\item $\nboxLx$  $(m^{*}, E)$-regular  \ $\Longrightarrow$  \ $\nboxLx$  
$\left(\tfrac{  m^{*}L  }{ 100 \log L  }, E\right)$-suitable.
\item $\nboxLx$  $(\theta, E)$-suitable  \ $\Longrightarrow$  \ $\nboxLx$ \,  $\left(\tfrac{\theta \, logL}{L}, E\right)$-regular.
\item $\nboxLx$ $\Bl(L^{\zeta-1}, E \Br)$-regular  \ $\Longrightarrow$  \ $\nboxLx$    $\left(\zeta-\tfrac{log 100}{log L}, E\right)$-SES.
\item 
$\nboxLx$   $\bigl(\zeta, E \bigr)$-SES  \ $\Longrightarrow$  \ $\nboxLx$   $(L^{\zeta-1}, E)$-regular.
\end{enumerate}
\end{remark}

We are ready to state our main result, which
 extends the bootstrap multiscale analysis of Germinet and Klein \cite{GK1} to the multi-particle Anderson model with short range interaction. 
\begin{theorem}  \label{maintheorem}
   There exist  $p_{0} (n)=p_0(d,n)>0$, $n=1,2,\ldots$,
   with the property that for every   $N \in \N$,  given $\theta > 8Nd$,   there exists ${\cL}=\cL(d,\norm{\rho}_\infty, N, \theta)$ such that if for some $L_0\ge \cL$ we have 
\beq \label{condpn}
\sup_{\x \in \R ^{nd}} \P  \Bigl\{ \Lambda_{L_0}^{(n)} (\x) \;\; \text{is} \;\; (\theta,\,E)\text{-nonsuitable} \Bigr\} \leq p_{0}(n) ,
\eeq
for every  $E \in \R$ and every $n=1,2,\ldots,N$,  
then, given $0< \zeta  <1$, we can find a length scale $L_{\zeta} = L_\zeta(d,\norm{\rho}_\infty, N, \theta,L_0)$,  $\delta_\zeta = \delta_{\zeta}(d,\norm{\rho}_\infty, N, \theta,L_0)>0$, and $m_{\zeta} = m_{\zeta}(\delta_\zeta,L_{\zeta}) > 0$,  so that the following holds for $n = 1, 2, ..., N$:

\begin{enumerate}
\item For every $E \in \R$,  $L \geq L_{\zeta}$, and $\bolda \in \R^{nd}$,   we have 
\begin{align} 
\P \Bigl\{   \nboxLa  \;\; \text{is} \;\;  \left ({m_\zeta}, \,E \right )\text{-nonregular} \Bigr\}    \leq e^{-L^{\zeta}}.
\end{align}
\item Given  $E_1 \in \R$, set  $I(E_1)=[E_1-\delta_{\zeta}, E_1+\delta_{\zeta}]$.  Then, for  every $E_1 \in \R$,  $L \geq L_{\zeta}$, and   $\bolda,\boldb \in \R^{nd}$
  with $d_{H} ( \bolda, \, \boldb) \ge L$,  we have 
\begin{align} \label{concmsa}
\P \Bigl\{ \exists \, E \in I(E_1)\;\,\text{so} \;\,  \nboxLa \;\, \text{and} \;\,  \nboxLb  \;\, \text{are} \;\,  \left ({m_\zeta}, \,E \right )\text{-nonregular} \Bigr\} \leq e^{-L^{\zeta}}. 
\end{align}
\end{enumerate}
\end{theorem}

\begin{remark} The hypotheses of Theorem \ref{maintheorem} are verified at high disorder.   Consider the $n$-particle Anderson model given in Definition~\ref{defAndmodel} with a disorder parameter $\lambda>0$ (cf. \eq{AndH}):
\beq \label{AndHlambda}
H_{\bom,\lambda}^{(n)}: =  -\Delta^{(n)} + \lambda\,V_{\bom}^{(n)} + U.
\eeq
 $H_{\bom,\lambda}$ can be rewritten as an $n$-particle Anderson model $H\up{\lambda}_{\bom} $ in the exact form of Definition~\ref{AndH} by replacing the probability distribution $\mu$ by the probability distribution $\mu\up{\lambda}$, defined by 
$\mu\up{\lambda}(B) = \mu(\lambda^{-1}B)$ for all Borels sets $B\subset \R$, with   density $\rho\up{\lambda}(t)=\frac 1{\lambda} \rho(\frac t{\lambda})$. Proceeding as in \cite[Proposition~3.1.2]{vDK}, we can show that
for all $N\in \N$, given a scale  $L_0$,
 there exists $\lambda_N <\infty$, such that for all   $\lambda\ge \lambda_N$ the condition \eq{condpn} is satisfied at scale $L_0$  by  $H_{\bom,\lambda}^{(n)} $  for every  $E \in \R$ and every $n=1,2,\ldots,N$.   Since  $\norm{\rho\up{\lambda}}_\infty= \frac 1{\lambda}\norm{\rho}_\infty\le \norm{\rho}_\infty$ for $\lambda \ge 1$, and $\norm{\rho\up{\lambda}}_\infty$ is the only constant that appears in the proof of the theorem that changes with $\lambda$, it follows that the conclusions of Theorem \ref{maintheorem}  are valid for for all $\lambda\ge \lambda_N$ with the same constants $L_\zeta, \delta_\zeta, m_\zeta$.
\end{remark}

Theorem \ref{maintheorem} is proved in Section~\ref{inductionhyp}. The theorem is proved by induction on the number of  particles.  The one particle case was proven in \cite{GK1,Kl}. (These papers deal with the continuum Anderson model, but the results  apply to the discrete Anderson model.) The proof of the induction step  proceeds as in  \cite{GK1,Kl}, with four multi-scale analyses, using  some  technical arguments of \cite{GKber}.  To deal with the fact that in the multi-particle case events based on disjoint boxes are not independent,     we use  the  partially and fully separated boxes and  partially and fully interactive boxes introduced by Chulaevsky and Suhov \cite{CS1,CS2,CS3}.  The relevant  distance between boxes is the Hausdorff distance, introduced in this context by  Aizenman and Warzel \cite{AWmp}.  We prove a Wegner estimate (Theorem~\ref{Wegner0}) and a Wegner estimate  between partially separated boxes (Theorem~\ref{Wegner2}).  In the multiscale analysis partially interactive boxes are handled by the induction hypothesis, i.e., by the conclusions of Theorem \ref{maintheorem} for a smaller number of particles (see Lemma~\ref{PIsuit}), and fully interactive boxes  are handled similarly to one particle boxes (see Lemma~\ref{part1prop2}).

 Theorem \ref{maintheorem} implies localization: Anderson localization, dynamical localization, and estimates on the behavior of eigenfunctions.

\begin{corollary} \label{PPS}
Assume  the conclusions of Theorem \ref{maintheorem}. Then:

\begin{enumerate}
 \item \emph{(Anderson localization)}   $H_{\bom}^{(N)}$ has pure point spectrum with exponentially decaying eigenfunctions for $\P$-a.e.\ $\bom$.

\item \emph{(Dynamical Localization)} For every $\y \in \Ndspace $ we can find a constant $C(\y)$ such that
\beq 
\E \set{\sup_{t \in \R} \abs{\scal{ \delta_{\x},  e^{-itH_{\bom}^{(N)}} \delta_{\y}  }   } } \leq C(\y) \, e^{-d_{H}(\x, \, \y)}\qtx{for every} \x \in \Ndspace.
\eeq

\item  \emph{(Summable Uniform Decay of Eigenfunction Correlations (SUDEC))} Fix $\nu > \tfrac{Nd}{2} + \tfrac{1}{2}$ and let  $T$ be the  operator on $\H$ given by multiplication  by the function  $\scal{\x}^{2\nu}$. Then, for $\P$-a.e.\ $\bom$
 $H_{\bom}$ has pure point spectrum in the open interval $I$ with finite multiplicity,  and  for  every $\zeta \in (0, \, 1)$  there exists a constant $C_{\bom,\zeta}$ such that for every  eigenvalue $E$ of $H_{\bom}^{(N)}$ and  $\psi, \, \phi \in \Ran \Chi_{\{E\}} (H_{\bom}) $, we have that, for all   $\x, \, \y \in \Ndspace$,
\begin{align}  
&\abs{\phi(\boldx)} \abs{\psi(\boldy)} \leq C_{\bom,\zeta} \norm{T^{-1}\phi}\norm{T^{-1}\psi} \scal{\x}^{2\nu} e^{-d_{H}(\boldx - \boldy)^{\zeta}}, \quad \text{and}\\
& \abs{\phi(\boldx)}\abs{\psi(\boldy)}\leq C_{\bom,\zeta} \norm{T^{-1}\phi}\norm{T^{-1}\psi} \scal{\x}^{\nu}  \scal{\y}^{\nu} e^{-d_{H}(\boldx - \boldy)^{\zeta}}.
\end{align}
\end{enumerate}
\end{corollary}

Corollary~\ref{PPS} is proven in Section~\ref{secloc}.

\section{Preliminaries to the multiscale analysis}

\subsection{Partially and fully separated boxes} We call 
     $\naboxLa$ an $n$-particle rectangle centered at $\bolda \in \R^{nd}$.  Given   subsets $\cJ, \cK \subseteq \setn $, with $\cK\not=\emptyset$,  we set
\begin{gather}\notag
\Pi_{i} \boldlambda^{(n)}(\bolda)  = \Lambda_{L}(a_i), \;
\Pi_{\cJ} \boldlambda^{(n)}(\bolda) = \bigcup_{i \in \cJ} \Lambda_{L}(a_i), \;
\Pi \boldlambda^{(n)}(\bolda)  = \Pi_{\setn} \boldlambda^{(n)}(\bolda),\\
\mathbf{a}_{\cK} = (a_i \,\, , \,i \in \cK),\quad\mathbf{a}=(\mathbf{a}_{\cK}, \mathbf{a}_{\cK^c}), \quad
\boldlambda(\bold{a}_{\cK})=\boldlambda^{\cK}(\bold{a}_{\cK}) = \prod_{i \in \cK} \Lambda_{L_i} (a_i). \notag
\end{gather}

\begin{definition}
Let $\nboxx= \prod_{i = 1}^{n}  \,\, \Lambda_{L_i}(x_i)$  and    $\nboxy= \prod_{i = 1}^{n}  \,\, \Lambda_{\ell_i}(y_i)$ be a pair of $n$-particle rectangles.   
\begin{enumerate}
\item  $\nboxx$ and $\nboxy$ are partially separated if, and only if, either  $\Lambda_{L_i}(x_i) \,\cap \, \Pi \nboxy = \emptyset$ {for some}  $i \in \setn$, or $\Lambda_{\ell_j}(y_j) \,\cap \, \Pi \nboxx = \emptyset $ for some $j \in \setn$.

\item  $\nboxx$ and $\nboxy$ are fully separated if, and only if, 
\beq 
\biggl(\Pi \nboxx  \biggr) \,\,  \bigcap \,\, \biggl(\Pi \nboxy \biggr) = \emptyset.
\eeq

\end{enumerate}
\end{definition}

Given a pair of $n$-particle rectangles  $\nboxx$ and $\nboxy$ as above, with $L_i,\ell_i\le L$ for all $\in \setn$,  if there exists $i \in \setn$ such that $\norm{x_i - y_j} \ge L$ for every $j \in \setn$, then $\Lambda_{L_i}(x_i) \,\, \cap \,\, \Pi \nboxy = \emptyset$. In other words, if there exists $i \in \setn$ such that $\dist(x_i, \,\, \cS_{\y}) \ge  L$, then $\Lambda_{L_i}(x_i) \,\, \cap \,\, \Pi \nboxy = \emptyset$. 
We have the following lemma.

\begin{lemma}\label{wegnerboxes}
Let $\nboxx= \prod_{i = 1}^{n}  \,\, \Lambda_{L_i}(x_i)\subseteq  \nboxLx $  and    $\nboxy= \prod_{i = 1}^{n}  \,\, \Lambda_{\ell_i}(y_i)\subseteq  \nboxLy $ be a pair of $n$-particle rectangles.   Then
\begin{enumerate}
\item  $\nboxx$ and $\nboxy$ are partially separated if $d_{H}(\x, \, \y) >L$.

\item  $\nboxLx$ and $\nboxLy$ are fully separated if   $\dist(\cS_{ \x}, \,\,\cS_{\y})> L$. 

\end{enumerate}
\end{lemma}

\subsection{Wegner estimates}

Wegner estimates have been previously  proved for   the $n$-particle Anderson model (e.g.,  \cite{CS1, CS3}).   We derive optimal Wegner estimates, that is, with the expected dependence on the volume and  interval length.

\begin{theorem} \label{Wegner0}
Consider the  $n$-particle rectangle  $\boldlambda = \prod_{i = 1, \dots, n} \Lambda_{L_i} (a_i) \subseteq \nboxLa $ and let  $ \Gamma = \Lambda_{L_k} (a_k) $ for some $k \in\set{1, \dots, n}$.  Then for any interval $I$  we have
\beq
\E_{\Gamma } \, \Bl( tr \,\Chi_{I} \left( H_{\bom, \, \boldlambda}   \right)  \Br) \leq n\, \norm{\rho}_{\infty} \abs{I} L^{nd} .
\eeq
In particular,  for any $ E \in \R$ and  $ {\eps} > 0$ we have
\beq \label{weggamma}
\P_\Gamma \Bl\{ \norm{G_{\boldlambda}(E)} \geq \tfrac{1}{{\eps}} \Br\} = \P_\Gamma  \Bl\{d\,(\sigma(H_{\boldlambda}), E) \leq {\eps} \Br\} \leq 2\, n\, \, \,
\norm{\rho}_\infty\, {\eps} L^{nd}.
\eeq
\end{theorem}

\begin{proof}
We begin by rewriting $H_{\bom, \,\boldlambda}^{(n)}$  as
\beq
H_{\bom, \,\boldlambda}^{(n)} = - \Delta _{\boldlambda} + U_{\boldlambda}+\sum_{\boldx \in \boldlambda} \sum_{i = 1, \dots, n} \bom_{x_i} \Pi_{\boldx},
\eeq
where $\Pi_{\x}$ denotes the rank one orthogonal projection onto $\delta_{\x}$. 
Given $y \in \dspace$, we set   $q_{y}(\boldx) = \# \left\{ i = 1, \dots, n \,\,|\,\, x_i = y \right\}$ for $\boldx = (x_1, \dots, x_n)\in \Z^{nd} $. Then (see \cite{AWmp})
\begin{align} \notag
H_{\bom, \,\boldlambda}^{(n)} & = - \Delta _{\boldlambda} + U_{\boldlambda} + \sum_{\boldx \in \boldlambda} \sum_{y \in \dspace} q_{y} (\boldx) \bom_{y} \Pi_{\boldx}  = - \Delta _{\boldlambda} + U_{\boldlambda} + \sum_{y \in \dspace} \bom_{y} \sum_{\boldx \in \boldlambda} q_{y} (\boldx)  \Pi_{\boldx} \\ 
&  = - \Delta _{\boldlambda} + U_{\boldlambda} + \sum_{y \in \dspace} \bom_{y} \Theta_{y}  = - \Delta _{\boldlambda} + U_{\boldlambda} + \sum_{y \notin \Gamma} \bom_{y} \Theta_{y}  + \sum_{y \in \Gamma} \bom_{y} \Theta_{y} ,
\end{align}
where $\Theta_{y}  =  \sum_{ \boldx \in \boldlambda  } q_{y} (\boldx)  \Pi_{\boldx}$. 

Let $I$ be an interval. Given 
 $\tilde{y} \in \Gamma$, we set
\beq
\widetilde{H}_{(\bom_{\tilde{y}})^{\perp},\, \boldlambda} =H_{\bom, \,\boldlambda}^{(n)}  -    \bom_{\tilde{y}} \Theta_{\tilde{y}}  
= - \Delta _{\boldlambda} + U_{\boldlambda} + \sum_{y \notin \Gamma} \bom_{y} \Theta_{y} + \sum_{y \in \Gamma \setminus \set{\tilde{y}}} \bom_{y} \Theta_{y} .
\eeq

As functions on $\ndspace$, we have
$
\Chi_{\boldlambda} \leq \sum_{y \in \Gamma} \Theta_{y} 
$,
so 
\begin{align}
\tr \, \Chi_{I}(H_{\bom, \, \boldlambda}) & \le  \sum_{y \in  \Gamma} \tr \, \Bigl( \Theta_{y}  \Chi_{I}(H_{\bom, \, \boldlambda}) \Bigr)  = \sum_{y \in  \Gamma} \tr \, \bigl( \Theta_{y} \, \Chi_{I}( \widetilde{H}_{(\bom_{y} )^{\perp}, \, \boldlambda} + \bom_{y} \Theta_{y} ) \bigr) .
\end{align}
Hence, 
\begin{align}
&\E_{\Gamma} \bigl( \tr \, \Chi_{I} (H_{\bom, \, \boldlambda})  \bigr)  \leq \sum_{y \in  \Gamma} \E_{\bom_{y}^{\perp}} \E_{\bom_y} \Bigl \{ \tr \, \sqrt{\Theta_y} \,\,
\Chi_{I} \bigl ( \widetilde{H}_{(\bom_{y})^{\perp}, \boldlambda }  + \bom_{y} \Theta_{y} \bigr )  \sqrt{\Theta_y} \Bigr \} \notag \\
& \leq \sum_{y \in \Gamma} \E_{\bom_{y}^{\perp}}\set{ \norm{\rho}_{\infty} \abs{I} n L^{(n-1)d}}
 \leq \abs{\Gamma} \norm{\rho}_{\infty} \abs{I} n L^{(n-1)d},
\end{align}
since $\dim \,\Ran\, \Theta_y \leq n \, L^{(n-1)d}.$
But $\abs{\Gamma} \le L^d$, so we conclude that \beq
\E_{\Gamma} \bigl( \tr \, \Chi_{I} (H_{\bom, \, \boldlambda})  \bigr)  \leq n \, \norm{\rho}_{\infty} \abs{I}  L^{nd}.
\eeq
\end{proof}

\begin{corollary}  \label{Wegner2}  Let 
  $\boldlambda_1 = \prod_{i = 1, \dots, n} \Lambda_{L_i} (a_i) \subseteq \nboxLa $ and $\boldlambda_2 = \prod_{i = 1, \dots, n} \Lambda_{L^\pr_i} (b_i) \subseteq \nboxLa $ be a pair of partially separated $n$-particle rectangles.   Then
\beq \label{wegboxes}
\P \Bl\{ d\,\Bl(\sigma (H_{\boldlambda_{1}}), \sigma (H_{\boldlambda_{2}}) \Br) \leq {\eps} \Br\} \leq 2 \, n \norm{\rho}_\infty \,{\eps} L^{2nd} \qtx{for all} \eps>0.
\eeq
\end{corollary}

\begin{proof}   Let $\boldlambda_1, \, \boldlambda_2$ be as above.  Since they are partially separated  there is,   $\Gamma= \Lambda_{L_k} (a_k)$, such that $\Gamma \cap \Pi  \boldlambda_2 =\emptyset$.  Note that  $H_{\boldlambda_2}$ depends only on   $\bom_{\Gamma^c}$, and thus $\sigma (H_{\boldlambda_2}) = \Bl \{ E_1{(\bom_{\Gamma^c})}, \dots, \E_{\abs{\boldlambda_2}}{(\bom_{\Gamma^c})} \Br\}$,
where  $ E_j{(\bom_{\Gamma^c})}$ is independent of $\bom_{\Gamma}$.
 Thus
\begin{align}\notag
& \P \Bl\{ d\,\Bl(\sigma (H_{\boldlambda_{1}}), \sigma (H_{\boldlambda_{2}}) \Br) \leq {\eps} \Br\} = \E_{\Gamma^c} \P_{\Gamma} \Bl\{ d\,\Bl(\sigma (H_{\boldlambda_{1}}), \sigma (H_{\boldlambda_{2}}) \Br) \leq {\eps} \Br\} \\
& = \E_{\Gamma^c} \P_{\Gamma}  \Bl\{ d\,\Bl(\sigma (H_{\boldlambda_{1}}), E_j {(\bom_{\Gamma^c})}\Br) \leq {\eps} \,\,\text{for some} \,\, j = 1, \dots, \abs{\boldlambda_2}\Br\} \\
& \leq \E_{\Gamma^c} \sum_{j = 1, \dots, \abs{\boldlambda_2} } \P_{\Gamma}  \Bl\{ d\,\Bl(\sigma (H_{\boldlambda_{1}}), E_j {(\bom_{\Gamma^c})} \Br) \leq {\eps}  \Br\} \leq \abs{\boldlambda_2} \pa{2n\norm{\rho}_\infty {\eps}L^{nd} },   \notag
\end{align}
using \eq{weggamma}. The estimate   \eq{wegboxes} follows since  $\abs{\boldlambda_2} \leq L^{nd}$.
\end{proof}

\subsection{Partially  and fully interactive boxes}

Following Chulaevsky and Suhov \cite{CS2,CS3}, we divide  boxes into partially and fully interactive.

\begin{definition}
An $n$-particle box $\nboxLa$ is said to be partially interactive \emph(PI\emph) if and only if there exists a nonempty proper subset $\cJ \subsetneq \setn$ such that
\beq 
\nboxLa \subseteq \cE_{\cJ}, \qtx{where} \cE_{\cJ} = \Bl\{ \boldx \in \ndspace \,\, \vert \,\, \min_{i \in \cJ, \, j \notin \cJ} \norm{x_i - x_j} > r_0   \Br\}.
\eeq
If $\nboxLa$ is not partially interactive, it is said to be  fully interactive \emph(FI\emph).
\end{definition}

\begin{remark}\label{piremark}
If the $n$-particle box $\nboxLu$ is partially interactive, by writing
 $\nboxLu = \boldlambda_{L}^{{\cJ}}(\bold{u}_{\cJ}) \times  \boldlambda_{L}^{{\cJ}^{c}}(\bold{u}_{\cJ^{c}} ) $  we are implicitly stating that   $\nboxLu \subseteq \cE_{\cJ}$ for some nonempty proper subset $\cJ \subsetneq \setn$. Moreover, we set $\sigma_\cJ =\sigma \left( H_{\boldlambda_{L}^{{\cJ}}(\bold{u}_{\cJ})}  \right)  $ and $\sigma_{\cJ^{c}}=\sigma \left( H_{ \boldlambda_{L}^{{\cJ}^{c}}(\bold{u}_{\cJ^{c}} )  }  \right)$.
  \end{remark}
 
 \begin{lemma} Let $\nboxLu = \boldlambda_{L}^{{\cJ}}(\bold{u}_{\cJ}) \times  \boldlambda_{L}^{{\cJ}^{c}}(\bold{u}_{\cJ^{c}} ) $ be a PI $n$-particle box. Then:

\begin{enumerate}
\item $ \Pi_{\cJ} \nboxLu  \,\, \bigcap \,\, \Pi_{\cJ^{c}} \nboxLu = \emptyset$. 
\item $H_{\nboxLa} = H_{\bold{\Lambda}_{L}^{{\cJ}}(\bold{u}_{\cJ})} \otimes I_{\bold{\Lambda}_{L}^{ {\cJ}^{c}}(\bold{u}_{\cJ^{c}})} + I_{\bold{\Lambda}_{L}^{{\cJ}}(\bold{u}_{\cJ})} \otimes H_{\bold{\Lambda}_{L}^{ {\cJ}^{c}}(\bold{u}_{\cJ^{c}})}$.

\item  For all $E \notin \sigma \left( \nboxLu   \right)  $ and $\bolda,\,\boldb \in \nboxLu$ we have
\begin{align}\label{GJ}
& \abs{G_{\nboxLu}  (E;\,\bolda,\,\boldb)}  \leq  \sum_{\lambda \in \sigma_{\cJ} } 
\abs{G_{ \boldlambda_{L}(\bold{u_{\cJ^{c}}})  } (E - \lambda;\,\bolda_{\cJ^{c}},\,\boldb_{\cJ^{c}})}, 	 \\
& \abs{G_{\nboxLu}  (E;\,\bolda,\,\boldb)} \leq \sum_{\mu \in \sigma_{\cJ^{c}} } \abs{G_{ \boldlambda_{L}(\bold{u_{\cJ}})  } (E - \mu;\,\bolda_{\cJ},\,\boldb_{\cJ})}.\label{GJc}
\end{align}
\end{enumerate}
 \end{lemma}
 
 \begin{proof} (i) and (ii) follows from the definition of a PI box.  To prove (iii), given $\lambda \in  
\sigma_\cJ$ we  let $\Pi^{(\cJ)}_\lambda$ denote the orthogonal projection onto the corresponding eigenspace.  
  If $E \notin \sigma \left( \nboxLu   \right)  $, it follows from (ii) that
   \begin{align}
 G_{\nboxLu}  (E)= \sum_{\lambda \in \sigma_{\cJ}, \mu \in \sigma_{\cJ^{c}}  }  \pa{E-\lambda -\mu}^{-1}  \Pi^{(\cJ)}_\lambda  \otimes\Pi^{(\cJ^c)}_\mu,
 \end{align}
which implies \eq{GJ} and \eq{GJc}.
 \end{proof}

As a consequence, we get

\begin{lemma}\label{PIsuit}
 Let $\nboxlu = \boldlambda_{\ell}^{{\cJ}}(\bold{u}_{\cJ}) \times  \boldlambda_{\ell}^{{\cJ}^{c}}(\bold{u}_{\cJ^{c}} ) $ be a PI $n$-particle box and  $E \in \R$. If $\ell$ is sufficiently large, the following holds:
\begin{enumerate}

\item   Given $\theta>0$, suppose  $\boldlambda_{\ell}(\bold{u_{\cJ}})$ is $(\theta, \,E - \mu  )$-suitable  for every $\mu \in \sigma_{\cJ^c}$ and  $\boldlambda_{\ell}(\bold{u_{\cJ^{c}}})$ is $(\theta, \,E - \lambda  )$-suitable  for every $\lambda \in \sigma_{\cJ}$.  Then  $\nboxlu$ is $\left(\tfrac{\theta}{2}, \, E \right)$-suitable.

\item   Given $m>0$, suppose  $\boldlambda_{\ell}(\bold{u_{\cJ}})$ is $(m,\, E - \mu)$-regular  for every $\mu \in \sigma_{\cJ^c}$ and  $\boldlambda_{\ell}(\bold{u_{\cJ^{c}}})$ is $(m,\, E - \lambda)$-regular  for every $\lambda \in \sigma_{\cJ}$.  Then  $\nboxlu$ is $\left(m-\tfrac{100nd \log\ell}{\ell}, \, E \right)$-regular.

\item   Given $0 < \zeta^{\prime} < \zeta < 1$, suppose  $\boldlambda_{\ell}(\bold{u_{\cJ}})$ is $(\zeta,\, E - \mu )$-SES  for every $\mu \in \sigma_{\cJ^c}$ and  $\boldlambda_{\ell}(\bold{u_{\cJ^{c}}})$ is $(\zeta,\, E - \lambda)$-SES  for every $\lambda \in \sigma_{\cJ}$.  Then  $\nboxlu$ is $\left(\zeta^{\prime}, \, E \right)$-SES.

\end{enumerate}
\end{lemma}

\begin{proof} We prove (i), the proofs of (ii) and (ii) are similar.
Given $\bolda, \boldb \in \Nboxlu$ with $\norm{\bolda - \boldb} \geq \tfrac{\ell}{100},$ then either we have $\norm{\bolda_{\cJ} - \boldb_{\cJ}} \geq \tfrac{\ell}{100},$ or $\norm{\bolda_{\cJ^{c}} - \boldb_{\cJ^{c}}} \geq \tfrac{\ell}{100}.$ Without loss of generality, we suppose that $\norm{\bolda_{\cJ} - \boldb_{\cJ}} \geq \tfrac{\ell}{100}.$ Then
\begin{align}
\abs{ G_{\Nboxlu} (E;\,\bolda,\,\boldb) } &\leq \sum_{ \mu \in \sigma_{\cJ^c} } 
\abs{ G_{ \boldlambda_{\ell}(\bold{u_{\cJ}})} (E - \mu;\,\bolda_{\cJ},\,\boldb_{\cJ}) }		 \\
&\leq \abs{\boldlambda_{\ell}(\bold{u_{\cJ^{c}}})} \ell^{-\theta} \leq \ell^{nd-\theta} \leq \ell^{-\tfrac{\theta}{2}},
\end{align}
provided $\ell$ is sufficiently large.
\end{proof}

\begin{definition}  Let $\nboxLa$ and $\nboxLb$ be a pair of  $n$-particle boxes.  We say that 
$\nboxLa$ and $\nboxLb$ are $L$-distant if  $\max \bigl\{\dist \, ( \,  \boldb, \, \cS_{\bolda}^{n} ) , \,\, \dist \, ( \, \bolda, \, \cS_{\boldb}^{n})   \bigr\} \ge  2nL$.
\end{definition}

For fully interactive boxes we have the following lemma.
\begin{lemma} \label{part1prop2} Let   
 $\nboxLa$ and $\nboxLb$ be a pair of FI $n$-particle boxes, where $L \ge 2(n-1)r_0$.  Then
 $\nboxLa$ and $\nboxLb$ are fully separated if
\beq\label{FIsep}
\max_{x \in \cS_{ \bolda},y \in  \cS_{\boldb}} \norm{x-y} \ge 2nL .
\eeq
In particular,  $L$-distant FI $n$-particle boxes are fully separated. 
\end{lemma}

\begin{proof}   If  a box $\nboxLa$ is FI, we have  $\max_{x,x^\pr \in \cS_{ \bolda}} \norm{x-x^\pr}\le (n-1)(L+r_0)$.  Thus, if  $\nboxLa$ and $\nboxLb$ are FI, and $L \ge 2(n-1)r_0$, the condition \eq{FIsep} implies
\beq
\dist(\cS_{ \bolda}, \,\,\cS_{\boldb })=\min_{x \in \cS_{ \bolda},y \in  \cS_{\boldb}} \norm{x-y} \ge 2nL - 2(n-1)(L+r_0)> L,
\eeq
so $\nboxLa$ and $\nboxLb$ are fully separated. 

Since
\begin{align}
\max \bigl\{\dist \, ( \,  \boldx, \, \cS_{\boldy}^{n} ) , \,\, \dist \, ( \, \boldy, \, \cS_{\boldx}^{n})   \bigr\} \le \max_{x \in \cS_{ \bolda},y \in  \cS_{\boldb}} \norm{x-y} , 
\end{align} 
we conclude that  $L$-distant FI $n$-particle boxes are fully separated. 
\end{proof}

\subsection{Resonant boxes}
\begin{definition} Let $\boldlambda = \prod_{i = 1, \dots, n} \Lambda_{L_i} (a_i)$ be an $n$-particle box, $L = \min_{i =1,.., n}\set{L_i}$, and $E\in \R$.
\begin{enumerate}
\item Let $s > 0$. Then  $\boldlambda$ is called  $(E,s)$-suitably resonant if and only if $\dist \Bl( \sigma \bigl( H_{\boldlambda}^{(n)} \bigr) , E \Br) < L^{-s}$.  Otherwise, $\boldlambda$ is said to be $(E,s)$-suitably nonresonant.

\item Let $\beta \in (0, \, 1)$.  Then $\boldlambda$ is called $(E,\beta)$-resonant if and only if $\dist \Bl( \sigma \bigl( H_{\boldlambda}^{(n)} \bigr) , E \Br) < \tfrac{1}{2}  e^{-L^{\beta}} $. Otherwise, $\boldlambda$ is said to be $(E,\beta)$-nonresonant.
\end{enumerate}

\end{definition}

\subsection{Suitable Cover}

We now introduce the concept of a  suitable cover as in \cite[Definition~3.12]{GKber}, adapted to the discrete case.

\begin{definition}\label{defcov} Let $\nboxLx$ be an $n$-particle box, and $\ell < L$.
The \emph{suitable $\ell$-covering} of $\nboxLx$ is
the collection of $n$-particle boxes  
\begin{align}\label{standardcover}
{\mathcal C}_{L,\ell}^{(n)} \left(\boldx \right)= \{ \mathbf{\Lambda}_{\ell}^{(n)}(\bolda)\}_{\bolda \in  \Xi_{L,\ell}^{(n)}},
\end{align}
where
\beq  \label{bbG}
 \Xi_{L,\ell}^{(n)}:= \{ \boldx + \alpha\ell  \Z^{nd}\}\cap \widehat{\boldlambda}_L^{(n)}\quad 
\text{with}  \quad \alpha =\max \br{\tfrac {3} {5},\tfrac {4} {5}}   \cap \set{\tfrac {L-\ell}{2 \ell k}; \, k \in \N }.
\eeq
\end{definition}

We recall \cite[Lemma~3.13]{GKber}, which we rewite in our context.

\begin{lemma}\label{lemcovering} Let $\ell \le \frac  L   6$. Then for  every $n$-particle box  $\nboxLx$ the  suitable $\ell$-covering
  ${\mathcal C}_{L,\ell}^{(n)} \left(\boldx \right) $  satisfies
  \begin{align}\label{nestingproperty}
&\nboxLx=\bigcup_{\bolda \in  \Xi_{L,\ell}^{(n)}}  \mathbf{\Lambda}_{\ell}^{(n)}(\bolda),\\ \label{covproperty}
&\text{for $\boldb \in \nboxLx$ there is $\mathbf{\Lambda}_{\ell}^{(n,\boldb)} \in {\mathcal C}^{(n)}_{L,\ell} \left(\boldx \right)$ with} \; 
 \boldlambda_{\tfrac{\ell}{10}} ^{(n)}(\boldb) \, \cap \, \mathbf{\Lambda}_{L}^{(n)}(\boldx)  \subseteq \mathbf{\Lambda}_{\ell}^{(n,\boldb)}
,\\
\label{freeguarantee}
&\boldlambda^{(n)}_{\frac{\ell}{5}}(\bolda)\cap \boldlambda^{(n)}_{\ell}(\boldb)=\emptyset
\quad \text{for all} \quad \bolda, \boldb\in \boldx + \alpha\ell  \Z^{nd}, \  \bolda\ne \boldb,\\ \label{number}
&  \pa{\tfrac{L} {\ell}}^{nd}\le \# \Xi_{L,\ell}^{(n)}= \pa{ \tfrac{L-\ell} {\alpha \ell}+1}^{nd }\le   \pa{\tfrac{2L} {\ell}}^{nd}.
\end{align}
Moreover,  given $\bolda \in  \boldx + \alpha\ell  \Z^{nd}$ and $k \in \N$, it follows that
\beq \label{nesting}
\mathbf{\Lambda}^{(n)}_{(2  k \alpha  +1) \ell}(\bolda)= \bigcup_{\boldb \in  \{  \boldx + \alpha\ell  \Z^{nd}\}\cap \widehat{\mathbf{\Lambda}}^{(n)}_{(2k \alpha  +1) \ell}(\boldb) } \mathbf{\Lambda}^{(n)}_{\ell}(\boldb).
\eeq
\end{lemma}

Note that $\boldlambda_{\ell}^{(n,\boldb)}$  does not denote a box  centered at $\boldb$, just some box in ${\mathcal C}^{(n)}_{L,\ell} \left(\boldx \right)$ satisfying \eq{covproperty}.     By  $\boldlambda_{\ell}^{(n,\boldb)}$, or just $\boldlambda_{\ell}^{(\boldb)}$, we willl always mean such a box.

\begin{remark}  It suffices to require  $\alpha \in\br{\tfrac {3} {5},\tfrac {4} {5}}   \cap \set{\tfrac {L-\ell}{2 \ell k}; \, k \in \N }$ in Definition~\ref{defcov}.  We specified $\alpha=\alpha_{L,\ell}$ for convenience, so there is no ambiguity  in the definition of ${\mathcal C}_{L,\ell}^{(n)} \left(\boldx \right) $.
\end{remark}

\begin{lemma}\label{lembadregions}  
Let $\NboxLx$ be an $N$-particle box and   $\ell < \frac L 6$.   Define $\set{k_j}_{j\in \N} \subset \N$   by    
\begin{gather}\notag
k_1=6\sqtx{and}  k_j =\min \set{k \in \N\; |\;  k  > k_{j-1} + 6 + 2 \pa{N\alpha}^{-1}}\sqtx{for} j=2,3,\ldots,\\
\text{so}\quad k_j \le 6 + (j-1)\pa{7 + 2\pa{N\alpha}^{-1}} \le  j \pa{7 + 2\pa{N\alpha}^{-1}} ,\label{defki}
\end{gather}
and set
\beq\label{Kj}
K_j= 2k_j N \alpha +1 \le 17jN \qtx{for} j=1,2\ldots.
\eeq
Then, given   $\bolda_{s}\in \Xi_{L,\ell}^{(N)}(\boldx)$,  $s=1,2,\ldots, S$,
  we can find $\boldu_{t}\in \Xi_{L,\ell}^{(N)}(\boldx)$ with $t=1,2,\ldots, T \le SN^N$, and $j_t \in \set{1,2,\ldots, SN^N}$  such that
\beq\label{disjointboxes}
\widetilde{\Upsilon}= \widetilde{\Upsilon}_{L,\ell}^{(N)}\pa{\set{\bolda_s}_{s=1}^S}:=\bigcup_{t=1}^T  \boldlambda^{(N)}_{K_{j_t}\ell}(\boldu_{t})\subseteq \NboxLx,
\eeq
\beq\label{distboxes}
\dist\set{\boldlambda^{(N)}_{K_{j_t}\ell}(\boldu_{t}), \boldlambda^{(N)}_{K_{j_{t\pr}}\ell}(\boldu_{t^\pr})}> 1\qtx{for} t\ne {t^\pr},
\eeq
 \beq\label{distboxes2}
 \partial_+ \boldlambda^{(N)}_{K_{j_t}\ell}(\boldu_{t}) \cap  \widetilde\Upsilon =\emptyset \qtx{for} t=1,2,\ldots, T,
 \eeq
 \beq\label{sumKj}
 \sum_{t=1}^{T} K_{j_t}\le  17S N^{N+1}
 \eeq
and  for  $\boldy  \in  \NboxLx \setminus \widetilde\Upsilon $ and   $\boldlambda_{\ell}^{(N,\boldy)} \in {\mathcal C}_{L,\ell}^{(N)} \left(\boldx \right)$ as in   \eq{covproperty}, the boxes  $\boldlambda_{\ell}^{(N,\boldy)}$ and $\boldlambda_{\ell}(\bolda_s)$ are $\ell$-distant  for $s=1,2,\ldots, S$.
\end{lemma}

  \begin{proof}
  Given   $\bolda \in \Xi_{L,\ell}^{(n)}$, we set
\beq
\Upsilon_{L,\ell}^{(N)}(\boldx,\bolda) =  \bigcup_{\boldb \in \cS_{\bolda}^{N}} \boldlambda^{(N)}_{(4N+2)\ell}(\boldb)\cap \NboxLx,
\eeq
and note that  for  $\boldy  \in  \NboxLx \setminus \Upsilon_{L,\ell}^{(N)}(\boldx,\bolda) $ and   $\boldlambda_{\ell}^{(N,\boldy)} \in {\mathcal C}_{L,\ell}^{(N)} \left(\boldx \right)$ as in   \eq{covproperty}, the boxes  $\boldlambda_{\ell}^{(N,\boldy)}$ and $\boldlambda_{\ell}(\bolda)$ are $\ell$-distant.

Given  $\bolda_{s}\in \Xi_{L,\ell}^{(N)}(\boldx)$,  $s=1,2,\ldots, S$, 
 we set  $\Upsilon= \cup_{s=1}^S \Upsilon_{L,\ell}^{(N)}(\boldx,\bolda_s)$. 
 In view of \eq{nestingproperty} and \eq{nesting}, we can find $\boldb_{r}\in \Xi_{L,\ell}^{(N)}(\boldx)$,
$r=1,2,\ldots, R \le SN^N$, such that   (we use  $(12N\alpha+1)\ell > (4N+2)\ell + 2\alpha \ell$)
\beq\label{disjointboxes1}
\Upsilon \subseteq \bigcup_{r=1}^R  \boldlambda^{(N)}_{(12N\alpha+1)\ell }(\boldb_{r})\subseteq \NboxLx. 
\eeq

Let $\set{k_j}_{j\in \N} \subset \N$ be as in \eq{defki}, so in particular
\beq
\tfrac 1 2(2k_j N\alpha +1)\ell > \tfrac 1 2 (2k_{j-1}N \alpha +1)\ell +1 +   (2k_1N \alpha +1)\ell
\qtx{for} j=2,3,\ldots.
\eeq
By a geometrical argument we can find $\boldu_{t}\in \Xi_{L,\ell}^{(N)}(\boldx)$ and and $j_t \in \set{1,2,\ldots, SN^N}$, $t=1,2,\ldots, T \le SN^N$, such that
\beq
\bigcup_{r=1}^R  \boldlambda^{(N)}_{(12N\alpha+1)\ell }(\boldb_{r})\subseteq 
\widetilde{\Upsilon}=\bigcup_{t=1}^T  \boldlambda^{(N)}_{K_{j_t}\ell }(\boldu_{t})\subseteq \NboxLx,
\eeq
and \eq{distboxes} holds,  implying \eq{distboxes2}, and $\boldlambda^{(N)}_{K_{j_t}\ell }(\boldu_{t})$ contains at least $j_t$ of the boxes $\boldlambda^{(N)}_{(12N\alpha+1)\ell }(\boldb_{r})$, so  $\sum_{t=1}^{T} {j_t}\le SN^N$.  Thus, using \eq{Kj},
\beq
 \sum_{t=1}^{T} K_{j_t}\le 17 N \sum_{t=1}^{T} {j_t}\le  17S N^{N+1}.
 \eeq

In view of   \eq{disjointboxes1} and \eq{disjointboxes}, we conclude that
for  $\boldy  \in  \NboxLx \setminus \widetilde\Upsilon $ and   $\boldlambda_{\ell}^{(N,\boldy)} \in {\mathcal C}_{L,\ell}^{(N)} \left(\boldx \right)$ as in   \eq{covproperty}, the boxes  $\boldlambda_{\ell}^{(N,\boldy)}$ and $\boldlambda_{\ell}(\bolda_s)$ are $\ell$-distant  for $s=1,2,\ldots, S$.

\end{proof}

\section{The bootstrap multiscale analysis} \label{inductionhyp}

We will now prove Theorem~\ref{maintheorem}  by  induction on  the number of particles,  The one particle case     was proven by Germinet and Klein \cite{GK1}.  We  fix $N \ge 2$,  assume Theorem~\ref{maintheorem} holds for $n=1,2,\ldots, N-1$ particles, and prove the theorem  for $N$ particles. As in \cite{GK1}, the proof will be done by a bootstrapping argument, making successive use  of four multiscale analyses. 

\emph{In this section we assume that the following induction hypothesis; it follows from
 assuming that Theorem~\ref{maintheorem} holds for $n=1,2,\ldots, N-1$.}

 \begin{induction}
 For every  $\tau \in (0, \, 1)$ there is a length scale $L_{\tau} $,  $\delta_{\tau}>0$, and $m_\tau^{*}  > 0$, such that the following hold for all $E\in \R$ and $n = 1, 2, ..., N-1$:

i)   For all $L \geq L_{\tau}$ and $\bolda \in  \R^{nd}$ we have 
 \begin{align} \label{eq001}
\P \Bigl\{   \nboxLa \,\, \text{is  $ \left (m^{*}_\tau, \,E \right )$-nonregular} \Bigr\}    \leq e^{-L^{\tau}}.
\end{align}

ii) Let $I(E)=[E-\delta_{\tau}, E+\delta_{\tau}]$.  For all $L \geq L_{\tau}$  and all pairs  of partially separated  $n$-particle boxes $\nboxLa$ and $\nboxLb$ we have 
\begin{align}  \label{eq002}
\P \Bigl\{ \exists \, E^\pr \in I(E) \; \text{so both} \;  \nboxLa \text{ and}\:\nboxLb \; \text{are}\; \left ({m_\tau^*}, \,E^\pr \right )\text{-nonregular} \Bigr\} \leq e^{-L^{\tau}} .
\end{align}
 \end{induction}

For partially interactive $N$-particle boxes we can estimate probabilities directly from the induction hypothesis, without a multiscale analysis for $N$-particles.

\begin{lemma} \label{PINSl}
Let $\Nboxlu = \boldlambda_{\ell}(\bold{u}_{\cJ}) \times  \boldlambda_{\ell}(\bold{u}_{\cJ^{c}} )$ be a PI $N$-particle box and $\varsigma \in (0,1)$. Then for $\ell$ large and all $E\in \R$  we have 
\begin{align} \notag
&\P  \Bigl\{ \Nboxlu \;\text{is} \; \left( m^*_{\varsigma}(\ell),E \right)\text{-nonregular} \Bigr\} \leq \ell^{Nd} e^{-\ell^{\varsigma}}\text{with}\;  m^*_{\varsigma}(\ell)= m_{\varsigma}^*-\tfrac{100Nd \log\ell}{\ell},\\
&\P  \Bigl\{ \Nboxlu \,\, \text{is} \,\, \left( \theta,\,E \right)\text{-nonsuitable} \Bigr\} \leq \ell^{Nd} e^{-\ell^{\varsigma}} \text{for} \; \; \theta < \tfrac \ell {\log \ell} \tfrac{ m^*_{\varsigma}(\ell)}{100}, \\
&\P  \Bigl\{ \Nboxlu \,\, \text{is} \,\, \left( \varsigma,\,E \right)\text{-nonSES} \Bigr\} \leq \ell^{Nd} e^{-\ell^{\varsigma}}.  \notag
\end{align} 
\end{lemma}

\begin{proof} Let $E\in \R$, and $\ell$ large.  It follows from Lemma~\ref{PIsuit}(ii) and the induction hypothesis that
\begin{align} 
&\P  \Bigl\{ \Nboxlu \; \text{is} \; \left( m^*_\tau(\ell),\,E \right)\text{-nonregular} \Bigr\}  \\ \notag
&\hskip30pt \leq  \sum_{\mu \in \sigma_{\cJ^c}} \P  \Bigl\{  \boldlambda_{\ell}(\bold{u_{\cJ}})\text{ is }(m^*_\tau, E - \mu,)\text{-nonregular}    \Bigr\} \\ & \hskip60pt  + \sum_{\lambda \in \sigma_{\cJ}} \P  \Bigl\{  \boldlambda_{\ell}(\bold{u_{\cJ^{c}}}) \text{ is }(m^*_\tau, E - \lambda)\text{-nonregular}   \Bigr\} \le  \ell^{Nd}e^{-\ell^{\tau}}, \notag
\end{align}
where we used ${\ell^{\abs{\cJ}d }+ \ell^{\abs{\cJ^c}d }}\le  \ell^{Nd}$ for $\ell$ large.
The other estimates now follow from Remark~\ref{goodbox}.
\end{proof}

In what follows, we  fix $\zeta, \, \tau, \beta, \, \zeta_0,\, \zeta_1,\, \zeta_2,\, \gamma$ such that 
\begin{gather} \label{constant}
0 < \zeta < \tau < 1,\quad  \zeta \, \gamma^{2} < \zeta_2,\\
0 < \zeta < \zeta_{2} < \gamma \zeta_2 <  \zeta_{1} < \gamma \zeta_1 < \beta < \zeta_0 < r < \tau < 1 \quad \text{with } \zeta \, \gamma^{2} < \zeta_2. \notag
\end{gather}
$\tau$ will play the role of $\varsigma$ in Lemma \ref{PINSl}, $\beta$ will control our resonant boxes, and $\gamma$ will control the growth of our length scales. We will let $m^*$ denote the mass $m_{\tau}^{*}$ that we get from the induction hypothesis. 

\subsection{The first multiscale analysis}
\begin{proposition} \label{part1mainthm} 
 Let $\theta > 8Nd$  and $E \in \R $.  Take $0 < p < p + Nd < s < s + 2Nd -2< \theta,$  $Y \ge 4000 N^{N+1} $, and $p_0=p_{0}(N) < \tfrac 12\pa{2Y}^{-Nd}$. Then there exists a length scale $Z_{0}^{*}$ such that if for some $L_0 \geq Z_{0}^{*}$ we have
\beq
\sup_{\boldx \in \R^{Nd}} \P \Bigl\{ \mathbf{\Lambda}_{L_0}^{(N)}(\boldx) \, \,  \text{is}\,\,(\theta, \,E)  \text{-nonsuitable} \Bigr\} \leq p_0, 
\eeq
then, setting $L_{k+1} = Y\,L_k,$ for $k = 0, 1, 2, ...,$ there exists $K_0 \in \N$ such that for every $k \geq K_0$ we have 
\beq
\sup_{\boldx \in \R^{Nd}} \P \Bigl\{ \mathbf{\Lambda}_{L_k}^{(N)} (\boldx)    \,\,\text{is}\,\,(\theta, \,E)  \text{-nonsuitable} \Bigr\} \leq L_{k}^{-p}.
\eeq
\end{proposition}

To prove the proposition we use the following deterministic lemma.  

\begin{lemma} \label{part1prop1a} Let $\theta > 8Nd$  and $E \in \R $.  Take $ Nd < s < s + 2Nd -2< \theta$. Let  $J \in \N$,   $Y \ge 4000J N^{N+1} $,  $L =  Y\ell$, and $\boldx \in \R^{Nd}$.
 Suppose  we have the following:
\begin{enumerate}
\item $\NboxLx$ is $E$-suitably nonresonant.

\item  Every box  $\mathbf{\Lambda}_{K_j\ell}^{(N)}(\boldu)\subseteq \NboxLx$, with $\boldu \in \Xi_{L,\ell}(\boldx)$, $j=1,2,\ldots, J N^{N}$,  where $K_j$ is
 given  in \eq{Kj},  is $E$-suitably nonresonant.

\item There are at most $J $ pairwise $\ell$-distant,  $(E, \, \theta)$-nonsuitable boxes in the $\ell$-suitable cover.

\end{enumerate}
Then    the $N$-particle  box  $\NboxLx$ is $\left(E, \,\theta \right)$-suitable for $L$ sufficiently large.
\end{lemma}

\begin{proof}
Since there at most $J$ pairwise $\ell$-distant $N$-particle boxes in the suitable cover that are $(E,\,\theta)$-nonsuitable. we can find $\bolda_{s}\in \Xi_{L,\ell}^{(N)}(\boldx)$,  $s=1,2,\ldots, j \leq J$, such that the boxes 
 $\boldlambda_{\ell}(\bolda_{s})$ are pairwise 
 $\ell$-distant, $(E,\,\theta)$-nonsuitable boxes, and any box $\boldlambda_{\ell}(\bolda)$ with
 $\bolda \in \Xi_{L,\ell}^{(N)}(\boldx)$ which is $\ell$-distant from all the $\boldlambda_{\ell}(\bolda_{s})$  must be $(E,\,\theta)$-suitable. Applying Lemma~\ref{lembadregions}, we obtain
 $\widetilde{\Upsilon}= \widetilde{\Upsilon}_{L,\ell}^{(N)}\pa{\set{\bolda_s}_{s=1}^J}$ as in \eq{disjointboxes}, satisfying the conclusions of that lemma.  In particular, for  $\boldy  \in  \NboxLx \setminus \widetilde\Upsilon $, the boxes  $\boldlambda_{\ell}^{(N,\boldy)}$ and $\boldlambda_{\ell}(\bolda_s)$ are $\ell$-distant  for $s=1,2,\ldots, J$, and hence $\boldlambda_{\ell}^{(N,\boldy)}$ is a $(E,\,\theta)$-suitable box.

Let $\bolda \in \NboxLx$.   Then, either $\boldlambda_{\ell}^{(\bolda)}$ is $(\theta,E)$-suitable or $\boldlambda_{\ell}^{(\bolda)}$ is $(\theta,E)$-nonsuitable. We  proceed as follows: 

\begin{enumerate}
\item If   $\boldlambda_{\ell}^{(\bolda)}$ is $(\theta,E)$-suitable, and  $ \boldb \in \NboxLx\setminus \boldlambda_{\ell}^{(\bolda)}$,  we use 
the resolvent identity to get 
\begin{align}\notag
&\abs{G_{\boldlambda_L^{(N)}(\x)}(\bolda, \boldb; E)}  \leq s_{Nd}\,\,\ell^{Nd-1-\theta}\,\abs{G_{\boldlambda_L^{(N)}(\x)}(\boldv^{\prime}, \boldb; E)}\sqtx{for some} \boldv^{\prime}\in \partial_+ \boldlambda_{\ell}^{(\bolda)},\\
&\text{so}\quad  \tfrac \ell{10} \le \norm{\boldv^{\prime}-\bolda} \le \ell + 1.
\label{asuitable}
\end{align}

\item If $\boldlambda_{\ell}^{(\bolda)}$ is  $(\theta,E)-$nonsuitable, we must have  $\bolda \in  \widetilde\Upsilon$, and hence $\bolda \in \boldlambda^{(N)}_{K_{j_t}\ell}(\boldu_{t}) $ for some $t$.  Let $ \boldb \in \NboxLx\setminus \boldlambda^{(N)}_{(K_{j_t}+2)\ell}(\boldu_{t})  $.  Applying  the resolvent identity, and using the fact that $\mathbf{\Lambda}(t):=\boldlambda^{(N)}_{K_{j_t}\ell}(\boldu_{t}) $ is $E$-suitably nonresonant by hypothesis, we get
\begin{align}\label{suitablebad}
\abs{G_{\boldlambda_L^{(N)}(\x)}(\bolda, \boldb; E)} & \leq \left|\partial \boldlambda(t)  \right| \, \left[\max_{(\boldu, \boldv) \in \partial\boldlambda(t)} 
\abs{G_{\boldlambda(t)}(\bolda, \boldu; E)\,G_{\boldlambda_L^{(N)}(\x)}(\boldv, \boldb; E)}\right]\\
& \leq \abs{\partial \boldlambda(t) } \, \pa{K_{j_t}\ell}^s \abs{G_{\boldlambda_L^{(N)}(\x)}(\boldv^{\prime}, \boldb; E)}\sqtx{for some} \boldv^{\prime}\in \partial_+\boldlambda(t),  \notag
\end{align}
where $\boldlambda_{\ell}^{(\boldv^{\prime})}$ is $(\theta,E)$-suitable.   We use \eq{asuitable} with $\bolda=\boldv^{\prime}$, getting
\begin{align}\notag
\abs{G_{\boldlambda_L^{(N)}(\x)}(\bolda, \boldb; E)} & \leq s^2_{Nd}L^{Nd-1+s}\ell^{Nd-1-\theta}
\abs{G_{\boldlambda_L^{(N)}(\x)}(\boldv^{\prime\prime}, \boldb; E)} \\
& <  \label{asuitableb}
\abs{G_{\boldlambda_L^{(N)}(\x)}(\boldv^{\prime\prime}, \boldb; E)} \sqtx{for some} \boldv^{\prime\prime}\in \partial_+\boldlambda_\ell^{(\boldv^{\prime})},\\
&\text{so} \quad   \norm{\boldv^{\prime\prime}-\bolda} \le (\ell + 1) + (K_{j_t}\ell + 1)\le 2K_{j_t}\ell,
\notag
\end{align}
if we can guarantee $s_{Nd}^{2} L^{Nd+s-1}\,\,\ell^{Nd-1-\theta} < 1$.    Since $L = Y\,\ell$, we need $$s_{Nd}^{2} {(Y\,\ell)}^{Nd+s-1}\,\,\ell^{Nd-1-\theta} = s_{Nd}^{2}  Y^{Nd+s-1} \ell^{2Nd+s-2-\theta} \leq 1,$$ which is certainly true by our choice of $s$ and $\theta$ provided that we take $\ell$ large enough.  
\end{enumerate}

Given $\bolda, \,\boldb \in \NboxLx$ with $\norm{\bolda - \boldb} \geq \tfrac{L}{100}$, we estimate $\abs{G_{\boldlambda_L^{(N)}(\x)}(\bolda, \boldb; E)}$ by, when possible, repeatedly using either \eq{asuitable} or \eq{asuitableb}, as appropriate,  and, when we must stop because we got too close to $\boldb$, using the hypothesis that  $\boldlambda_L^{(N)}(\x)$ is $E$-suitably nonresonant, obtaining
\begin{equation}\label{finalres}
\abs{G_{\boldlambda_L^{(N)}(\x)}(\bolda, \boldb; E)} \leq \bigl(s_{Nd}\,\,\ell^{Nd-1-\theta}\bigr)^{N(Y)}\,\,L^s,
\end{equation}
where $N(Y)$ is the  number of times we used \eq{asuitable}. We can always use either \eq{asuitable} or \eq{asuitableb},  unless we got to some    $\boldv$ where $\boldlambda_{\ell}^{(\boldv)}$ is $(E, \theta)$-suitable and
 $\boldb \in \boldlambda_{\ell}^{(\boldv)}$, or  $\boldv \in \boldlambda^{(N)}_{K_{j_t}\ell}(\boldu_{t})  $ for some $t$ and  $\boldb 
\in \boldlambda^{(N)}_{K_{j_t}\ell}(\boldu_{t}) $.   It follows that we will not have to stop before
\beq \label{NYest11}
N(Y) (\ell +1) + \sum_{t=1}^{T} 2K_{j_t}\ell+ (\ell +1) \ge \norm{\boldb-\bolda}\ell\ge \tfrac{L}
{100} .
\eeq
Thus, using  \eq{sumKj},  we can achieve 
\begin{align}\label{NYest}
N(Y) & \ge \pa{\tfrac{Y}{100} -  34JN^{N+1} }\tfrac \ell{\ell +1}-2.
\end{align}

We take $Y \ge 4000J N^{N+1} $, which guarantees
$N(Y)\ge 2$ for large $\ell$ by \eq{NYest}. It then follows from
\eq{finalres} that for $\bolda, \,\boldb \in \NboxLx$ with $\norm{\bolda - \boldb} \geq \tfrac{L}{100}$ we have 
\begin{align}
\abs{G_{\boldlambda_L^{(N)}(\x)}(\bolda, \boldb; E)} \le L^{-\theta},
\end{align}
and we conclude that $\NboxLx$ is $\left(E, \,\theta \right)$-suitable for $L$ sufficiently large.
\end{proof}

\begin{proof}[Proof of Proposition~\ref{part1mainthm}]
Given a scale $L$, we set
\beq
p_L= \sup_{\boldx \in \R^{Nd}} \P \Bigl\{ \NboxLx    \,\,is\,\,(\theta, \,E) \text{-nonsuitable} \Bigr\} .
\eeq
We assume $L=Y \ell$,  with  $ \ell$ is sufficiently large when necessary.

 Let $\NboxLx$ be an $N$-particle box with an $\ell$-suitable cover, 
${\mathcal C}_{L,\ell} (\boldx)$, where $L = Y \ell$.  Let $J\in 2\N$, to be specified later.  We define several events:  $\cE=  \Bl\{\NboxLx \; \text{ is } (\theta, E)\text{-nonsuitable}\Br\}$, $\cA$ is the event that at least one of the PI boxes in  ${\mathcal C}_{L,\ell} (\boldx)$ is  $(\theta, E)$-nonsuitable, 
$\cW_J$ is the event that (i) and (ii) in Lemma~\ref {part1prop1a} hold, and $\cF_J$ is the event that (iii) in Lemma~\ref {part1prop1a} holds.   It follows from Lemma~\ref {part1prop1a}  that,
taking  $Y \ge 4000J^2 N^{2N+1} $,
\beq\label{PcE}
\P\set{\cE} \le \P\set{\cW_J^c}  + \P\set{\cF_{J}^c} \le  \P\set{\cW_J^c}  + \P\set{\cF_{J}^c\cap \cA^c}+\P\set{\cA}.
\eeq
Lemma~\ref{PINSl} yields
\beq \label{PcE1}
\P(\cA) \leq (2Y)^{Nd}\,\ell^{Nd}e^{-\ell^{\tau}} \leq \tfrac{1}{4} L^{-p}.
\eeq
Since $s >Nd+p$, Theorem \ref{Wegner0} implies
\begin{align}\label{PcE2}
\P(\cW_J^c) & \leq {2N\, \norm{\rho}_\infty} \pa{L^{Nd-s} + (2Y)^{Nd}\sum_{j=1}^{JN^N}
(K_j\ell)^{Nd-s}}\\
& \le {2N\, \norm{\rho}_\infty}  \pa{1 +(2Y)^{s} JN^N  } L^{Nd-s}
\notag   \\
&   \le {2N\, \norm{\rho}_\infty}  \pa{1+ (2Y)^{s} JN^N  }L^{-p} \le \tfrac{1}{4} L^{-p}.
\notag
\end{align}
To estimate $\P\set{\cF_{J}^c\cap \cA^c}$, note that if $\bom \in \cF_{J}^c\cap \cA^c$, then there exist  $J+1$ FI pairwise $\ell$-distant boxes in the suitable cover that are $(\theta, E)$-nonsuitable.   By Lemma~\ref{part1prop2} these boxes are fully separated.   Thus
\beq\label{PcE3}
\P \set{\cF_{J}^c\cap \cA^c} \leq   (2Y)^{(J+1)Nd} p_\ell^{J+1}.
\eeq
Since $\boldx \in \R^{Nd}$ is arbitrary, we conclude that 
 \beq \label{part1keyeq}
 p_L \leq \tfrac{1}{2}L^{-p} +  (2Y)^{(J+1)Nd} p_\ell^{J+1}.
  \eeq
  
 Since $J \ge  1$, it follows immediately from \eq{part1keyeq} that $p_\ell \le \ell^{-p}$ implies $p_L \le  L^{-p}$. 

 We now fix  $L_0$, set $L_k = Y L_{k-1}$ for $k=1,2,\ldots$, and let $p_k=p_{L_k}$.  To finish the proof,  we need to  show  that, if $p_0< \tfrac 12\pa{2Y}^{-Nd}$,  for sufficiently large $L_0$,
 we have
 \beq
 K_0=\inf  \set{k=0,1,\ldots\; | \; p_k \le L_{k}^{-p}} < \infty.
 \eeq

It follows from  equation \eqref{part1keyeq} that 
  \begin{align}\label{part1mainthmeq1}
  p_{k+1}\le \tfrac 1 2L_{k+1}^{-p} +  \pa{\pa{2Y}^{Nd} p_k}^{J+1} \qtx{for} k=1,2\ldots..  
    \end{align}
If $k+1<K_0$, we conclude that
  \beq
   p_{k+1}< 2 \pa{\pa{2Y}^{Nd} p_k}^{J+1}
  \eeq

 If $K_0 = 0$ or $K_0=1$  we are done.  If not, we have $p_1 \le 2 \pa{\pa{2Y}^{Nd} p_0}^{J+1}$.  If $K_0 >2$, we have
 \beq\label{proc1}
 p_2 \le 2 \pa{\pa{2Y}^{Nd} p_1}^{J+1}\le  2 \pa{\pa{2Y}^{Nd}2 \pa{\pa{2Y}^{Nd} p_0}^{J+1}}^{J+1}=
 \pa{2\pa{2Y}^{Nd}}^{1+J} p_0^{(J+1)^2}.
 \eeq
 Repeating this procedure, if $k < K_0$ we obtain
 \beq\label{proc12}
\pa{Y^k L_0}^{-p} = L_k^{-p}< p_k \le 
  \pa{2\pa{2Y}^{Nd}}^{\tfrac {(J+1)^k-1}{(J+1)-1}} p_0^{(J+1)^k}.
 \eeq
 We now choose $J=1$, obtaining
 \beq
2\pa{2Y}^{Nd}\pa{Y^k L_0}^{-p} < 
  \pa{2\pa{2Y}^{Nd} p_0}^{2^k}.
 \eeq
 We conclude that $K_0 < \infty$, since by hypothesis \  $2\pa{2Y}^{Nd} p_0 <1$.
  \end{proof}

\subsection{The second multiscale analysis}
\begin{proposition} \label{part2mainthm}
 Let $E \in \R$,  $p>0$, $0< m_0 < m^{*}$, $1< \gamma <1 + \frac p{p+2Nd}$. Then  there exists a length scale $Z_{1}^{*}$ such that if for some $L_0 \geq Z_{1}^{*}$ we can verify
\beq
\sup_{\boldx \in \R^{Nd}} \P \Bigl\{ \mathbf{\boldlambda}_{L_0}^{(N)}(\boldx) \, \,  \text{is}\,\,(m_0, \,E) \text{-nonregular} \Bigr\} \leq L_{0}^{-p}, 
\eeq
then, setting $L_{k+1} = L_k^{\gamma},$ for $k = 1, 2, ...,$ we get
\beq
\sup_{\boldx \in \R^{Nd}} \P \Bigl\{ \mathbf{\boldlambda}_{L_k}^{(N)} (\boldx)    \,\,\text{is}\,\,\left(\tfrac{m_0}{2}, \,E \right) \text{-nonregular} \Bigr\} \leq L_{k}^{-p}\qtx{for all} k = 0, 1, 2, ...
\eeq
\end{proposition}

To prove the proposition we use the following deterministic lemma.  

\begin{lemma} \label{part2prop1a}  
Let $E \in \R$, $L = \ell^{\gamma}$,  $J \in \N$, $m_0>0$, and 
\beq\label{mellcond}
m_\ell \in  [\tfrac{1}{\ell^{\kappa}},m_0], 
\qtx{where} 0< \kappa <\min \set { \gamma -1, \gamma (1-\beta), 1} . 
\eeq 
Suppose that we have the following:

\begin{enumerate}
\item $\NboxLx$ is $E$-nonresonant.

\item  Every box  $\mathbf{\Lambda}_{K_j\ell}^{(N)}(\boldu)\subseteq \NboxLx$, with $\boldu \in \Xi_{L,\ell}(\boldx)$, $j=1,2,\ldots, J N^{N}$,  where $K_j$ is
 given  in \eq{Kj},  is $E$-nonresonant.

\item There are at most J pairwise $\ell$-distant, $(E, \, m_{\ell})$-nonregular boxes in the suitable cover.  
\end{enumerate}
Then $\NboxLx$ is $\left(E, \,m_L \right)$-regular for $L$ large, where 
\beq \label{mlmL}
 m_{\ell} \geq m_L \geq m_{\ell} -  \tfrac{1}{2\ell^{\kappa}}\ge  \tfrac{1}{L^{\kappa}}.
\eeq
\end{lemma}

\begin{proof}
Since there at most $J$ pairwise $\ell$-distant $N$-particle boxes in the suitable cover that are $(E,\,m_\ell)$-nonregular, we can find $\bolda_{s}\in \Xi_{L,\ell}^{(N)}(\boldx)$,  $s=1,2,\ldots, j \leq J$, such that the boxes 
 $\boldlambda_{\ell}(\bolda_{s})$ are pairwise 
 $\ell$-distant, $(E,\,m_\ell)$-nonregular boxes, and any box $\boldlambda_{\ell}(\bolda)$ with
 $\bolda \in \Xi_{L,\ell}^{(N)}(\boldx)$ which is $\ell$-distant from all the $\boldlambda_{\ell}(\bolda_{s})$  must be $(E,\,m_\ell)$-regular.  Proceeding as in the proof of Lemma~\ref{part1prop1a}, we obtain $\widetilde{\Upsilon}= \widetilde{\Upsilon}_{L,\ell}^{(N)}\pa{\set{\bolda_s}_{s=1}^J}$ as in \eq{disjointboxes}, satisfying the conclusions of Lemma~\ref{lembadregions}.  In particular, for  $\boldy  \in  \NboxLx \setminus \widetilde\Upsilon $, the boxes  $\boldlambda_{\ell}^{(N,\boldy)}$ and $\boldlambda_{\ell}(\bolda_s)$ are $\ell$-distant  for $s=1,2,\ldots, J$, and hence $\boldlambda_{\ell}^{(N,\boldy)}$ is a $(E,\,\theta)$-regular box.

Let $\bolda \in \NboxLx$.   Then, either $\boldlambda_{\ell}^{(\bolda)}$ is $(E,\,m_\ell)$-regular or $\boldlambda_{\ell}^{(\bolda)}$ is $(E,\,m_\ell)$-nonregular. We  proceed as follows: 

\begin{enumerate}
\item If   $\boldlambda_{\ell}^{(\bolda)}$ is $(E,\,m_\ell)$-regular, and  $ \boldb \in \NboxLx$ with $\boldb \notin \boldlambda_{\ell}^{(\bolda)}$, then 
we use the resolvent identity as in \eq{asuitable} to get 
\begin{align}\label{aregular}
&\abs{G_{\NboxLx}(\bolda, \boldb; E) }  \leq \abs{\partial \boldlambda_{\ell} ^{(\bolda)} } \, \left[\max_{(\boldu, \boldv) \in \partial \boldlambda_{\ell} ^{(\bolda)}} \abs{G_{\boldlambda_{\ell} ^{(\bolda)}}(\bolda, \boldu; E)\,G_{\NboxLx}(\boldv, \boldb; E)}         \right]  \notag\\
&\qquad \leq s_{Nd}\,\,\ell^{Nd-1}\,\,\,\,e^{-m_{\ell}\,\norm{\bolda - \boldb_1^\pr} }\,\,\abs{G_{\NboxLx}(\boldb_1, \boldb; E)}  \qtx{for some} (\boldb_1^\pr,\boldb_1)\in \partial\boldlambda_{\ell}^{(\bolda)} \notag\\
&\qquad \leq s_{Nd}\,\,\ell^{Nd-1}\,\,\,\,e^{-m_{\ell}\,\left( \norm{\bolda - \boldb_1} -1 \right)}\,\,\abs{G_{\NboxLx}(\boldb_1, \boldb; E)}  \qtx{for some} \boldb_1\in \partial_+ \boldlambda_{\ell}^{(\bolda)} \notag\\
&  \qquad \leq e^{-m_\ell^\pr \,\norm{\bolda - \boldb_1}}\,\abs{G_{\NboxLx}(\boldb_1, \boldb; E)} \qtx{for some} \boldb_1\in \partial_+ \boldlambda_{\ell}^{(\bolda)}.
\end{align}
Since  $ \tfrac \ell{10} \le \norm{\boldb_1-\bolda} \le \ell + 1$, and we assumed \eq{mellcond}, this holds with
\beq
m_\ell^\pr = (1 - \tfrac {10} \ell)m_\ell  - \tfrac {10} \ell \log \pa{s_{Nd}\ell^{Nd-1}}\ge 
  m_\ell - C_1(d,N,m_0)\tfrac {\log \ell}{\ell} >0.
\eeq

\item If $\boldlambda_{\ell}^{(\bolda)}$ is  $(E,\,m_\ell)$-nonregular, we must have  $\bolda \in  \widetilde\Upsilon$, and hence $\bolda \in \boldlambda^{(N)}_{K_{j_t}\ell}(\boldu_{t}) $ for some $t$. Let $ \boldb \in \NboxLx\setminus \boldlambda^{(N)}_{(K_{j_t}+2)\ell}(\boldu_{t})  $. Proceeding as in \eq{suitablebad}-\eq{asuitableb}, and using  \eq{aregular}, we get
\begin{align}\notag
\abs{G_{\boldlambda_L^{(N)}(\x)}(\bolda, \boldb; E)} & \leq s_{Nd}\pa{K_{j_t}\ell}^{Nd-1}e^{\pa{K_{j_t}\ell}^{\beta}}e^{-m_\ell^\pr \frac \ell{10}}
\abs{G_{\boldlambda_L^{(N)}(\x)}(\boldv^{\prime\prime}, \boldb; E)} \\
& < \label{aregularb}
\abs{G_{\boldlambda_L^{(N)}(\x)}(\boldv^{\prime\prime}, \boldb; E)} ,
\end{align}
for some $\boldv^{\prime\prime}\in \NboxLx$ with $\norm{\boldv^{\prime\prime}-\bolda} \le (\ell + 1) + (K_{j_t}\ell + 1)\le 2 K_{j_t}\ell$, since we have 
\beq
s_{Nd}\pa{K_{j_t}\ell}^{Nd-1}e^{\pa{K_{j_t}\ell}^{\beta}}e^{-m_\ell^\pr \frac \ell{10}}\le
s_{Nd}\pa{17JN^N\ell}^{Nd-1}e^{\pa{17JN^N\ell}^{\beta}}e^{-m_\ell^\pr \frac \ell{10}}
<1,
\eeq
by our choice of  $m_\ell$, provided that we take $\ell$ large enough.  

\end{enumerate}

Given $\bolda, \,\boldb \in \NboxLx$ with $\norm{\bolda - \boldb} \geq \tfrac{L}{100}$, we estimate $\abs{G_{\boldlambda_L^{(N)}(\x)}(\bolda, \boldb; E)}$, using repeatedly  \eq{aregular} and  \eq{aregularb}, as appropriate, and, when we must stop because we got too close to $\boldb$, using the hypothesis that  $\boldlambda_L^{(N)}(\x)$ is $E$-NR.   Similarly to \cite[Proof of Lemma~3.11]{GKber}, we can find  $\boldv_1,\boldv_2.\ldots,\boldv_R \in  \NboxLx$, such that
\beq
\sum_{r=1}^{R-1}\norm{\boldv_r - \boldv_{r+1}} + \sum_{t=1}^{T} 2K_{j_t}\ell+ (\ell +1) \ge \norm{\boldb-\bolda} ,
\eeq
so
\beq
\sum_{r=1}^{R-1}\norm{\boldv_r - \boldv_{r+1}}\ge\norm{\boldb-\bolda} -36 JN^{N+1}\ell.
\eeq
and we have 
\begin{align}\label{finalreges}
&\abs{G_{\boldlambda_L(\x)}(\bolda, \boldb; E)} \le \prod_{r=1}^{R-1} e^{-m_\ell^\pr \,\norm{\boldv_r - \boldv_{r+1}}}\abs{G_{\boldlambda_L^{(N)}(\x)}(\boldv_R, \boldb; E)}\\
& \quad \notag \leq   e^{-m_\ell^\pr \sum_{r=1}^{R-1}\norm{\boldv_r - \boldv_{r+1}}}e^{L^\beta}
 \le  e^{-m_\ell^\pr \,\norm{\bolda - \boldb}+m_\ell^\pr \, 36JN^{N+1}\ell+ \ell^{\gamma \beta} } \le e^{-m_\L\,\norm{\bolda - \boldb}},
\end{align}
where, using \eq{mellcond},
\begin{align}
m_L &= m_\ell^{\pr}-\tfrac {100}L \pa{36JN^{N+1}m_\ell^{\pr} \ell + \ell^{\gamma \beta} }\\
&\ge m_\ell - \tfrac {C_1(d,N,m_0) \log \ell}{\ell} -  \tfrac {C_2(d,N,m_0,J)} {\ell^{\gamma-1}}  - \tfrac {100}{\ell^{\gamma(1-\beta)}}
 \ge m_\ell - \tfrac 1 {2\ell^\kappa}\ge  \tfrac 1 {2\ell^\kappa}\ge \tfrac 1 {\L^\kappa}.\notag
\end{align}

We proved that  $\NboxLx$ is $\left(E, \,m_L \right)$-regular for $L$ large, with $m_L$ as in \eq{mlmL}.
\end{proof}

\begin{proof}[Proof of Proposition~\ref{part2mainthm}]
Given a scale $L$ and $m_L>0$, we set
\beq
p_L(m_L)= \sup_{\boldx \in \R^{Nd}} \P \Bigl\{ \NboxLx    \,\,is\,\,(m_L, \,E) \text{-nonregular} \Bigr\} .
\eeq

We start by showing that there exists  $\widetilde{Z_1}$ such that if  $p_\ell(m_\ell) \le \ell^{-p}$, where $m_\ell$ satisfies \eq{mellcond},    and $\ell \geq \widetilde{Z_1}$, then, setting $L = \ell^{\gamma}$, we have
 $p_L(m_L) \le  L^{-p}$ with $m_L$ as in \eq{mlmL}.
 
 Let $\NboxLx$ be an $N$-particle box and
 $J\in \N$.  We define several events:  $\cE=  \Bl\{\NboxLx \; \text{ is } (m_L, E)\text{-nonregular}\Br\}$, $\cA$ is the event that at least one of the PI boxes in  ${\mathcal C}_{L,\ell} (\boldx)$ is  $(m_L, \,E)$\text{-nonregular}, 
$\cW_J$ is the event that (i) and (ii) in Lemma~\ref {part2prop1a} hold, and $\cF_J$ is the event that (iii) in Lemma~\ref {part2prop1a} holds.   It follows from Lemma~\ref {part2prop1a}  that
\beq
\P\set{\cE} \le \P\set{\cW_J^c}  + \P\set{\cF_{J}^c} \le  \P\set{\cW_J^c}  + \P\set{\cF_{J}^c\cap \cA^c}+\P\set{\cA}.
\eeq

Lemma~\ref{PINSl} yields (large $\ell$, so $m_\ell \le m^*_{\tau}(\ell)$ in Lemma~\ref{PINSl})
\beq \label{PcE13}
\P(\cA) \leq \pa{\tfrac {2L}\ell}^{Nd}\,\ell^{Nd}e^{-\ell^{\tau}} \leq \tfrac{1}{4} L^{-p}.
\eeq
Theorem \ref{Wegner0} implies
\begin{align}\label{PcE23}
\P(\cW_J^c) & \leq {N\, \norm{\rho}_\infty} \pa{L^{Nd}e^{-L^\beta} + \pa{\tfrac {2L}\ell}^{Nd}\sum_{j=1}^{JN^N}
(K_j\ell)^{Nd}e^{-(K_j\ell)^{\beta}}}\\
& \le {N\, \norm{\rho}_\infty}  \pa{1 +(2\ell^{\gamma-1})^{Nd} JN^N  } e^{-\frac 1 2\ell^\beta}\le \tfrac{1}{4} L^{-p}.
\end{align}
To estimate $\P\set{\cF_{J}^c\cap \cA^c}$, note that if $\bom \in \cF_{J}^c\cap \cA^c$, then there exist  $J+1$  FI pairwise $\ell$-distant boxes in the suitable cover that are $(\theta, E)$-nonregular.   By Lemma~\ref{part1prop2} these boxes are fully separated.   Thus
\beq\label{PcE33}
\P \set{\cF_{J}^c\cap \cA^c} \leq  \pa{\tfrac {2L}\ell}^{(J+1)Nd} \pa{p_\ell(m_\ell)}^{J+1}\le \tfrac{1}{2}L^{-p} + (2\ell^{\gamma-1})^{(J+1)Nd} \ell^{-(J+1)p}.
\eeq
We now take $J=1$, require $1< \gamma <1 + \frac p{p+2Nd}$ and conclude that, 
since $\boldx \in \R^{Nd}$ is arbitrary, 
 \beq \label{part1keyeq33}
 p_L (m_L)\leq \tfrac{1}{2}L^{-p} + (2\ell^{\gamma-1})^{2Nd} \ell^{-2p}\le \tfrac{1}{2}L^{-p} + \tfrac{1}{2}L^{-p} \le L^{-p} .
  \eeq

   We now fix  $L_0$ and $m_0>0$. We take $L_0$ is sufficiently large, so $m_0 \ge \L_0^{-\kappa}$  We  set $L_k =  L_{k-1}^\gamma$ and $m_k=m_{k-1}- \frac 1 {2L_{k-1}^\kappa}$ for $k=1,2,\ldots$, and let $p_k=p_{L_k}(m_k)$. If    $p_0\le    L_0^{-p}$, we conclude that 
 $p_k\le  L_k^{-p}$ for $k=1,2,\ldots$.  Moreover,
 \beq \label{halfmass}
 m_0- m_k \le \sum_{j=1}^\infty (m_{j-1}-m_{j}) = \tfrac 12\sum_{j=1}^\infty L_{j-1}^{-\kappa}
 =\tfrac 12\sum_{j=1}^\infty L_{0}^{-\kappa\gamma^{j-1}}\le \tfrac {m_0}2,
 \eeq
 so $m_k \ge  \tfrac {m_0}2$ for $k=1,2,\ldots$. 
 \end{proof}

\subsection{The third multiscale analysis}

\begin{proposition}  \label{part3mainthm}
Let $0 < \zeta_1 < \zeta_0 < 1$ as in  \eq{constant}, $E\in\R,$ and  assume   $Y\ge \pa{3800N^{N+1}}^{\frac 1 {1-\zeta_0}}$. Then there exists $Z_2^{*} > L_{\tau}$ such that, if for some scale $L_0 > Z_2^{*}$ we have 
\beq   
\sup_{\x \in \R^{Nd}}\P \Bl\{ \boldlambda_{L_0}^{(N)}(\x) \text{ is } (\zeta_0, E)\text{-nonSES} \Br\}\le  \pa{2\pa{2Y}^{Nd}}^{-\tfrac {1}{Y^{\zeta_0} -1}}	,
\eeq
then, setting $L_{k+1}= Y\,L_k$, $k=0, 1, 2, ...,$ there exists $K_1 \in \mathbb{N}$ such that for every $k \geq K_1$ we have  
\beq 
\sup_{\x \in \R^{Nd}}\P \Bl\{\boldlambda _{L_k}(\boldx) \text{ is } (\zeta_0, E)\text{-nonSES}\Br\} \le  e^{-L_k^{\zeta_1}}.
\eeq
As a consequence, for every $k \geq K_1$, we have 
\beq 
\sup_{\x \in \R^{Nd}} \P \set{\boldlambda_{L_k}(\boldx) \text{ is } \left(L_k^{\zeta_0-1}, E\right)\text{-nonregular}} \le  e^{-L_k^{\zeta_1}}.
\eeq
\end{proposition}

To prove the proposition, we use the following deterministic lemma.

\begin{lemma} \label{part3prop1a}
Let  $L = Y \ell$, where   $Y\ge \pa{3800N^{N+1}}^{\frac 1 {1-\zeta_0}}$, and set $J = \lfloor  Y^{\zeta_0}\rfloor$, the largest integer $\le  Y^{\zeta_0}$.
 Suppose the following are true:
\begin{enumerate}
\item $\NboxLx$ is $E$-nonresonant.

\item  Every box  $\mathbf{\Lambda}_{K_j\ell}^{(N)}(\boldu)\subseteq \NboxLx$, with $\boldu \in \Xi_{L,\ell}(\boldx)$, $j=1,2,\ldots, J N^{N}$,  where $K_j$ is
 given  in \eq{Kj},  is $E$-nonresonant.

\item There are at most J pairwise $\ell$-distant, $(E,\,\zeta_0)$-nonSES boxes in the suitable cover.
\end{enumerate}

Then $\NboxLx$ is $(E,\,\zeta_0)$-SES, provided $\ell$ is sufficiently large.
\end{lemma}

\begin{proof}
Since there at most $J$ pairwise $\ell$-distant $N$-particle boxes in the suitable cover that are $(E,\,\zeta_0)$-nonSES, we can find $\bolda_{s}\in \Xi_{L,\ell}^{(N)}(\boldx)$,  $s=1,2,\ldots, j \leq J$, such that the boxes 
 $\boldlambda_{\ell}(\bolda_{s})$ are pairwise 
 $\ell$-distant, $(E,\,\zeta_0)$-nonSES boxes, and any box $\boldlambda_{\ell}(\bolda)$ with
 $\bolda \in \Xi_{L,\ell}^{(N)}(\boldx)$ which is $\ell$-distant from all the $\boldlambda_{\ell}(\bolda_{s})$  must be $(E,\,\zeta_0)$-SES. Applying Lemma~\ref{lembadregions}, we obtain
 $\widetilde{\Upsilon}= \widetilde{\Upsilon}_{L,\ell}^{(N)}\pa{\set{\bolda_s}_{s=1}^J}$ as in \eq{disjointboxes}, satisfying the conclusions of that lemma.  In particular, for  $\boldy  \in  \NboxLx \setminus \widetilde\Upsilon $, the boxes  $\boldlambda_{\ell}^{(N,\boldy)}$ and $\boldlambda_{\ell}(\bolda_s)$ are $\ell$-distant  for $s=1,2,\ldots, J$, and hence $\boldlambda_{\ell}^{(N,\boldy)}$ is a $(E,\,\zeta_0)$-SES box.

Let $\bolda \in \NboxLx$.   Then, either $\boldlambda_{\ell}^{(\bolda)}$ is $(E,\,\zeta_0)$-SES or $\boldlambda_{\ell}^{(\bolda)}$ is $(E,\,\zeta_0)$-nonSES. We  proceed as follows: 

\begin{enumerate}
\item If   $\boldlambda_{\ell}^{(\bolda)}$ is $(E,\,\zeta_0)$-SES, and  $ \boldb \in \NboxLx\setminus \boldlambda_{\ell}^{(\bolda)}$,  we use 
the resolvent identity to get 
\begin{align}\notag
&\abs{G_{\boldlambda_L^{(N)}(\x)}(\bolda, \boldb; E)}  \leq s_{Nd}\,\,\ell^{Nd-1}\,e^{-\ell^{\zeta_0}}\abs{G_{\boldlambda_L^{(N)}(\x)}(\boldv^{\prime}, \boldb; E)}\sqtx{for some} \boldv^{\prime}\in \partial_+ \boldlambda_{\ell}^{(\bolda)},\\
&\qquad \qquad\text{so}\quad  \tfrac \ell{10} \le \norm{\boldv^{\prime}-\bolda} \le \ell + 1.
\label{ases}
\end{align}

\item If $\boldlambda_{\ell}^{(\bolda)}$ is  $(E,\,\zeta_0)$-nonSES, we must have  $\bolda \in  \widetilde\Upsilon$, and hence $\bolda \in \boldlambda^{(N)}_{K_{j_t}\ell}(\boldu_{t}) $ for some $t$.  Let $ \boldb \in \NboxLx\setminus \boldlambda^{(N)}_{(K_{j_t}+2)\ell}(\boldu_{t})  $.  Applying  the resolvent identity, and using the fact that $\mathbf{\Lambda}(t):=\boldlambda^{(N)}_{K_{j_t}\ell}(\boldu_{t}) $ is $E$-nonresonant by hypothesis, we get
\begin{align}\notag
&\abs{G_{\boldlambda_L^{(N)}(\x)}(\bolda, \boldb; E)} \leq \left|\partial \boldlambda(t)  \right| \, \left[\max_{(\boldu, \boldv) \in \partial\boldlambda(t)} 
\abs{G_{\boldlambda(t)}(\bolda, \boldu; E)\,G_{\boldlambda_L^{(N)}(\x)}(\boldv, \boldb; E)}\right] \\    \label{sesbad}
& \quad \leq s_{Nd}\pa{K_{j_t}\ell}^{Nd-1}e^{\pa{K_{j_t}\ell}^{\beta}}\abs{G_{\boldlambda_L^{(N)}(\x)}(\boldv^{\prime}, \boldb; E)}\sqtx{for some} \boldv^{\prime}\in \partial_+\boldlambda(t),  \end{align}
where $\boldlambda_{\ell}^{(\boldv^{\prime})}$ is $(E,\,\zeta_0)$-SES.   We use \eq{ases} with $\bolda=\boldv^{\prime}$, getting
\begin{align}\notag
\abs{G_{\boldlambda_L^{(N)}(\x)}(\bolda, \boldb; E)} & \leq s^2_{Nd}L^{2(Nd-1)} e^{\pa{17JN^{N+1}\ell}^\beta}  e^{-\ell^{\zeta_0}}          \abs{G_{\boldlambda_L^{(N)}(\x)}(\boldv^{\prime\prime}, \boldb; E)} \\
& \le  \label{asesb}
\abs{G_{\boldlambda_L^{(N)}(\x)}(\boldv^{\prime\prime}, \boldb; E)} \sqtx{for some} \boldv^{\prime\prime}\in \partial_+\boldlambda_\ell^{(\boldv^{\prime})},\\
&\text{so} \quad   \norm{\boldv^{\prime\prime}-\bolda} \le (\ell + 1) + (K_{j_t}\ell + 1)< 18JN^{N+1}\ell,
\notag
\end{align}
if we can guarantee $s^2_{Nd}L^{2(Nd-1)} e^{\pa{17JN^{N+1}\ell}^\beta}  e^{-\ell^{\zeta_0}}  < 1$.    Since $L = Y\,\ell$, we need $$s^2_{Nd}Y^{2\pa{Nd-1}} \ell^{2\pa{Nd-1}}e^{\pa{17JN^{N+1}\ell}^\beta}   e^{-\ell^{\zeta_0}} \leq 1,$$ which is certainly true since $\beta < \zeta_0$, provided that we take $\ell$ large enough.  
\end{enumerate}

Given $\bolda, \,\boldb \in \NboxLx$ with $\norm{\bolda - \boldb} \geq \tfrac{L}{100}$, we estimate $\abs{G_{\boldlambda_L^{(N)}(\x)}(\bolda, \boldb; E)}$ by, when possible, repeatedly using either \eq{ases} or \eq{asesb}, as appropriate,  and, when we must stop because we got too close to $\boldb$, using the hypothesis that  $\boldlambda_L^{(N)}(\x)$ is $E$-NR, obtaining
\begin{equation}\label{finalses}
\abs{G_{\boldlambda_L^{(N)}(\x)}(\bolda, \boldb; E)} \leq \pa{s_{Nd}\,\,\ell^{Nd-1}\,e^{-\ell^{\zeta_0}}}^{N(Y)}\,\,e^{L^{\beta}},
\end{equation}
where $N(Y)$ is the  number of times we used \eq{ases}. We can always use either \eq{ases} or \eq{asesb},  unless we got to some    $\boldv$ where $\boldlambda_{\ell}^{(\boldv)}$ is $(E, \zeta_0)$-SES and
 $\boldb \in \boldlambda_{\ell}^{(\boldv)}$, or  $\boldv \in \boldlambda^{(N)}_{K_{j_t}\ell}(\boldu_{t})  $ for some $t$ and  $\boldb 
\in \boldlambda^{(N)}_{K_{j_t}\ell}(\boldu_{t}) $.  As in the proof of Lemma~\ref{part1prop1a}, we have \eq{NYest11} and \eq{NYest}.

If we have  
\beq\label{Yzeta}
N(Y)\ge 2Y^{\zeta_0},
\eeq  
 it follows from
\eq{finalses} that for $\bolda, \,\boldb \in \NboxLx$ with $\norm{\bolda - \boldb} \geq \tfrac{L}{100}$, and $L$ sufficiently large,  we have 
\begin{align}
\abs{G_{\boldlambda_L^{(N)}(\x)}(\bolda, \boldb; E)} \le e^{-L^{\zeta_0}},
\end{align}
and we conclude that $\NboxLx$ is $\left(E, \,\zeta_0 \right)$-SES.

To finish the proof we need to show that we can guarantee  \eq{Yzeta} for large $\ell$.
 It follows from \eq{NYest} that it suffices to have
$Y\ge 200\pa{18JN^{N+1}  +Y^{\zeta_0} }$.  
We fix $J = \lfloor  Y^{\zeta_0}\rfloor$, the largest integer $\le  Y^{\zeta_0}$, so 
it suffices to require $Y\ge \pa{3800N^{N+1}}^{\frac 1 {1-\zeta_0}}$ to get \eq{Yzeta} .
\end{proof}

\begin{proof}[Proof of Proposition~\ref{part3mainthm}]
Given a scale $L$, we set
\beq
p_L= \sup_{\boldx \in \R^{Nd}} \P \Bigl\{ \NboxLx    \,\,is\,\,(\zeta_0, \,E) \text{-nonSES} \Bigr\} .
\eeq
We assume $L=Y \ell$,  with  $ \ell$ is sufficiently large when necessary.

We proceed as in the proof of Proposition~\ref{part1mainthm}. Let $\NboxLx$ be an $N$-particle box with an $\ell$-suitable cover, 
${\mathcal C}_{L,\ell} (\boldx)$, where $L = Y \ell$.  Assume   $Y\ge \pa{3800N^{N+1}}^{\frac 1 {1-\zeta_0}}$, and set $J = \lfloor  Y^{\zeta_0}\rfloor$ We define  events
$\cE$, $\cA$, $\cW_J$, $\cF_J$ as in the proof of Proposition~\ref{part1mainthm}, with 
$(\zeta_0, \,E)$-nonSES boxes instead of  $(\theta, E)$-nonsuitable boxes, etc.  It follows from Lemma~\ref{part3prop1a} that \eq{PcE} holds. Using   Lemma~\ref{PINSl} we get 
\beq \label{PcE1ses}
\P(\cA) \leq (2Y)^{Nd}\,\ell^{Nd}e^{-\ell^{\tau}} \leq \tfrac{1}{4} e^{-L^{\zeta_1}}.
\eeq
Proceeding as in \eq{PcE23}, with our choice of $\beta$  Theorem \ref{Wegner0} implies
\begin{align}\label{PcE23ses}
\P(\cW_J^c) 
& \le {N\, \norm{\rho}_\infty}  \pa{1 +(2\ell^{\gamma-1})^{Nd} JN^N  } e^{-\frac 1 2\ell^\beta}\le \tfrac{1}{4} e^{-L^{\zeta_1}}.
\end{align}
We also have \eq{PcE3}, so, similarly to \eq{part1keyeq}, we get
 \beq \label{part1keyeqses}
 p_L \leq \tfrac{1}{2} e^{-L^{\zeta_1}} + (2Y)^{(J+1)Nd} p_\ell^{J+1}.
  \eeq

Since  $J+1=    \lfloor  Y^{\zeta_0}\rfloor + 1> Y^{\zeta_0}>Y^{\zeta_1}$, it follows immediately from \eq{part1keyeqses} that  $p_\ell \le e^{-\ell^{\zeta_1}}$ implies $p_L \le  e^{-L^{\zeta_1}}$.

We now fix  $L_0$, set $L_k = Y L_{k-1}$ for $k=1,2,\ldots$, and let $p_k=p_{L_k}$.  To finish the proof,  we need to  show  that, if $p_0< \tfrac 12\pa{2Y}^{-Nd}$,  for sufficiently large $L_0$,
 we have
 \beq
 K_0=\inf  \set{k=0,1,\ldots\; | \; p_k \le e^{-L^{\zeta_1}}} < \infty.
 \eeq

It follows from  equation \eqref{part1keyeqses} that 
  \begin{align}\label{part1mainthmeq10}
  p_{k+1}\le \tfrac{1}{2} e^{-L_{k+1}     ^{\zeta_1}} +  \pa{\pa{2Y}^{Nd} p_k}^{J+1} \qtx{for} k=1,2\ldots..  
    \end{align}
If $k+1<K_0$, we conclude that
  \beq \label{proc122}
   p_{k+1}< 2 \pa{\pa{2Y}^{Nd} p_k}^{J+1}.
  \eeq
   If $K_0 = 0$ or $K_0=1$  we are done.  If not,  we have $p_1 \le 2 \pa{\pa{2Y}^{Nd} p_0}^{J+1}$.   If $K_0 >2$ and $k < K_0$, proceeding as in  \eq{proc1}-\eq{proc122} we get
 \beq\label{proc1222}
 e^{-\pa{Y^k  L_0}^{\zeta_1}}=e^{-L_{k}^{\zeta_1}}< p_k \le 
  \pa{2\pa{2Y}^{Nd}}^{\tfrac {(J+1)^k-1}{(J+1)-1}} p_0^{(J+1)^k}.
 \eeq
 Since $Y^{\zeta_0} -1< J=    \lfloor  Y^{\zeta_0}\rfloor\le  Y^{\zeta_0}  $,  
and
 we assume $\pa{2\pa{2Y}^{Nd}}^{\tfrac {1}{Y^{\zeta_0} -1}} p_0<1$,
  we get
 \beq
 e^{-Y^{k\zeta_1}  L_0^{\zeta_1}}\le \pa{ \pa{2\pa{2Y}^{Nd}}^{\tfrac {1}{Y^{\zeta_0} -1}} p_0}^{ Y^{k\zeta_0}}.
 \eeq
Since $\zeta_0 > \zeta_1$ and $\pa{2\pa{2Y}^{Nd}}^{\tfrac {1}{Y^{\zeta_0} -1}} p_0<1$, we conclude that $K_0 < \infty$.
\end{proof}

\subsection{The fourth multiscale analysis}

We fix $\zeta, \, \tau, \beta, \, \zeta_1,\, \zeta_2,\, \gamma$ as in \eq{constant}.

\subsubsection{The single energy multiscale analysis}
\begin{proposition} \label{part4mainthm1}
Let   $0<m_0 <m^* = m_{\tau}$. 
 Then  there exists a length scale $Z_{3}^{*}$ such that, given an energy  $E \in \R$,  if for some $L_0 \geq Z_{3}^{*}$ we can verify
\begin{align} 
\sup_{\bolda \in \R^{Nd}}\P \Bigl\{   \mathbf{\Lambda}_{L_0}^{(N)}(\bolda)  \sqtx{is}  \left ({m_0}, \,E \right )\text{-nonregular} \Bigr\}    \leq e^{-L_0^{\zeta_2}},
\end{align}
then for sufficiently large $L$ we have
\begin{align} 
\sup_{\bolda \in \R^{Nd}}\P \Bigl\{  \mathbf{\Lambda}_{L}^{(N)}(\bolda)  \sqtx{is}  \left (\tfrac {m_0}2, \,E \right )\text{-nonregular} \Bigr\}    \leq e^{-L^{\zeta_2}}.
\end{align}
\end{proposition}

Proposition~\ref{part4mainthm1} is proved first  for a sequence of length scale $L_k$ similarly  to   Proposition~\ref{part2mainthm};  to obtain the sub-exponential decay of probabilities we choose $J$, the number of bad boxes, dependent on the scale $L$ as in the proof of Proposition~\ref{part4mainthm} below. To obtain Proposition~\ref{part4mainthm1} as stated, that is, for all sufficiently large scales, we prove a slightly more general result.

\begin{definition}
Let $E \in \R$.   An N-particle box, $\NboxLx$, is said to be $(E, m_{L})$-good if and only if 
it is $(E, m_{L})$-regular and $E$-nonresonant.
\end{definition}

\begin{lemma} \label{GKlemma} Let $\NboxLx$ be an N-particle box, $\gamma > 1$,  $\ell=L^{\frac 1 {\gamma^\pr}} $ with $\gamma \le \gamma^\pr \le \gamma^2$, and $m>0$.
Suppose every box in ${\mathcal C}_{L,\ell}^{(N)}(\x)$  is $(E, m)$-good. Then $\NboxLx$ is $(E, \frac m 2)$-good.          
\end{lemma}

This lemma is  a straightforward adaptation of \cite[Lemma~3.16]{GKber} to the discrete case.

\begin{lemma} \label{part1thm}
Let $E_1 \in \R$, $\zeta_2 \in (\zeta, \, \tau)$, and $\gamma \in (1, \, \tfrac{1}{{\zeta_2}})$  with $\zeta \, \gamma^{2} < \zeta_2  $.   Assume there exists a mass $m_{\zeta_2}>0$ and  a length scale $L_0 = L_{0}(\zeta_2)$, such that,  taking $L_{k+1} = L_{k}^{\gamma}$ for $k=0,1,\ldots$,  we have 
\begin{align} \label{appeneq1}
\sup_{\bolda\in \R^{Nd}} P \Bigl\{ \mathbf{\Lambda}_{L_{k}}^{(N)}(\bolda)\sqtx{is not} \left ({m_{\zeta_2}}, \,E_1 \right)\text{-good}  \Bigr\}  \leq e^{-L_{k}^{\zeta_2}} \qtx{for} k=0,1,\ldots.\end{align}
Then there exists $L_{\zeta}$ such that for every $L \geq L_{\zeta}$ we have 
\begin{align}
\sup_{\bolda\in \R^{Nd}} P \Bigl\{ \mathbf{\Lambda}_{L}^{(N)}(\bolda)\sqtx{is not} \left ({m_{\zeta_2}}, \,E_1 \right)\text{-good}  \Bigr\}  \leq e^{-L^{\zeta}}.
\end{align}
\end{lemma}

\begin{proof}Given a scale $L$ we take $K$ such that $ L_K \leq L < L_{K+1}$, and set $\ell = L_{K - 1}$.  Note that $L_K=\ell^\gamma$ and  $L_{K+1} = L_{K}^{\gamma} = L_{K-1}^{\gamma^{2}} = \ell^{\gamma^2}$, so $L=\ell^{\gamma^\pr}$ with $\gamma \le \gamma^\pr <\gamma^2$.

 Given an N-particle box  $\NboxLx$, let
 \beq
\cF_1 = \bigcup_{\boldu \in \suitcx} \cR_{\boldu},\qtx{where} \cR_{\boldu} = \Bigl\{ \Nboxlu \,\,is\,\,not\,\, \left ({m_{\zeta_2}}, \,E_1 \right )\text{-good} \Bigr\} .
\eeq
If  $\bom \notin \cF_1$,  every box in ${\mathcal C}_{L,\ell}^{(N)}(\x)$ is  $(m_{\zeta_2}, \, E_1)$-good, and hence $\NboxLx$ is $\left( \tfrac{m_{\zeta_2}}{2}, \, E_1 \right)$-good by Lemma \ref{GKlemma}. The lemma follows since
\beq
\P \left( \cF_1 \right) \leq \pa{\tfrac{2L}{\ell}}^{Nd} e^{-\ell^{\zeta_2}} \leq e^{- L^{\zeta}}. 
\eeq
\end{proof}

\subsubsection{The energy interval multiscale analysis}

\begin{lemma} \label{part4lem0}
Let $\NboxLx$ be an N-particle box and  $m > 0$. Let  $E_0 \in \R$, and suppose that

\begin{enumerate}
\item $\NboxLx$ is $(m,E_0)$-regular,
\item $ \dist \Bl( \sigma\left (H_{\NboxLx}  \right) , E_0 \Br) \geq e^{-L^{\beta}}$, i.e.,  $\norm{G_{\NboxLx}(E_0)} \leq e^{L^{\beta}}$.
\end{enumerate}
Then $\NboxLx$ is $\left(m - \frac {100\log 2} L,E  \right)$-good for every 
$E \in I = \left( E_0 - \eta, E_0 + \eta \right)$, where 
\beq\label{defeta3}
\eta = \tfrac{1}{2} e^{-mL- 2L^\beta}.
\eeq
\end{lemma}

\begin{proof} Let $\abs{E - E_0} \leq \tfrac{1}{2}\,e^{-L^{\beta}}$, so assumption (ii) implies\beq
\dist \Bl( \sigma\left (H_{\NboxLx}  \right) , E \Br) \geq \tfrac{1}{2}e^{-L^{\beta}}, \qtx{i.e.,}\norm{G_{\NboxLx}(E)} \leq 2\,e^{L^{\beta}}.
\eeq 
The resolvent equation gives 
\beq
G_{\NboxLx}(E) = G_{\NboxLx}(E_0) + (E - E_0)\, G_{\NboxLx} (E)\, G_{\NboxLx}(E_0),
\eeq  
so 
for all  $\bolda, \boldb \in \NboxLx$  we have
\begin{align}\notag
\abs{G_{\NboxLx} \left( E; \bolda, \boldb   \right)} 
& \leq e^{-m \norm{\bolda - \boldb}} + \abs{E - E_0} \norm{G_{\NboxLx} (E)} \norm{G_{\NboxLx}(E_0) } \\
& \leq e^{-m \norm{\bolda - \boldb}} + 2 \,e^{2\,L^{\beta}} \abs{E - E_0} .
\end{align} 

Now let $E \in I = \left( E_0 - \eta, E_0 + \eta \right)$, where $\eta$ is as in \eq{defeta3}. Since
\beq
\eta< \tfrac{1}{2}\,e^{-L^{\beta}} \qtx{and} \eta 2 \,e^{2\,L^{\beta}}= e^{-mL}\le  e^{-m \norm{\bolda - \boldb}}
\eeq
 we conclude that  if  $\norm{\bolda - \boldb} \geq \tfrac{L}{100}$ we have 
 \beq
 \abs{G_{\NboxLx} \left( E; \bolda, \boldb   \right)}  \le 2e^{-m \norm{\bolda - \boldb}}\le 
 e^{-\pa{m - \frac {100\log 2} L}\norm{\bolda - \boldb}}.
 \eeq
\end{proof}

Proposition~\ref{part3mainthm}, combined with Theorem~\ref{Wegner0}  and Lemma \ref{part4lem0},    yields the following proposition.

\begin{proposition} \label{bridgethm}
Let $0 < \zeta_2 < \zeta_1 < \zeta_0 < 1$, and assume the conclusions of Proposition~\ref{part3mainthm}.   There exists scales $L_k$, $k=1,2,\ldots$, such that $\lim_{k\to \infty} L_k=\infty$, with the following property: Let 
\beq
m_k=\left(L_k^{\zeta_0-1} - \tfrac {100\log 2} {L_k}  \right) \qtx{and}
\eta_k= \tfrac{1}{2} e^{- L_k^{\zeta_0}- 2L_k^\beta}.
\eeq
 Then for all  $E_0 \in \R$ we have
 \begin{align}
\sup_{\x \in \R^{Nd}} \P \set{\exists E\in \left( E_0 - \eta_k, E_0 + \eta_k \right) \sqtx{such that} \boldlambda_{L_k}(\boldx) \sqtx{is} \left(m_k, E\right)\text{-nonregular}} \le  e^{-L_k^{\zeta_1}},\end{align}
and
\begin{align}\notag
&\sup_{\x \in \R^{Nd}} \P \set{\exists E\in \left( E_0 - \eta_k, E_0 + \eta_k \right) \sqtx{such that} \boldlambda_{L_k}(\boldx) \text{ is not } \left(m_k, E\right)\text{-good}} \\
& \hskip50pt \le  e^{-L_k^{\zeta_1}} + 2N \norm{\rho}_\infty L_k^{Nd}e^{-L_k^\beta} \le e^{-L_k^{\zeta_2}}.
\end{align}
\end{proposition}

 We now  take $L=\ell^\gamma$.

\begin{definition}
Let $\NboxLx = \boldlambda_{L}(\bold{x}_{\cJ}) \times  \boldlambda_{L}(\bold{x}_{\cJ^{c}} )$ be a PI N-particle box with the usual $\ell$ suitable cover, and consider an energy $E \in \R$. Then:

\begin{enumerate}
\item $\NboxLx$ is not $E$-Lregular \emph(for left regular\emph) if and only if there are two partially separable boxes in $\cC_{L, \ell}^{\cJ}(\boldx_{\cJ})$
that are $(m^{*}, \, E-\mu)$-nonregular  
for some $\mu \in \sigma \left( H_{ \boldlambda_{L}^{{\cJ}^{c}}(\bold{x}_{\cJ^{c}} )  }  \right)$. 

\item $\NboxLx$ is not $E$-Rregular \emph(for right regular\emph) if and only if there are two partially separable boxes in $\cC_{L, \ell}^{\cJ^c}(\boldx_{\cJ^c})$   that are $(m^{*}, \, E-\lambda)$-nonregular for some $\lambda \in \sigma \left( H_{ \boldlambda_{L}(\bold{x}_{\cJ} )  }  \right)$.

\item $\NboxLx$ is  $E$-preregular if and only if  $\NboxLx$ is  $E$-Lregular and $E$-Rregular. 
\end{enumerate}
\end{definition}

\begin{lemma} \label{prereg}
Let $E_0 \in \R$, $I = [E_0-\delta_\tau, \, E_0+\delta_\tau]$,  and consider a PI N-particle box $\NboxLx = \boldlambda_{L}(\bold{x}_{\cJ}) \times  \boldlambda_{L}(\bold{x}_{\cJ^{c}} )$.
Then 
\begin{enumerate}
\item $\P\set{\NboxLx \text{ is not }E\text{-Lregular for some } E \in I} \leq L^{3Nd}  e^{-\ell^{\tau}}$, 
\item $\set{\NboxLx \text{ is not }E\text{-Lregular for some } E \in I}) \leq L^{3Nd} e^{-\ell^{\tau}}$.

\end{enumerate}
We conclude that 
\item  \beq
\P \bigl\{ \NboxLx \text{ is not }E\text{-preregular for some } E \in I  \bigr\} \leq 2L^{3Nd}  e^{-\ell^{\tau}} .
\eeq

\end{lemma}
\begin{proof}
We prove (i), the proof of (ii) is similar. Let us set $\cS = \Pi_{\cJ} \NboxLx$ and
$\cB = \Bl\{ \exists \, E \in I \sqtx{such that}\NboxLx \,\,\text{is not } E\text{-Lregular} \Br\}$.
Since $\NboxLx$ is PI, we have $\P (\cB) = \E_{\cS^{c}} \P_{\cS}(\cB)$. 

Let us fix  $\bom_{\cS ^{c}}$, and pick $\mu \in \sigma \left( H_{ \bom_{\cS ^{c}},\boldlambda_{L}^{{\cJ}^{c}}(\bold{x}_{\cJ^{c}} )  }  \right)$. Let  $\cD$ denote  the event that there exists $E \in I$ such that $\boldlambda_{L}(\x_{\cJ})$ contains two partially separable boxes in the $\ell$-suitable cover that are $(E - \mu, m^*)\text{-nonregular}$. We can rewrite $\cD$ as the event that there exists $E^{\prime} \in I - \mu$ such that $\boldlambda_{L}(\x_{\cJ})$ contains a pair of  partially separable boxes in the $\ell$-suitable cover that is $(E^{\prime}, m^*)\text{-nonregular}$, where $I - \mu = \Bl \{ \cE - \mu \,\,| \,\, \cE \in I   \Br\}.$ Applying the bootstrap MSA result to the interval $I - \mu $ for $\abs{\cJ}$ particles (induction hypothesis), we get
$\P(\cD) \leq \left( \tfrac{2L}{\ell} \right) ^{2\abs{\cJ}d} e^{-\ell^{\tau}}$.  We conclude that
\beq  
\P(\cB) \leq \abs{\boldlambda_{L}(\bolda_{\cJ^{c}})} \left( \tfrac{2L}{\ell} \right) ^{2\abs{\cJ}d} e^{-\ell^{\tau}}   \leq  L^{3Nd} e^{-\ell^{\tau}}.
\eeq
\end{proof}

\begin{definition} \label{CNR}
Let $\NboxLu = \boldlambda_{L}(\bold{u}_{\cJ}) \times  \boldlambda_{L}(\bold{u}_{\cJ^{c}} ) $ be a PI N-particle box, and consider an energy $E \in \R$. Then:

\begin{enumerate}
\item $\NboxLu$ is $E$-left nonresonant \emph(or LNR\emph) if and only if for every box $\boldlambda_{K_j \ell}(\bolda) \subseteq \boldlambda_{L}(\boldu_{\cJ})$, with $\bolda \in \suitc (\boldu_{\cJ})$ and $j \in \set{1, \, 2, \ldots \abs{\cJ}^{\abs{\cJ}}}$, is $(E-\mu)$-nonresonant for every 
$\mu \in \sigma \pa{H_{\boldlambda_{L}(\boldu_{\cJ^c})}}$. Otherwise we say $\NboxLu$ is $E$-left resonant \emph(or LR\emph).

\item $\NboxLu$ is $E$-right nonresonant \emph(or RNR\emph) if and only if for every box $ \boldlambda_{K_j \ell}(\bolda) \subseteq \boldlambda_{L}(\boldu_{\cJ^{c}})$ with $\bolda \in \suitc (\boldu_{\cJ^c})$ and $j \in \set{1, \, 2, \ldots \abs{\cJ^c}^{\abs{\cJ^c}}}$ is $(E-\lambda)$-nonresonant for every 
$\lambda \in \sigma \pa{H_{\boldlambda_{L}(\boldu_{\cJ})}}$. Otherwise we say $\NboxLu$ is $E$-right resonant \emph(or RR\emph).

\item We say $\NboxLu$ is
$E$-highly nonresonant \emph(or HNR\emph) if and only if $\NboxLu$ is $E$-nonresonant, 
$E$-LNR, and $E$-RNR.
\end{enumerate} 
\end{definition}

\begin{lemma}\label{part2firstthm}
Let $E \in \R,$ and 
$\NboxLu = \boldlambda_{L}(\bold{u}_{\cJ}) \times  \boldlambda_{L}(\bold{u}_{\cJ^{c}} )$ 
be a PI N-particle box. Assume that the following are true:
\begin{enumerate}
\item $\NboxLu$ is $E$-HNR.
\item $\NboxLu$ is $E$-preregular.
\end{enumerate}
Then $\NboxLu$ is $\left(m(L), \, E \right)$-regular, where 
 \beq   \label{mL}
m(L) \geq m^{*} - c_1(\ell^{1 - \gamma}) - c_2(\ell^ {1 - \beta})  - c_3 \left( \tfrac{\log L}{L} \right),
 \eeq
with the constants $c_1, \, c_2, \, c_3$ do not depend on the scale $\ell$ and $m^* = m_{\tau}$.
\end{lemma}

\begin{lemma}
Let $E \in \R,$ and $\NboxLu = \boldlambda_{L}(\bold{u}_{\cJ}) \times  \boldlambda_{L}(\bold{u}_{\cJ^{c}} ) $ be a PI N-particle box.
\begin{enumerate}
\item If $\NboxLu$ is  E-right resonant, then there exists an N-particle box
\beq   
\mathbf{\boldlambda} = \boldlambda_{L}(\boldu_{\cJ}) \times  \boldlambda_{K_{j}\ell} (\boldx) ,
\eeq
where $j \in \set{1, 2, \ldots, \abs{\cJ^c}^{\abs{\cJ^c}} } $,  $\boldx \in \suitc (\boldu_{\cJ^c})$, and $  \boldlambda_{K_{j}\ell} (\boldx)   \subseteq \boldlambda_{L}(\boldu_{\cJ^{c}}) $,  such that 
\beq 
\dist \bigl( \sigma \left(  \mathbf{\boldlambda} \bigr) , E    \right) < \tfrac{1}{2} e^{-\pa{K_{j}\ell}^{\beta}} \leq \tfrac{1}{2} e^{-\ell^{\beta}}.
\eeq

\item If $\NboxLu$ is E-left resonant, then there exists an N-particle box
\beq   
\mathbf{\boldlambda} = \boldlambda_{K_{j}\ell}   \times \boldlambda_{L}(\boldu_{\cJ^{c}}) ,
\eeq
where  $j \in \set{1, 2, \ldots, \abs{\cJ}^{\abs{\cJ}} }$,  $\boldx \in \suitc (\boldu_{\cJ})$, and $  \boldlambda_{K_{j}\ell}(\boldx)    \subseteq \boldlambda_{L}(\boldu_{\cJ}) $,   such that 
\beq 
\dist \bigl( \sigma \left(  \mathbf{\boldlambda} \bigr) , E    \right) < \tfrac{1}{2} e^{-\pa{K_{j}\ell}^{\beta}} \leq \tfrac{1}{2} e^{-\ell^{\beta}}.
\eeq
\end{enumerate}
\end{lemma} 

\begin{proof}
Let $E \in \R$ and $\NboxLu = \boldlambda_{L}(\bold{u}_{\cJ}) \times  \boldlambda_{L}(\bold{u}_{\cJ^{c}} ) $ be a PI N-particle box.  Assume that $\NboxLu$ is E-right resonant. Then by definition, we can find $\lambda \in \sigma \left( H_{\boldlambda_{L}(\boldu_{\cJ})}  \right)$ and an $ \bigl ( N- \abs{\cJ} \bigr)$-particle box, $\boldlambda_{K_{j}\ell} (\boldx) \subseteq \boldlambda_{L}(\boldu_{\cJ^{c}})$, with $\x \in \Xi_{L,\ell} \left( \boldu_{\cJ^{c}} \right)$ and  $j \in \set{1, 2, \ldots, \abs{\cJ}^{\abs{\cJ}} }$, such that $\boldlambda_{K_{j}\ell} (\boldx)$ is $(E-\lambda)$-resonant, i.e., 
$\dist \pa{ \sigma \pa{ H_{ \boldlambda_{K_{j}\ell} (\boldx) }}  , E - \lambda    }< \tfrac{1}{2} e^{-\pa{K_{j}\ell}^{\beta}}$.
Thus there exists $\eta \in \sigma \bigl( H_{ \boldlambda_{K_{j}\ell} (\boldx) } \bigr)$ such that 
\beq  
\abs{ E - \lambda - \eta }< \tfrac{1}{2} e^{-\pa{K_{j}\ell}^{\beta}}.
\eeq

Moreover, 
$ \boldlambda_{L}(\bold{u}_{\cJ}) \times  \boldlambda_{L}(\bold{u}_{\cJ^{c}} )$ is PI and  $  \boldlambda_{K_{j}\ell} \subseteq \boldlambda_{L}(\boldu_{\cJ^{c}})$, so if we take $\mathbf{\boldlambda} = \boldlambda_{L}(\boldu_{\cJ}) \times  \boldlambda_{K_{j}\ell} (\boldx)$ we get
$H_{\mathbf{\boldlambda}}^{(N)} = H_{\boldlambda_{L}(\boldu_{\cJ})} \otimes I + I \otimes H_{\boldlambda_{K_{j}\ell}}$, which means that
\beq 
\sigma \left( H_{ \mathbf{\boldlambda} } \right)  = \sigma \Bl( H_{\boldlambda_{L}(\boldu_{\cJ}) }  \Br) + \sigma \left( H_{  \boldlambda_{K_{j}\ell} } \right).
\eeq
Hence, if a PI N-particle box $\NboxLu = \boldlambda_{L}(\bold{u}_{\cJ}) \times  \boldlambda_{L}(\bold{u}_{\cJ^{c}} ) $ is $E$-right resonant, then there exists an N-particle box
$\mathbf{\boldlambda} = \boldlambda_{L}(\boldu_{\cJ}) \times  \boldlambda_{K_{j}\ell},$ 
where $  \boldlambda_{K_{j}\ell} \subseteq \boldlambda_{L}(\boldu_{\cJ^{c}})$,       such that 
\begin{align*}
\dist \bigl( \sigma \left( H_{ \mathbf{\boldlambda} } \bigr) , E    \right) < \tfrac{1}{2}e^{-\pa{K_{j}\ell}^{\beta}}.
\end{align*}

The same argument applies to a PI N-particle box being $E$-left resonant.
\end{proof}

We now state the energy interval multiscale analysis.
Given $m > 0$,  $L \in \N$, $\x, \, \y \in \Ndspace$, and an interval $I$, we define the event
\begin{align}
&R \left( m, \, I,\, \x, \, \y, \, L, \, N \right) =   \notag\\
& \qquad \left\{ \exists \, E \in I \sqtx{such that}\boldlambda _{L}^{(N)}(\boldx) \text{ and } \boldlambda _{L}^{(N)}(\boldy)  \text{ are not } \left(m, E\right)\text{-regular}  \right\} . 
\end{align}

\begin{proposition} \label{part4mainthm}.
 Let $\zeta, \, \tau, \beta, \, \zeta_1,\, \zeta_2,\, \gamma$ as in \eq{constant} and   $0<m_0 <m^*$. 
 Then  there exists a length scale $Z_{3}^{*}$ such that, given an interval    $I \subseteq \R$,  if for some $L_0 \geq Z_{3}^{*}$ we can verify
\beq
\P \Bigl\{R \left( m_0, \, I,\, \x, \, \y, \, L_0, \, N \right)\Bigr\}  \leq e^{-L_0^{\zeta_2}}
%\sqtx{for all} \boldx,\boldy \in \R^{Nd}\sqtx{with} d_H(\boldx,\boldy)> L_0,
\eeq
for every pair of partially separable $N$-particle boxes $\boldlambda_{L_0}^{N}(\x)$ and $\boldlambda_{L_0}^{(N)}(\y)$, 
then, for all $k = 0, 1, 2, ...$ we have,  setting $L_{k+1} = L_k^{\gamma} = L_0^{\gamma^k}$, that
\beq
\P \Bigl\{R \left( \tfrac{m_0}{2}, \, I,\, \x, \, \y, \, L_k, \, N \right) \Bigr\} \leq e^{-L_k^{\zeta_2}}
\eeq
for   every pair of partially separable $N$-particle  boxes $\boldlambda_{L_k}^{(N)}(\boldx)$ and $\boldlambda_{L_k}^{(N)}(\boldy)$.
\end{proposition}

\begin{proof}
Given $\ell$ (sufficiently large) and $0<m_{\ell} < m^*$, we set $L = \ell^{\gamma}$ and take  $m_L$ as in \eq{mlmL}.  If $\ell$ is large, we have $m(\ell) > m_\ell$, where  $m(\ell)$ is given in  \eqref{mL}, and conclude that $m(L) \ge m(\ell) > m_\ell >m_L$.

We start by showing that if
\beq
\P \Bigl\{R \left( m_{\ell}, \, I,\, \x, \, \y, \, \ell, \, N \right) \Bigr\}\leq e^{-\ell^{\zeta_2}}
\eeq
for every pair of partially separable N-particle boxes $\boldlambda_{\ell}^{(N)}(\x)$ and $\boldlambda_{\ell}^{(N)}(\y)$,
 then
\beq
\P \Bigl\{R \left( m_{L}, \, I,\, \x, \, \y, \, L, \, N \right) \Bigr\}\leq e^{-L^{\zeta_2}}
\eeq
for every pair of partially separable N-particle boxes $\boldlambda_{L}^{(N)}(\x)$ and $\boldlambda_{L}^{(N)}(\y)$.

Let $\boldlambda_{L}^{(N)}(\x)$ and  $\boldlambda_{L}^{(N)}(\y)$ be a pair of partially separable N-particle boxes.  Let $J \in 2N$. Let
$\cB_{J}$ be the  the event that there exists $E \in I$ such that  either ${\mathcal C}_{L,\ell} (\boldx)$ or ${\mathcal C}_{L,\ell} (\boldy)$ contains $J$ pairwise $\ell$-distant FI boxes that are $(m_{\ell}, \, E)$-nonregular, and let
$\cA$ be  the event that there exists $E \in I$ such that
either ${\mathcal C}_{L,\ell} (\boldx)$ or ${\mathcal C}_{L,\ell} (\boldy)$ contains one PI box   that is not $E$-preregular.
If $\bom \in \cB_{J}^{c} \cap \cA^{c}$, then for all $E \in I$ the following holds:
\begin{enumerate}
\item ${\mathcal C}_{L,\ell} (\boldx)$ and ${\mathcal C}_{L,\ell} (\boldy)$ contain at most $J-1$ pairwise $\ell$-distant FI $(m_{\ell}, \, E)$-nonregular boxes.   

\item Every PI box in ${\mathcal C}_{L,\ell} (\boldx)$ and ${\mathcal C}_{L,\ell} (\boldy)$ is $E$-preregular.
\end{enumerate}

We also define the event
\beq
\cU_{J} = \bigcup_{\boldlambda^{\pr} \in \cM_{\boldx},\boldlambda^{\pr\pr} \in  \cM_{\boldy}} \set{ \dist \left( \sigma(H_{\boldlambda^{\pr}}), \sigma(H_{\boldlambda^{\pr\pr}}) \right) < e^{-\ell^{\beta}} },
\eeq
where,
given an $N$-particle box $\NboxLa$,   by $\cM_{\bolda}$ we denote  the collection of all boxes of the following three types: 
\begin{enumerate}
\item $\NboxLa$,
\item $\boldlambda_{K_{j_t}\ell}(\boldu) \subseteq \NboxLa$, where $\boldu \in \suitc(\bolda)$, $t \in \set {1, \, 2, \,  \ldots, T_{\bolda} \,\, | \,\, T_{\bolda} \leq JN^{N}}$, and $j_t \in \set{1,\,2,\ldots, JN^{N}}$  ,
\item $\boldlambda = \boldlambda_{t_1}^{\cJ}(\boldv) \times \boldlambda_{L}(\bolda_{\cJ^c})$, where $\boldv \in \suitc(\bolda_{\cJ})$, $\cJ$ is any non empty subset of $\setN$, $t_1 \in \set{K_{j_t}\ell \,\, | \,\, t = 1, \ldots, T_{\boldv} \le \abs{\cJ}^{\abs{\cJ}}, j_t \in \set{1, \, 2, \ldots \abs{\cJ}^{\abs{\cJ}}} }$,  
\end{enumerate}
It is clear that $\abs{\cM_{\bolda}} < 2^{N+2} J^{2} N^{2N} \pa{\tfrac{2L}{\ell}}^{Nd}$, and hence it follows from Corollary~\ref{Wegner2} that
\beq\label{probUJ}
\P \pa{\cU_{\cJ}} \leq  2 N^{4N+1} J^4 \norm{\rho}_\infty\pa{\tfrac{2L}{\ell}}^{2Nd} L^{Nd} e^{-\ell^{\beta}}.
\eeq
Note that  for $\bom \in \ \cU_J^{c}$ and  $E\in I$, either every box in $\cM_{\boldx}$ is $E$-nonresonant or every box in $\cM_{\boldy}$ is $E$-nonresonant.

Let $\bom \in \cB_{J}^{c} \cap \cA^{c}  \cap \cU_J^{c}$ and $E \in I$. 
If every box in $\cM_{\boldx} $  is $E$-nonresonant, then, in particular, every PI box in ${\mathcal C}_{L,\ell} (\boldx)$ is $E$-HNR and $E$-preregular, and hence $\pa{m(\ell), \, E}$-regular by Lemma \ref{part2firstthm}.   As $m(\ell)>m_\ell$, 
we conclude that  every PI box in ${\mathcal C}_{L,\ell} (\boldx)$ is $\pa{m_\ell ,\, E}$-regular. Since  $\bom \in \cB_{J}^{c} \cap \cA^{c}$, ${\mathcal C}_{L,\ell} (\boldx)$ contains at most $J-1$ pairwise $\ell$-distant FI  $(m_\ell, \, E)$-nonregular boxes in ${\mathcal C}_{L,\ell} (\boldx)$, and all other boxes in ${\mathcal C}_{L,\ell} (\boldx)$ are  $(m_\ell, \, E)$-regular, it follows from Lemma~\ref{part2prop1a} that   $\NboxLx$ is $(m_L, \, E)$-regular. If there exists a box in $\cM_x$ that is $E$-nonresonant, then every box in $\cM_y$ must be $E$-nonresonant, and thus $\NboxLy$ is $(m_L, \, E)$-regular
by the previous argument.   Thus  for every $E \in I$ either $\NboxLx$ is $(m_L, \, E)$-regular or $\NboxLy$ is $(m_L, \, E)$-regular.    It follows that 
\beq
R \pa{m_{L}, \, I,\, \x, \, \y, \, L, \, N} \subseteq \pa{\cB_{J}^{c} \cap \cA^{c}  \cap \cU_J^{c}}^c,
\eeq
so
\beq \label{probR}
\P \pa{R \pa{ m_{L}, \, I,\, \x, \, \y, \, L, \, N} } \leq \P \pa{\cB_{J}} + \P \pa{\cA} + \P \pa{\cU_{J}}.
\eeq

Using Lemma~\ref{part1prop2} and Lemma \ref{prereg}, 
  we get 
  \beq\label{probAB}
  \P(\cB_{J})  \leq 2\pa{ \tfrac {2L}\ell}^{2Nd} e^{- \tfrac{J}{2} L^{\tfrac{\zeta_2}{\gamma}}} \qtx{and} \P(\cA)  \leq 2\pa{ \tfrac {2L}\ell}^{2Nd}  e^{-L^{\tfrac{\tau}{\gamma}}}.
  \eeq 
We now fix  
$$
J \in \Bigl( 2 L^{\beta - \tfrac{\zeta_2}{\gamma}},2 L^{\beta - \tfrac{\zeta_2}{\gamma}}+2\Bigr] \cap 2\N, $$
so  ($L$ large)  $\P(\cB_{J})  \leq  \tfrac{1}{3} e^{-L^{\zeta_2}}$ , $\P(\cA)  \leq \tfrac{1}{3} e^{-L^{\zeta_2}}$, and $\P(\cU_{J})  \leq  \tfrac{1}{3} e^{-L^{\zeta_2}}$, and we conclude from 
\eq{probR} that
\beq
\P(R \left( m_{L}, \, I,\, \x, \, \y, \, L, \, N \right)) \leq e^{-L^{\zeta_2}} .
\eeq

We now take $L_0$ large enough so that  $m(L_0)>m_{L_0}=m_0$, and the above procedure can be carried out with $\ell=L_0$,  let  $L_{k+1} = L_k^{\gamma}$ for $k=0,1,\ldots$, and set $m_k=m_{L_k}$.
To finish the proof, we just need to make sure $m_{L_k} > \tfrac{m_0}{2}$ for all $k=0,1,\ldots$, but this clearly  can be achieved by taking $L_0$ sufficiently large, similarly to the argument in \eq{halfmass}.
\end{proof}

\begin{remark}
The proof of Proposition \ref{part4mainthm} gives us more than just our desired conclusion. It shows that for $\bom \in \cB_{J}^{c} \cap \cA^{c}  \cap \cU_J^{c}$ and $E \in I$, either $\NboxLx$ is $(m_L, \,E)$-good or $\NboxLy$ is $(m_L, \,E)$-good. Hence,
\begin{align} 
& \P \Bigl\{ \exists \, E \in I \,\,so \,\, \mathbf{\Lambda}_{L_k}^{(N)}(\boldx) \,\, and \,\, \mathbf{\Lambda}_{L_k}^{(N)}(\boldy)  \,\,are\,\, not \,\, \left ({m_{\zeta_2}}, \,E \right )\text{-good} \Bigr\} \leq e^{-L_{k}^{\zeta_2}}.
\end{align}
As a consequence, we  also get a stronger form of Proposition \ref{part4mainthm}.
\end{remark}

\begin{theorem} \label{part2thm}
Let $\zeta_2 \in (\zeta, \, \tau)$ and $\gamma \in (1, \, \tfrac{1}{{\zeta_2}})$  with $\zeta \, \gamma^{2} < \zeta_2  $ be given. Assume there exists a mass $m_{\zeta_2}>0$, a length scale $L_0 = L_{0}(\zeta_2)$, and $\delta_{\zeta_2} > 0$ such that if we take $L_{k+1} = L_{k}^{\gamma}$, then for every $k \in \N$ and for every $E_1 \in \R$, setting $I(E_1)=[E_1-\delta_{\zeta_2}, E_1+\delta_{\zeta_2}]$,  we have
\begin{align} \label{appeneq2}
& \P \Bigl\{ \exists \, E \in I(E_1) \,\,so \,\, \mathbf{\Lambda}_{L_k}^{(N)}(\bolda) \,\, and \,\, \mathbf{\Lambda}_{L_k}^{(N)}(\boldb)  \,\,are\,\, not \,\, \left ({m_{\zeta_2}}, \,E \right )\text{-good} \Bigr\} \leq e^{-L_{k}^{\zeta_2}} , 
\end{align}
for every pair of $L_k-$distant N-particle boxes $\NboxLa$ and $\NboxLb$. Then there exists $L_{\zeta}$ such that for every for every $E_1 \in \R$ and every $L \geq L_{\zeta}$ 
\begin{align}
& \P \Bigl\{ \exists \, E \in I(E_1) \,\,so \,\, \NboxLa \,\, and \,\, \NboxLb  \,\,are\,\, not \,\, \left ( \tfrac{m_{\zeta_2}}{2}, \,E \right )\text{-good}\Bigr\} \leq e^{-L^{\zeta}} \notag, 
\end{align}
for every pair of L-distant N-particle boxes $\NboxLa$ and $\NboxLb$.
\end{theorem}

\begin{proof}
 Given a scale $L$ we take $K$ such that $ L_K \leq L < L_{K+1}$, and set $\ell = L_{K - 1}$.  Note that $L_K=\ell^\gamma$ and  $L_{K+1} = L_{K}^{\gamma} = L_{K-1}^{\gamma^{2}} = \ell^{\gamma^2}$, so $L=\ell^{\gamma^\pr}$ with $\gamma \le \gamma^\pr <\gamma^2$.

 Let $\NboxLa$ and $\NboxLb$ be a pair of L-distant N-particle boxes.  Given $\boldu \in \suitc(\bolda)$ and $\boldv \in \suitc(\boldb)$, we set
 \beq \notag
 \cR_{\boldu, \, \boldv} = \Bigl\{ \exists \, E \in I(E_1) \,\,so \,\, \mathbf{\Lambda}_{\ell}^{(N)}(\boldu) \,\, and \,\, \mathbf{\Lambda}_{\ell}^{(N)}(\boldv)  \,\,are\,\, not \,\, \left ({m_{\zeta_2}}, \,E \right )\text{-good} \Bigr\},
 \eeq
 and let
\beq
\cF_2 = \bigcup_{\boldu \in \suitc(\bolda), \;\boldv \in \suitc(\boldb)} \cR_{\boldu, \,\, \boldv}.
\eeq
Let  $\bom \in \cF_2^{c}$. Then, either every box in ${\mathcal C}_{L,\ell}^{(N)}(\bolda)$ is  $(m_{\zeta_2}, \, E)$-good, or every box in ${\mathcal C}_{L,\ell}^{(N)} (\boldb)$ is  $(m_{\zeta_2}, \, E)$-good for every $E \in I(E_1)$. Hence, by Lemma \ref{GKlemma}, either $\NboxLa$ is $\left( \tfrac{m_{\zeta_2}}{2}, \, E \right) $-good or $\NboxLb$ is $\left( \tfrac{m_{\zeta_2}}{2}, \, E \right) $0good for every $E \in I(E_1)$ and for every $\bom \in \cF_1^{c}$ . The conclusion follows since
\beq
\P \left( \cF_2 \right) \leq e^{- L^{\zeta}}. 
\eeq
\end{proof} 

\subsection{Completing the proof of the bootstrap multiscale analysis}
Proceeding as in \cite[Section~6]{GK1}, Theorem~\ref{maintheorem} follows from Propositions~\ref{part1mainthm}, \ref{part2mainthm}, \ref{part3mainthm},  plus Proposition~\ref{part4mainthm1} for Part (i) (the single energy bootstrap multiscale analysis), and Propositions~\ref{bridgethm},  \ref{part4mainthm} and \ref{part2thm} for Part (ii) (the energy interval bootstrap multiscale analysis).

\section{Localization for the multi-particle Anderson model}\label{secloc}
In this section we prove Corollary \ref{PPS}.  We assume that  \eq{concmsa} holds on an interval $I$ and prove the conclusions of Corollary \ref{PPS} in $I$. Since $H_{\bom}^{(N)}$ has compact spectrum, it can be covered by a finite number of intervals where  \eq{concmsa} holds,
and hence the conclusions of Corollary \ref{PPS} hold on $\R$.

We take  $L_0$ large enough  so  \eq{concmsa} holds for all  $L\ge L_0$, and set $L_{k+1}=2L_k$ for $k \in \N$.

\subsection{Anderson Localization}
\begin{proof}[Proof of Corollary \ref{PPS}(i)] 
For $\boldx_0 \in \Ndspace$ and $k \in \N$, we set\beq
A_{k+1} (\boldx_0) = \Bigl\{ \boldx \in \Ndspace \,\, | \,\, L_k < d_{H}( \boldx, \, \boldx_0 ) \leq  L_{k+1} = 2L_k\Bigr\}.
\eeq
Then for $\x \in A_{k+1} (\boldx_0)$, we have $d_{H}(\x_0, \, \x) >L_k$. 
Moreover, it follows from the definition of the Hausdorff distance  that
\beq
d_{H}( \boldx, \, \boldy )  \leq \norm{\boldx - \boldy} \leq d_{H}( \boldx, \, \boldy )  + \diam\x \qtx{for} \boldx, \, \boldy \in \Ndspace ,
\eeq
where 
\beq
\diam \x := \max_{i, \, j = 1, \dots, N} \norm{x_i - x_j}.
\eeq
Hence $\abs{ A_{k+1} (\boldx_0) } \leq \bigl( 2L_{k} + \diam\x_0 \bigr) ^{Nd}$. 
Let us define the event
\begin{align*}
& E_{k}(\x_0) := \\
& \bigcup_{ \x \in  A_{k+1} (\boldx_0) }\Bigl\{ \exists \, E \in I\sqtx{such that}\mathbf{\Lambda}_{L_k} (\x_0) \sqtx{and}\mathbf{\Lambda}_{L_k} (\x)\sqtx{are}(m, \,E)\text{-nonregular} \Bigr\}.
\end{align*}
Applying \eq{concmsa}, we get
\beq
\P  \Bigl( E_{k}(\x_0)  \Bigr) \leq \bigl( 2L_{k} + \diam\x_0 \bigr) ^{Nd} e^{-L_{k}^{\zeta}} ,
\eeq
so we have
\beq
\sum_{k = 0} ^{\infty} \P  \Bigl( E_{k}(\x_0)  \Bigr) < \infty.
\eeq
It  follows from the  Borel Cantelli Lemma that
\beq
\P \Bigl\{ E_{k}(\x_0) \,\,\text{occurs infinitely often}  \Bigr\} = 0 \qtx{for all} \x_0 \in \Ndspace,
\eeq
so
\beq
\P \Bl\{  \bigcup_{\x_0 \in \Ndspace} \bigl\{ E_{k}(\x_0) \,\,\text{occurs infinitely often} \bigr\} \Br\} = 0,
\eeq
i.e., 
\beq
\P  \Bigl\{ E_{k}(\x_0) \,\,\text{occurs infinitely often for some } \x_0 \in \Ndspace \Bigr\}  = 0.
\eeq

Let $\Omega_0 := \Bigl\{ E_{k}(\x_0) \,\,\text{occurs infinitely often for some } \x_0 \in \Ndspace \Bigr\}.$ Take $\bom \in \Omega_0$ and let $H = H_{\bom}$. We will be done if we can prove that every generalized eigenvalue of $H$ in $I$ is actually an eigenvalue by showing that any corresponding generalized eigenfunction has exponential decay.

Let $E \in I$ be a generalized eigenvalue of $H$ with the corresponding nonzero polynomially bounded generalized eigenfunction $\psi$, that is $H \psi = E \psi$ and for some $ C< \infty, \, t \in \N$, we have
\beq
\abs{\psi(\x)} \leq C \bigl( 1 + \norm{\x} \bigr)^t \quad \text{for every } \x \in \Ndspace.
\eeq
Since $\psi$ is non zero, there exists $\x_0 \in \Ndspace$ such that $\psi(\x_0) \neq 0$ .  We know that $E_{k}(\x_0)$ can only occur finitely many times. Thus there exists $k_1$ such that for every $k > k_1$, and for any $\x \in A_{k+1} (\x_0)$, either $ \mathbf{\Lambda}_{L_k}(\x_0) $ is $(m , \, E)$-regular or $\mathbf{\Lambda}_{L_k}(\x)$ is $(m , \, E)$-regular.

If $E \notin \sigma(H_{ \mathbf{\Lambda}_{L_k}(\x_0)    })$,  we have
\beq
\psi(\x_0) = \sum_{ (\bolda, \,\boldb) \in \partial \mathbf{\Lambda}_{L_k}(\x_0)  } G_{\mathbf{\Lambda}_{L_k}(\x_0)  } (E; \, \x_0 , \, \bolda) \psi(\boldb).
\eeq
Moreover, if  $ \mathbf{\Lambda}_{L_k}(\x_0)$ is $(m,\, E)$-regular, then 
\begin{enumerate}
\item $E \notin \sigma(H_{ \mathbf{\Lambda}_{L_k}(\x_0)    })$,  and
\item $\abs{ \psi(\x_0)  } \leq s_{Nd} L_{k}^{Nd-1} e^{-m \left ( \frac{L_k}{2} -1 \right )} \abs{ \psi(\x_1)  }$ for some $x_1 \in \partial_+  \mathbf{\Lambda}_{L_k}(\x_0)$, and hence
$\abs{ \psi(\x_0)  } \leq C  s_{Nd} L_{k}^{Nd-1} e^{-m \left ( \frac{L_k}{2} -1 \right )}
 \bigl( 1 + \norm{\x_0} + L_k \bigr)^t$ .

\end{enumerate}
Since we know $\psi(\x_0) \neq 0$, this implies there must exist $k_2$ such that for every $k > k_2$, $ \mathbf{\Lambda}_{L_k}(\x_0) $ is not $(m , \, E)$-regular. Taking $k_3 = \max \{ k_1, \, k_2  \}$, we can conclude that for every $k > k_3$, $\mathbf{\Lambda}_{L_k}(\x)$ is $(m , \, E)$-regular for every $\x \in A_{k+1} (\x_0)$. For what we are doing, we will be taking k such that 
\beq
L_k\gg \diam \x_0.
\eeq
Thus, if $\x \in {A}_{k}(\x_0)$ with $k > k_3$, we have $\boldlambda_{L_k}(\x)$ is $(m , \, E)$-regular and thus
\begin{align}
 \abs{\psi(\x)} &\le  C  s_{Nd} L_{k}^{Nd-1} e^{-m  \pa{\frac{L_k}{2} -1 }}   (1 + \norm{\x_0} +  \diam \x_0+ 2L_{k} )^t  \\   \notag
& \le e^{-\frac m 4{L_k}}\le e^{-\frac m 8 {d_{H}( \boldx, \, \boldx_0 )}} \le  e^{-\frac m 8 (\norm{\x - \x_0}-\diam \boldx_0)}=  e^{\frac m 8 \diam \boldx_0}   e^{-\frac m 8 \norm{\x - \x_0}},
\end{align}
provided $k$ is sufficiently large, so
$\psi$ decays exponentially.
\end{proof}

\subsection{Dynamical Localization} We will use the generalized eigenfunction expansion for  $H_{\bom} =H_{\bom}^{(N)}$ to prove dynamical localization (and SUDEC in Subsection~\ref{secSUDEC}).
We will follow the short review (and the notation) given in \cite[Section~5]{GKber}, and   refer to  \cite[Section 3]{KKS} for full details.
 We fix $\nu = \frac{Nd+1}{2} $, and for $\bolda \in \Ndspace$ let  $T_{\bolda}$ denote  the  operator on $\H = \ell^{2} (\Ndspace)$ given  by multiplication by the function $\scal{\boldx - \bolda}^{\nu}$, where $\scal{\boldx} = \sqrt{1 +\norm{\boldx}^2}$,  and set $T = T_{\mathbf{0}}$.  We consider
  weighted spaces $\H_-$ and $\H_+$, operators  $T_-$ and $T_+$, and spectral measure $\mu_{\bom}$ and  generalized eigenprojectors $\Pb_{\bom} (\lambda)$  in terms of which we have the generalized eigenfunction expansion for the (bounded operator)  $H_{\bom} =H_{\bom}^{(N)}$  given in \cite[Eq.~(5.23)]{GKber}.

For $\boldx \in \Ndspace$, we denote $\Chi_{\boldx}$ to be the orthogonal  projection onto $\delta_{\boldx}$ where the family $\bigl \{\delta_{\boldx} \,\, | \,\,\boldx \in \Ndspace \bigr \}$ is the standard orthonormal basis of $\H$. 

\begin{lemma}
There exists a constant $c = c(d, \, N) < \infty$ such that for $\P$ almost every $\bom$
\beq
\text{tr } \Bl( T^{-1} \, f \, (H_{\bom}) \, T^{-1} \Br) \leq c \norm{f}_{\infty} .
\eeq
\end{lemma}

\begin{proof}
Given $\boldx \in \Ndspace$, we have
\begin{align} 
 \scal{\delta_{\boldx} \, , \, T^{-1}   \, f \, (H_{\bom}) \, T^{-1}\, \delta_{\boldx}}   
=\scal{\boldx }^{-2\nu} \, \scal{ \delta_{\boldx}  \, , \, 
f \, (H_{\bom}) \delta_{\boldx} } \leq \scal{\boldx }^{-2\nu}\, \norm{f}_{\infty}.
\end{align}
It follows that
\begin{align}
\text{tr } \Bl( T^{-1} \,  f \, (H_{\bom}) \, T^{-1} \Br)
 \leq \norm{f}_{\infty}  \displaystyle{ \sum_{\boldx \in \Ndspace}} 
\scal{\boldx }^{-2\nu} = c  \norm{f}_{\infty}. 
\end{align}
\end{proof}

For $\boldx \in \Ndspace$, consider $\norm{T_+\,\Chi_{\boldx}}$ and $\norm{\Chi_{\boldx}\,T_-}$ as operators from $\H$ to $\H$. Note that 
  \beq
 \norm{T_+\,\Chi_{\boldx}}=\norm{\Chi_{\boldx}\,T_-} =\scal{\boldx }^{\nu}.
 \eeq

\begin{lemma} \label{dynlem1}
For $\P$-almost every $\bom$, for every $\boldx,\,\,\boldy \in \Ndspace$, and for $\mu_{\bom}-$almost every $\lambda$, we have
\beq
\norm{\Chi_{\boldx} \, \Pb_{\bom} (\lambda)\,\Chi_{\boldy}}_{1} \leq \scal{\boldx }^{\nu} \scal{\boldy }^{\nu}.
\eeq
\end{lemma}

\begin{proof}
We have $\norm{\Chi_{\boldx} \, \Pb_{\bom} (\lambda)\,\Chi_{\boldy}}_{1} = 
\norm{\Chi_{\boldx} T_-\,T_-^{-1}\, \Pb_{\bom} (\lambda) \, T_+^{-1}\,T_+ \, \Chi_{\boldy}}_{1}$.  Thus  
\beq
\norm{\Chi_{\boldx} \, \Pb_{\bom} (\lambda)\,\Chi_{\boldy}}_{1} 
\leq \norm{ \Chi_{\boldx}\,T_- } \, 
\norm{T_-^{-1}\, \Pb_{\bom} (\lambda) \, T_+^{-1} }_1 \, \norm{ T_+\,\Chi_{\boldx} }= \scal{\boldx }^{\nu} \scal{\boldy }^{\nu},
\eeq
since $\norm{T_-^{-1}\, \Pb_{\bom} (\lambda) \, T_+^{-1} }_1 = 1$ for  $\mu_{\bom}-$almost every $\lambda$  (see \cite[Eq.~(5.23)]{GKber}).
\end{proof}

\begin{lemma} \label{dynlem2}
Let $\boldx, \, \boldy \in \Ndspace$ with $d_{H}(\x, \, \y) > \ell$, and suppose 
\beq
\P \set{R \pa{m, \, I,\, \boldx, \, \boldy, \, \ell, \, N } }\leq e^{-\ell^{\zeta}}.
\eeq
Then for $\bom \notin  R \pa{m, \, I,\, \boldx, \, \boldy, \, \ell, \, N }$ we have
\beq
\norm{\Chi_{\boldx} \, \Pb_{\bom} (\lambda)\,\Chi_{\boldy}}_{1} \leq 
 s_{Nd}\,  \ell^{Nd+\nu -1 } \, e^{-m \frac{\ell}{2}} \,  
\scal{\boldx }^{\nu} \scal{\boldy }^{\nu}
\eeq
for $\mu_{\bom}-$almost every $\lambda \in I$.
\end{lemma}

\begin{proof} Let  $\bom \notin  R \pa{m, \, I,\, \boldx, \, \boldy, \, \ell, \, N }$. Then 
for every $\lambda \in I$, either $\Nboxlx$ or  $\Nboxly$
is $ (m, \lambda) \text{-regular }$. Moreover, we have that $ \norm{\Chi_{\boldx} \, \Pb_{\bom} (\lambda)\,\Chi_{\boldy}}_{1}
= \norm{\Chi_{\boldy} \, \Pb_{\bom} (\lambda)\,\Chi_{\boldx}}_{1}$, so without loss of generality, we may assume $\Lambda_{\ell}(\boldx) \text{ is } (m, \lambda) \text{- regular}$. 

 For $\mu_{\bom}$-almost every $\lambda \in I,$   $\psi= \Pb_{\bom} (\lambda)\,\Chi_{\boldy} \, \phi$, with $\phi \in \H$, is a generalized eigenfunction of 
$H_{\bom}$ corresponding to the generalized eigenvalue $\lambda$.  Then 
\begin{align} 
\psi(\boldx) &= 
\displaystyle{ \sum_{(\bolda, \, \boldb) \in \partial \Nboxlx}}
G_{\Nboxlx} (\lambda; \, \boldx, \,\bolda) \, \psi(\boldb) .\notag
\end{align}
Thus it follows from the  regularity of $\Lambda_{\ell}(\boldx)$ that 
\begin{align} 
&\norm{\Chi_{\boldx} \, \Pb_{\bom} (\lambda)\,\Chi_{\boldy}}_{1}  
\leq \displaystyle{ \sum_{(\bolda, \, \boldb) \in \partial \Lambda_{\ell}(\boldx)}}
\abs{ G_{\Nboxlx} (\lambda; \, \boldx, \,\bolda)}  \,  \norm{\Chi_{\boldb} \, \Pb_{\bom} (\lambda)\,\Chi_{\boldy}}_{1} \\ \notag
& \qquad  \leq \abs{\partial \Nboxlx} \max_{(\bolda, \, \boldb) \in \partial \Lambda_{\ell}(\boldx)} \abs{ G_{\Nboxlx} (\lambda; \, \boldx, \,\bolda)}  \,  \norm{\Chi_{\boldb} \, \Pb_{\bom} (\lambda)\,\Chi_{\boldy}}_{1} \\ \notag
&\qquad  \leq s_{Nd}\,  \ell^{Nd-1 } \, e^{-m \frac{\ell}{2}} \,  
\scal{ \norm{\boldx } +  \tfrac \ell 2 +1 }^{\nu}  \scal{\boldy }^{\nu}\\ \notag
&\qquad \le  s_{Nd}\,  \ell^{Nd-1 } \, e^{-m \frac{\ell}{2}} \, 2^{\frac \nu 2} \scal{ \tfrac \ell 2+1} ^{\nu}\,
\scal{\boldx }^{\nu} \scal{\boldy }^{\nu} \\ \notag
&\qquad \le  s_{Nd}\,  \ell^{Nd+\nu -1 } \, e^{-m \frac{\ell}{2}} \, 
\scal{\boldx }^{\nu} \scal{\boldy }^{\nu},
\end{align}
where we used  
\beq\label{scalsum}
\langle y_{1} +y_{2} \rangle \le \sqrt{2}\langle y_{1} \rangle
\langle y_{2}\rangle \qtx{for}  y_1,y_2 \in \R^k, \qtx{any} k\in \N.
\eeq
\end{proof}

Corollary \ref{PPS}(ii) is an immediate consequence of the following theorem.
\begin{theorem} 
(Decay of the Kernel) 

 Let $I$ be an open interval where the conclusions of Theorem \ref{maintheorem} holds. Then for every $0<\zeta_4 < 1$ and $\boldy \in \Ndspace$ there exists a constant $ C(\boldy) $ such that 
\beq
\E \Bl ( \sup_{\norm{g} \leq 1}  \norm{ \Chi_{\boldx} \, (g \, \Chi_{I})(H_{\bom})   \,\Chi_{\boldy} }_1   \Br) \leq C(\boldy)\,e^{-d_{H}( \boldx , \, \boldy )^{\zeta_4}} \qtx{for all} \boldx \in \Ndspace,
\eeq
where the supremum is taken over all bounded Borel functions $g$ on $\R$, and $\norm{g} = \sup_{t \in \R} \abs{g(t)}$. 
\end{theorem}

\begin{proof}
Let us fix $\boldy \in \Ndspace.$   We will apply our main result using $\zeta_2 \in (\zeta\, , \, 1)$. 

For $\boldx \in \Ndspace$, let us denote $$ F_{\boldx} (\bom) = \sup_{\norm{g} \leq 1}  \norm{ \Chi_{\boldx} \, (g \, \Chi_{I})(H_{\bom})   \,\Chi_{\boldy} }_1. $$
Thus our goal is to show that $\E \Bl ( F_{\boldx} (\bom)   \Br) \leq  C\,e^{-d_{H}( \boldx , \, \boldy )^{\zeta_4}}$ for all $x \in \Z^d$  for some  constant $C=C(\boldy)$.

As in  \cite{GK1}, we have
\begin{align}
\norm{ \Chi_{\boldx} \, (g \, \Chi_{I})(H_{\bom})   \,\Chi_{\boldy} }_1 &\leq 
\int_I \norm{ \Chi_{\boldx} \, g(\lambda) \, \Pb_{\bom} (\lambda)  \,\Chi_{\boldy} }_1 d\mu_{\bom} (\lambda)  \\
&=  \int_I \abs{g(\lambda)} \norm{ \Chi_{\boldx} \, \Pb_{\bom} (\lambda)  \,\Chi_{\boldy} }_1 d\mu_{\bom} (\lambda),\notag
\end{align}
 and thus
\beq
F_{\boldx} (\bom) \leq \int_I  \norm{ \Chi_{\boldx} \, \Pb_{\bom} (\lambda)  \,\Chi_{\boldy} }_1 d\mu_{\bom} (\lambda).
\eeq

We  will divide the proof into the case where $d_{H}(\boldx, \, \boldy) >L_k$ for some  $k $ large enough ($k\ge \cK_0$), and the case where $d_{H}(\boldx, \, \boldy) \leq L_{\cK_0}$.

Case 1: If $d_{H}(\boldx, \, \boldy) >L_k$ for some $k\ge \cK_0$,  let us take the largest $k$ such that $d_{H}(\boldx, \, \boldy) >L_k$ but $d_{H}(\boldx, \, \boldy) \leq L_{k+1}$.
Let us denote the set 
\beq \notag
\cA = \Bl\{\exists \,E \in I \sqtx{such that}\mathbf{\Lambda}_{L_k}^{(N)}(\boldx) \text { and } \mathbf{\Lambda}_{L_k}^{(N)}(\boldy) \text{ are } (m, E) \text{-nonregular} \Br\}.
\eeq
Then $\E \Bl ( F_{\boldx} (\bom)   \Br) = \E \Bl ( F_{\boldx} (\bom) \, ; \, \bom \in \cA  \Br) +
\E \Bl ( F_{\boldx} (\bom) \,;\, \bom \notin \cA  \Br)  $.

To estimate $\E \Bl ( F_{\boldx} (\bom) \, ; \, \bom \in \cA  \Br)$, we apply Lemma \ref{dynlem1} to get
\begin{align}
\E \Bl ( F_{\boldx} (\bom) \, ; \, \bom \in \cA  \Br) & \leq  \E \Bl (\scal{\boldx }^{\nu} \scal{\boldy }^{\nu} \mu_{\bom}(I)  ; \, \bom \in \cA  \Br) = \scal{\boldx }^{\nu} \scal{\boldy }^{\nu} \E \Bl( \mu_{\bom}(I) \,\, \Chi_{\cA} (\bom)   \Br).
\end{align}
But we know that $\P(\cA) \leq e^{-L_{k}^{\zeta_2}}$, and
\begin{align}
\E \Bl( \mu_{\bom}(I) \,\, \Chi_{\cA} (\bom)   \Br) & \leq \E \Bl( (\mu_{\bom}(I))^{2} \Br)^{\frac{1}{2}} \E \Bl( ( \Chi_{\cA} (\bom)  )^{2} \Br)^{\frac{1}{2}}  =  \E \Bl( (\mu_{\bom}(I))^{2} \Br)^{\frac{1}{2}} \P(\cA)^{\frac{1}{2}} ,
\end{align}
so
\begin{align}
\E \Bl ( F_{\boldx} (\bom) \, ; \, \bom \in \cA  \Br) & \leq 
\scal{\boldx }^{\nu} \scal{\boldy }^{\nu}\E \Bl( (\mu_{\bom}(I))^{2} \Br)^{\frac{1}{2}} e^{- \tfrac{1}{2} L_{k}^{\zeta_2}   } \notag \\
& = C_{1}  \, \E \Bl( (\mu_{\bom}(I))^{2} \Br)^{\frac{1}{2}} e^{- \tfrac{1}{2} L_{k}^{\zeta_2}   },
\end{align}
where $C_1 =\scal{\boldx }^{\nu} \scal{\boldy }^{\nu}  = C_{1}(\x, \, \y, \nu).$

To estimate $\E \Bl ( F_{\boldx} (\bom) \, ; \, \bom \notin \cA  \Br)$, we apply Lemma \ref{dynlem2} to get
\begin{align}
\E \Bl ( F_{\boldx} (\bom) \, ; \, \bom \notin \cA  \Br) & \leq  \E \Bl ( C_1 s_{Nd}\,L_{k}^{Nd-1 + \nu} \, e^{-m \frac{L_k}{2}} \mu_{\bom}(I)   ; \, \bom \notin \cA  \Br) \notag\\
& \leq  C_1 s_{Nd}\,L_{k}^{Nd-1 + \nu} \, e^{-m \frac{L_k}{2}}  \E ( \mu_{\bom}(I) ).
\end{align}
Hence
\beq \notag
\E \Bl ( F_{\boldx} (\bom) \Br) \leq    C_{1}  \, \E \Bl( (\mu_{\bom}(I))^{2} \Br)^{\frac{1}{2}} e^{- \tfrac{1}{2} L_{k}^{\zeta_2}   }       +  C_1 s_{Nd}\,L_{k}^{Nd-1 + \nu} \, e^{-m \frac{L_k}{2}}  \E ( \mu_{\bom}(I) ).
\eeq
Since $k \geq \cK_0$, i.e. $L_k$ is large enough, we can conclude
\beq
\E \Bl ( F_{\boldx} (\bom) \Br) \leq 2 C_{1}  \, \E \Bl( (\mu_{\bom}(I))^{2} \Br)^{\frac{1}{2}} e^{- \tfrac{1}{2} L_{k}^{\zeta_2}   }.
\eeq
But $d_{H}(\boldx, \, \boldy) \leq L_{k+1}$, and $\norm{\x} \leq \norm{\x - \y} + \norm{\y} \leq L_{k+1} + \diam{\y} + \norm{\y}$, so
\beq
\scal{x}^{\nu} \leq 2^{\frac \nu 2}\scal{ L_{k+1}}^{\nu}\scal{\diam{\y} + \norm{\y} }^{\nu},
\eeq
which means 
\beq
\E \Bl ( F_{\boldx} (\bom) \Br) \leq C_{2} e^{- \tfrac{1}{2} L_{k}^{\zeta_2}   },
\eeq
where $C_2 = 2^{1+\frac \nu 2}\scal{ L_{k+1}}^{\nu}\scal{\diam{\y} + \norm{\y} }^{\nu} \scal{\y} ^{\nu} \E \Bl( (\mu_{\bom}(I))^{2} \Br)^{\frac{1}{2}} $.
If we take $L_{k}$ to be sufficiently large (which is the same as saying $\cK_0$ is sufficiently large), \begin{align}
 \E \Bl ( F_{\boldx} (\bom) \Br) & \leq\scal{\diam{\y} + \norm{\y} }^{2\,\nu} e^{- \tfrac{1}{4} L_{k}^{\zeta_2}   } =\scal{\diam{\y} + \norm{\y} }^{2\,\nu} e^{- \tfrac{1}{8} L_{k+1}^{\zeta_2}   }  \notag \\
& \leq\scal{\diam{\y} + \norm{\y} }^{2\,\nu}  e^{- L_{k+1}^{\zeta}   }  \leq\scal{\diam{\y} + \norm{\y} }^{2\,\nu} e^{- d_{H}(\x, \, \y)^{\zeta}   }. 
\end{align}

Case 2: If $d_{H}(\boldx, \, \boldy) \leq L_{\cK_0}$, we have  $d_{H}(\x, \, \y) \leq L_{\cK_0}.$ Once again we  apply Lemma \ref{dynlem1} to get
\begin{align}
\E \Bl ( F_{\boldx} (\bom)  \Br)  &\leq  \E \Bl ( \scal{\boldx }^{\nu} \scal{\boldy }^{\nu}\mu_{\bom}(I)    \Br) = \scal{\boldx }^{\nu} \scal{\boldy }^{\nu}\E \Bl( \mu_{\bom}(I)    \Br) \\
& \leq \scal{ \text{d}_{H}(\x, \, \y) + \diam{\y} + \norm{\y}} ^{\nu}  \scal{\boldy }^{\nu}\E \Bl( \mu_{\bom}(I)    \Br) \notag \\
&  \leq  \scal{  L_{\cK_0}+ \diam{\y} + \norm{\y}} ^{\nu}  \scal{\boldy }^{\nu}\E \Bl( \mu_{\bom}(I)    \Br) \leq C_2 \, e^{L_{\cK_0}^{\zeta}} e^{- d_{H}(\x, \, \y) ^{\zeta}} , \notag
\end{align}
where $C_2 =  \scal{  L_{\cK_0}+ \diam{\y} + \norm{\y}} ^{\nu} \scal{\boldy }^{\nu} \E \Bl( \mu_{\bom}(I)    \Br)$.

Thus, we get
\beq
\E \Bl ( F_{\boldx} (\bom)  \Br) \leq
\begin{cases}
C_2 \, e^{L_{\cK_0}^{\zeta}} e^{- d_{H}(\x, \, \y) ^{\zeta}}, & \text{provided }d_{H}(\boldx, \, \boldy) \leq L_{\cK_0} \\
(1+ \diam{\y} + \norm{\y} )^{2\,\nu} e^{- \text{d}_{H}(\x, \, \y)^{\zeta}   }, & \text{provided } d_{H}(\boldx, \, \boldy) > L_{\cK_0} 
\end{cases}.
\eeq
To get our desired result, we can just take 
\beq
C=C(y) = \Bl( 1+  L_{\cK_0} + \diam{\y} + \norm{\y} \Br) ^{2\nu} \Bl(   \E \bigl( \mu_{\bom}(I) \bigr)  + 1 \Br) .
\eeq
\end{proof}

\subsection{SUDEC}\label{secSUDEC}

To prove Corollary \ref{PPS}(iii)  we  follow \cite{GKjsp}.  Note  that for all $\bolda, \, \boldb \in \Ndspace$ we have $\norm{\Chi_{\boldb} T_{\bolda}} \leq \scal{\bolda - \boldb}^{\nu}$,
and
 it follows from \eq{scalsum} that
\beq \label{part7eq0}
\norm{T_{\boldb}^{-1} T_{\bolda} } \leq 2 ^{\tfrac{\nu}{2}} \scal{\boldb - \bolda}^{\nu}.
\eeq
We  write
$E_{A}(H_{\bom}):=\Chi_{A}(H_{\bom})$ for a Borel measurable set $A\subset \R$, and let  $E_{\lambda}(H_{\bom}):=E_{\{ \lambda\}}(H_{\bom})=\Chi_{\{ \lambda\}}(H_{\bom}) .$

\begin{definition}
Given $\bom$, $\lambda \in \R$, and $\bolda \in \Ndspace$, define
\beq
\mathbf{W}_{\bolda, \, \bom}(\lambda) :=
\begin{cases}
\sup_{\phi \in \cS_{\bom, \, \lambda}} \frac{\norm{ \Chi_{\bolda} {\Pb}_{\bom}(\lambda) \phi   }}{\norm{T_{\bolda} ^{-1} {\Pb}_{\bom}(\lambda) \phi   }}, &if \,\, {\Pb}_{\bom}(\lambda) \neq 0, \\
0, &otherwise,
\end{cases}
\eeq
where $\cS_{\bom, \, \lambda} = \Bl\{ \phi \in \H_{+} \,\,: \,\,  \Pb_{\bom}(\lambda) \phi \neq 0    \Br\}$. We also define
\beq
W_{\bolda, \, \bom}(\lambda) :=
\begin{cases}
\sup_{\phi \in \cT_{\bom, \, \lambda}} \frac{\norm{ \Chi_{\bolda} E_{\bom}(\lambda) \phi   }}{\norm{T_{\bolda} ^{-1} E_{\bom}(\lambda) \phi   }}, &if \,\, E_{\bom}(\lambda) \neq 0, \\
0, &otherwise,
\end{cases}
\eeq
where $\cT_{\bom, \, \lambda} = \Bl\{ \phi \in \H \,\,: \,\,  E_{\bom}(\lambda) \phi \neq 0    \Br\}$, and
\beq
Z_{\bolda, \, \bom}(\lambda) :=
\begin{cases}
 \frac{  \norm{ \Chi_{\bolda} E_{\bom}(\lambda)    } _{2}}{\norm{T_{\bolda} ^{-1} E_{\bom}(\lambda)   }_{2}}, &if \,\, E_{\bom}(\lambda) \neq 0, \\
0, &otherwise.
\end{cases}
\eeq

\end{definition}

Note that $
Z_{\bolda, \, \bom}(\lambda) \leq W_{\bolda, \, \bom}(\lambda) \leq \mathbf{W}_{\bolda, \, \bom}(\lambda)  \leq 1  $ (see \cite{GKjsp}).

\begin{remark} \label{part7rem1}
Let  $ \phi \in \H_{+} $, then $\Chi_{\bolda} \Pb_{\bom}(\lambda) \phi = \Chi_{\bolda} T_{\bolda} T_{\bolda}^{-1} \Pb_{\bom}(\lambda) \phi $. Then
\begin{align}
\norm{\Chi_{\bolda} \Pb_{\bom}(\lambda) \phi} &= \norm{\Chi_{\bolda} T_{\bolda} T_{\bolda}^{-1} \Pb_{\bom}(\lambda) \phi} \leq \norm{\Chi_{\bolda} T_{\bolda} } \norm{ T_{\bolda}^{-1} \Pb_{\bom}(\lambda) \phi} \notag\\
& \leq \norm{ T_{\bolda}^{-1} \Pb_{\bom}(\lambda) \phi}.
\end{align}
Thus $\mathbf{W}_{\bolda, \, \bom}(\lambda) \leq 1$ for every $\bolda \in \Ndspace$, every $\bom$, and  $\mu_{\bom}-$almost every $\lambda \in \R.$
Moreover, 
\begin{align}
\norm{ T_{\bolda}^{-1} \Pb_{\bom}(\lambda) \phi} &\leq \norm{T_{\bolda}^{-1} T } \norm{T^{-1} \Pb_{\bom}(\lambda) \phi} \leq 2^{\tfrac{\nu}{2}} \scal{\bolda}^{\nu} \norm{  \Pb_{\bom}(\lambda) \phi}_{-} \notag\\
& \leq 2^{\tfrac{\nu}{2}} \scal{\bolda}^{\nu} 
\norm{  \phi}_{+}.
\end{align}
\end{remark}

\begin{remark}
Given $\boldx, \, \boldy \in \Ndspace$, by \cite{GKjsp}, 
$
\mathbf{W}_{\boldx, \, \bom}(\lambda) \,\,\mathbf{W}_{\boldy, \, \bom}(\lambda)
$
is measurable (in $\lambda$) with respect to the measure $\mu_{\bom}$ for $\P$-a.e.\  $\bom$. Moreover, we have measurability of $\norm{ \mathbf{W}_{\boldx, \, \bom}(\lambda) \,\,\mathbf{W}_{\boldy, \, \bom}(\lambda) } _{L^{\infty}(I, d\mu_{\bom}(\lambda))} $ with respect to $\bom$.            From Remark \ref{part7rem1}, we also have $\norm{ \mathbf{W}_{\boldx, \, \bom}(\lambda) \,\,\mathbf{W}_{\boldy, \, \bom}(\lambda) } _{L^{\infty}(I, \,d\mu_{\bom}(\lambda))} \leq 1$ for $\P-a.e. \,\,\bom.$ 
\end{remark}

\begin{lemma} \label{part7lem1}
Let $\boldx, \, \boldy \in \Ndspace$ and $\bom \in R(m, \, L, \, I, \, \boldx, \, \boldy)$, where
\begin{align} \label{usefulno}
R&(m, \, L, \, I, \, \boldx, \, \boldy) = \notag \\
& \Bl\{ \forall \, E \in I, \, either \, \NboxLx \,\,or\,\, \NboxLy\, \,is\,\, (m,\,E)-regular  \Br\}.
\end{align}
Then there exists a constant $C>0$ such that
\beq
\norm{ \mathbf{W}_{\boldx, \, \bom}(\lambda) \,\,\mathbf{W}_{\boldy, \, \bom}(\lambda) } _{L^{\infty}(I, d\mu_{\bom}(\lambda))}  \leq C e^{-m\tfrac{L}{4}}
\eeq
\end{lemma}

\begin{proof}
Let $\bom \in R(m, \, L, \, I, \, \boldx, \, \boldy)$. Since $\bom \in R(m, \, L, \, I, \, \boldx, \, \boldy)$, we know that for every $\lambda \in I$, either $\NboxLx$ or $\NboxLy$ is $(m, \, \lambda)-regular.$ Without loss of generality, we may assume $\NboxLx$ is $(m, \, \lambda)-regular.$

From \cite{KKS} we have that  for $\mu_{\bom}-a.e.$ $\lambda \in I$, ${\Pb}(\lambda)\,\psi := \Pb_{\bom}(\lambda)\,\psi $ is a generalized eigenfunction of $H:= H_{\bom}$ for every $\phi \in \H_{+}$. Let $\phi \in \H_{+}$, and denote $\psi = {\Pb}(\lambda)\,\phi$. Then
\begin{align}
&\abs{\Chi_{\x} {\Pb}(\lambda)\,\phi}   =\abs{\psi(\boldx)}  = \abs{ \sum_{ (\bolda, \,\boldb) \in \delta \mathbf{\Lambda}_{L}(\x)  } G_{\mathbf{\Lambda}_{L_k}(\x_0)  } (E; \, \x , \, \bolda) \psi(\boldb)} \notag \\
& \qquad \leq s_{Nd} \,  L^{Nd-1} \,e^{-m\norm{\x -\bolda}} \abs{\Chi_{\boldb} \psi } \notag \leq s_{Nd} \,  L^{Nd-1} \,e^{-m\norm{\x -\bolda}} \norm{\Chi_{\boldb} T_{\x}} \norm{T_{\x}^{-1} \psi    } \notag \\
&\qquad \leq s_{Nd} \,  L^{Nd-1} \,e^{-m\norm{\x -\bolda}} \scal{\boldx - \boldb}^{\nu}\norm{T_{\x}^{-1} {\Pb}(\lambda)\,\phi} \notag \\
&\qquad \leq s_{Nd} \,  L^{Nd-1} \,e^{-m \tfrac{L}{2}} L^{\nu}\norm{T_{\x}^{-1} {\Pb}(\lambda)\,\phi}.
\end{align}
Thus there exists $K_0 > 0$ such that if $L \geq K_0$, then
\beq
\abs{\Chi_{\x} {\Pb}(\lambda)\,\phi} \leq  e^{-m \tfrac{L}{4}} \norm{T_{\boldx}^{-1} {\Pb}(\lambda)\,\phi}.
\eeq 
If $L < K_0$, then there exists a constant $C_1 > 1$ such that
\begin{align}
\abs{\Chi_{\x} {\Pb}(\lambda)\,\phi} & \leq s_{Nd} \,  K_{0}^{Nd-1} \,e^{-m\tfrac{L}{2}} K_{0}^{\nu}\norm{T_{\x}^{-1} {\Pb}(\lambda)\,\phi} \notag\\
&\leq C_1 \,e^{-m\tfrac{L}{4}} \norm{T_{\x}^{-1} {\Pb}(\lambda)\,\phi}.
\end{align}
Using the bound from Remark \ref{part7rem1} for the term in $\boldy$, we get our desired result.
\end{proof}

\begin{theorem} 
Let $I$ be an open interval where the conclusion of Theorem \ref{maintheorem} holds. Then for every $\zeta_1 \in (0, \, 1) $, there exists a constant $C_{\zeta_1}$ such that for every $ \x, \, \y \in \Ndspace$,
\beq
\E \Bl\{ \norm{ \mathbf{W}_{\boldx, \, \bom}(\lambda) \,\,\mathbf{W}_{\boldy, \, \bom}(\lambda) } _{L^{\infty}(I, d\mu_{\bom}(\lambda))}    \Br\} \leq C_{\zeta_1} e^{-d_{H}(\x,\,\y)^{\zeta_1}} .
\eeq
\end{theorem}

\begin{proof}
Let us denote $f(\bom) = \norm{ \mathbf{W}_{\boldx, \, \bom}(\lambda) \,\,\mathbf{W}_{\boldy, \, \bom}(\lambda) } _{L^{\infty}(I, d\mu_{\bom}(\lambda))}$. Take $\x , \y \in \Ndspace.$ We will divide the proof into several cases.

Case 1: There exists $k \in N$ such that $L_k < d_{H}(\x,\,\y) \leq L_{k+1}$; i.e. the pair $\x$ and $\y$ is $L_k-distant.$
Denote $\cA_{k} = R(m, \, L_k, \, I, \, \boldx, \, \boldy)$. Then
\beq
\E \bigl\{f(\bom) \bigr\} = \E \bigl\{f(\bom) \,;\, \cA_k \bigr\} + \E \bigl\{f(\bom) \,;\, \cA_k^{c} \bigr\}
\eeq
On the set $\cA_k$, we have $f(\bom) \leq C e^{-m\tfrac{L}{4}} $ (Lemma \ref{part7lem1}),  
 so 
\beq
\E \bigl\{f(\bom) \,;\, \cA_k \bigr\} \leq C e^{-m\tfrac{L_k}{4}}.
\eeq
On the set $\cA_k^{c}$, we have $f(\bom) \leq 1 $ (Remark \ref{part7rem1}); thus
\beq
\E \bigl\{f(\bom) \,;\, \cA_k^{c} \bigr\} \leq \P (\cA_k^{c}) \leq e^{-L_k^{\zeta}}.
\eeq
Hence $\E \bigl\{f(\bom) \bigr\} \leq C_{1} e^{-d_{H}(\x,\,\y)^{\zeta_1}}$
for a slightly smaller $\zeta_1$.

Case 2: $d_{H}(\x,\,\y) \leq L_0$. By Remark \ref{part7rem1}, we have
\begin{align}
\E \bigl\{f(\bom) \bigr\} &\leq e^{d_{H}(\x,\,\y)^{\zeta_1}} e^{-d_{H}(\x,\,\y)^{\zeta_1}} \leq e^{L_0} e^{-d_{H}(\x,\,\y)^{\zeta_1}} \notag\\
& = C_2 \, e^{-d_{H}(\x,\,\y)^{\zeta_1}}.
\end{align}
Thus we get our desired result.
\end{proof}

The following result of \cite{GKjsp}, though trivial, plays a crucial role in this section, so we state it here without providing the proof.

\begin{lemma} \label{part7lem2}
Assume that for every $\zeta_1 \in (0, \, 1) $, we can find a constant $C_{\zeta_1}$ such that for every $ \x, \, \y \in \Ndspace$,
\beq \notag
\E \Bl\{ \norm{ \mathbf{W}_{\boldx, \, \bom}(\lambda) \,\,\mathbf{W}_{\boldy, \, \bom}(\lambda) } _{L^{\infty}(I, d\mu_{\bom}(\lambda))}    \Br\} \leq C_{\zeta_1} e^{-d_{H}(\x,\,\y)^{\zeta_1}} .
\eeq
Then for any $\zeta \in (0, \, 1)$ there exists a constant $C_{\zeta}$ such that
\beq \notag
\E \Bl\{ \sum_{\x, \, \y \in \Ndspace} e^{d_{H}(\x,\,\y)^{\zeta}} \scal{\x}^{-2\nu} \norm{ \mathbf{W}_{\boldx, \, \bom}(\lambda) \,\,\mathbf{W}_{\boldy, \, \bom}(\lambda) } _{L^{\infty}(I, d\mu_{\bom}(\lambda))} \Br\} < C_{\zeta}.
\eeq
Thus, it follows that for $\P$-a.e.\ $\bom$ we have
\beq
\sum_{\x, \, \y \in \Ndspace} e^{d_{H}(\x,\,\y)^{\zeta}} \scal{\x}^{-2\nu} \norm{ \mathbf{W}_{\boldx, \, \bom}(\lambda) \,\,\mathbf{W}_{\boldy, \, \bom}(\lambda) } _{L^{\infty}(I, d\mu_{\bom}(\lambda))} < \infty \notag.
\eeq
\end{lemma}

\begin{corollary} \label{part7cor1}
Suppose that for every $\zeta_1 \in (0, \, 1) $, we can find a constant $C_{\zeta_1}$ such that for every $ \x, \, \y \in \Ndspace$,
\beq \notag
\E \Bl\{ \norm{ \mathbf{W}_{\boldx, \, \bom}(\lambda) \,\,\mathbf{W}_{\boldy, \, \bom}(\lambda) } _{L^{\infty}(I, d\mu_{\bom}(\lambda))}    \Br\} \leq C_{\zeta_1} e^{-d_{H}(\x,\,\y)^{\zeta_1}} .
\eeq
Then for $\P$-a.e.\  $\bom$, $H_{\bom}$ exhibits pure point spectrum in the interval $(a, \,b)$ with the corresponding eigenfunctions decaying exponentially fast at infinity. Moreover, for $\mu_{\bom}-a.e.$ $\lambda \in I$, $\lambda$ is an eigenvalue of $H_{\bom}$ with finite multiplicity.
\end{corollary}

\begin{proof}
Since we know that there exists $\Omega_1$ where $\P(\Omega_1) = 1$ and for every $\bom \in \Omega_1$
\beq
\sum_{\x, \, \y \in \Ndspace} e^{d_{H}(\x,\,\y)^{\zeta}} \scal{\x}^{-2\nu} \norm{ \mathbf{W}_{\boldx, \, \bom}(\lambda) \,\,\mathbf{W}_{\boldy, \, \bom}(\lambda) } _{L^{\infty}(I, d\mu_{\bom}(\lambda))} < \infty \notag,
\eeq
let us take $\bom \in \Omega_1$. Then there exists a constant $C_{\bom} = C_{\bom,\,\zeta}$ such that 
\beq
\sum_{\x, \, \y \in \Ndspace}   \scal{\x}^{-2\nu} { e^{d_{H}(\x,\,\y)^{\zeta}}  \norm{ \mathbf{W}_{\boldx, \, \bom}(\lambda) \,\,\mathbf{W}_{\boldy, \, \bom}(\lambda) } _{L^{\infty}(I, d\mu_{\bom}(\lambda))}   } < C_{\bom}.
\eeq
It follows from Lemma \ref{part7lem2} that for any $\phi \in \H_{+}$, any $\x, \, \y \in \Ndspace$, and for $\mu_{\bom}-a.e.\,\,\lambda \in I$, we have
\beq
\norm{\Chi_{\x} \Pb_{\bom} (\lambda) \phi }\norm{\Chi_{\y} \Pb_{\bom} (\lambda) \phi } \leq {  C_{\bom}  \scal{\x}^{2\nu} \norm{T_{\x}^{-1} \Pb_{\bom} (\lambda) \phi }\norm{T_{\y}^{-1} \Pb_{\bom}(\lambda)  \phi }    }{e^{-d_{H}(\x,\,\y)^{\zeta}}  }.
\eeq
But $ \norm{ T_{\bolda}^{-1} \Pb_{\bom}(\lambda) \phi} \leq  2^{\tfrac{\nu}{2}} \scal{\bolda}^{\nu} \norm{  \phi}_{+},$ so
\begin{align}
\norm{\Chi_{\x} \Pb_{\bom} (\lambda) \phi }\norm{\Chi_{\y} \Pb_{\bom} (\lambda) \phi } \leq { C_{\bom} \,2^{\nu} \, \scal{\x}^{3\nu} \scal{\y}^{\nu} \norm{\phi}_{+}^{2}  }{ e^{-d_{H}(\x,\,\y)^{\zeta}}   } .
\end{align}
We know there exists a constant $\cZ < \infty$ such that if $d_{H}{(\x, \, \y)} \geq \cZ$ we get  $ \scal{\y}^{\nu} e^{-  d_{H}{(\x, \, \y)}^{\zeta} }  \leq e^{ -\tfrac{1}{2} d_{H}{(\x, \, \y)}^{\zeta} } $. Moreover, there exists a constant $\tilde{C} < \infty$  such that if $d_{H}{(\x, \, \y)} \leq \cZ$, then
$  \scal{\y}^{\nu} e^{-  d_{H}{(\x, \, \y)}^{\zeta} }  \leq \tilde{C}$. So for every $\y \in \Ndspace$,
\beq
\norm{\Chi_{\x} \Pb_{\bom} (\lambda) \phi } \norm{\Chi_{\y} \Pb_{\bom} (\lambda) \phi } \leq C_{1} \scal{\x}^{3\nu} e^{-\tfrac{1}{2}  d_{H}{(\x, \, \y)}^{\zeta} } \norm{\phi}_{+}^{2}
\eeq
for some constant $C_{1} = C_{1}(\bom, \, \zeta)$.
In particular, if $\Pb_{\bom} (\lambda) \phi$ is a generalized eigenfunction of $H_{\bom}$, then we can pick $\x_{0} \in \Ndspace$ such that $ \norm{\Chi_{\x_{0}} \Pb_{\bom} (\lambda) \phi } \neq 0$. So we get that for every $\y \in \Ndspace$
\beq
 \norm{\Chi_{\y} \Pb_{\bom} (\lambda) \phi } \leq \frac{ C_{1} \scal{\x_{0}}^{3\nu} e^{-\tfrac{1}{2}  d_{H}{(\x, \, \y)}^{\zeta} } \norm{\phi}_{+}^{2}  }{ \norm{\Chi_{\x_{0}} \Pb_{\bom} (\lambda) \phi }   }.
\eeq
It follows that $\Pb_{\bom} (\lambda) \phi \in \H $, and hence $\mu_{\bom} -a.e.$ $\lambda \in I$ is an eigenvalue of $H_{\bom}$.

To show finite multiplicity, it is enough for us to show that $\text{tr}\,E_{\lambda}(H_{\bom}):=\text{tr}\,\Chi_{\{ \lambda\}}(H_{\bom}) < \infty$. But
\begin{align}
&\mu_{\bom}(\lambda) \,\text{tr}\,E_{\lambda}(H_{\bom}) = \norm{T^{-1} E_{\lambda}(H_{\bom})}_{2}^{2} \,\,\text{tr}\,E_{\lambda}(H_{\bom}) = \text{tr } \Bl ( E_{\lambda} (H_{\bom}) T^{-2}E_{\lambda}(H_{\bom}) \Br ) \,\text{tr}\,E_{\lambda}(H_{\bom})  \notag \\
&  \;= \text{tr } \Bl (  T^{-2}  E_{\lambda} (H_{\bom}) \Br ) \,\text{tr}\,E_{ \lambda }(H_{\bom}) \leq \text{tr} \Bl (  \sum_{\x \in \Ndspace} \scal{ \x} ^{-2\nu}\Chi_{\x}  E_{\lambda} (H_{\bom}) \Br ) \,\text{tr} \Bl ( \sum_{\y \in \Ndspace} \Chi_{\y} E_{\lambda}(H_{\bom}) \Br ) \notag \\
&\; \leq \sum_{\x, \, \y \in \Ndspace} \scal{\x}^{-2\nu}\norm{ \Chi_{\x}  E_{\lambda} (H_{\bom})  }_{2}^{2} \norm{ \Chi_{\y}  E_{\lambda} (H_{\bom})}_{2}^{2}   \notag\\
& \; \leq \sum_{\x, \, \y \in \Ndspace} \scal{\x}^{-2\nu} Z_{\boldx, \, \bom}(\lambda)  ^{2} \norm{ T_{\boldx}^{-1}E_{\bom}(\lambda) }_2^{2} Z_{\boldy, \, \bom}(\lambda)  ^{2} \norm{ T_{\boldy}^{-1}E_{\bom}(\lambda) }_2^{2}. 
\end{align}
Since $\norm{ T_{\boldx}^{-1}E_{\bom}(\lambda) }_2^{2}$ is bounded uniformly for every $\x \in \Ndspace$, every $\bom$, and every $\lambda$, and $Z_{\bolda, \, \bom}(\lambda) \leq W_{\bolda, \, \bom}(\lambda) \leq \mathbf{W}_{\bolda, \, \bom}(\lambda) \leq 1  $,  we get
\begin{align}
\mu_{\bom}(\lambda) \,\text{tr}\,E_{\lambda}(H_{\bom}) &\leq C_{2}^{4} \sum_{\x, \, \y \in \Ndspace} \scal{\x}^{-2\nu} Z_{\boldx, \, \bom}(\lambda)  ^{2}  Z_{\boldy, \, \bom}(\lambda)  ^{2} \notag\\
&\leq C_{2}^{4} \sum_{\x, \, \y \in \Ndspace} \scal{\x}^{-2\nu} Z_{\boldx, \, \bom}(\lambda)    Z_{\boldy, \, \bom}(\lambda).
\end{align}
The result now follows from Lemma \ref{part7lem1}.
\end{proof}

\begin{proof}[Proof of Corollary \ref{PPS} (iii)]
Let us take $\zeta_1 > \zeta,$ and let $E_{n, \, \bom}$ be an eigenvalue of  $H_{\bom}$. For $\psi, \, \phi \in \text{Ran} \,\,\Chi_{E_{n, \, \bom}} (H_{\bom}) $, and $\x, \, \y \in \Ndspace$, we have
\begin{align}
\norm{\Chi_{\x}\phi} \norm{\Chi_{\y}\psi}  &\leq \Bl( W_{\boldx, \, \bom}(E_{n, \, \bom})     W_{\boldy, \, \bom}(E_{n, \, \bom}) \Br)
 \Bl( \norm{T_{\x}^{-1}\phi} \norm{T_{\y}^{-1}\psi} \Br) \\
 & \leq \Bl( \mathbf{W}_{\boldx, \, \bom}(E_{n, \, \bom})     \mathbf{W}_{\boldy, \, \bom}(E_{n, \, \bom}) \Br)  \Bl( \norm{T_{\x}^{-1}\phi} \norm{T_{\y}^{-1}\psi} \Br) \notag \\
 & \leq \norm{ \mathbf{W}_{\boldx, \, \bom}(\lambda) \,\,\mathbf{W}_{\boldy, \, \bom}(\lambda) } _{L^{\infty}(I, d\mu_{\bom}(\lambda))} \Bl( \norm{T_{\x}^{-1}\phi} \norm{T_{\y}^{-1}\psi} \Br) . \notag
\end{align}
Thus,  we can apply Lemma \ref{part7lem2} to get
\beq
\norm{\Chi_{\x}\phi} \norm{\Chi_{\y}\psi} \leq  C_{\zeta_1} e^{-d_{H}(\x,\,\y)^{\zeta}} \scal{\x}^{2\nu}  \Bl( \norm{T_{\x}^{-1}\phi} \norm{T_{\y}^{-1}\psi} \Br),
\eeq
so applying equation \eq{part7eq0} we get our desired result.
\end{proof}


\begin{thebibliography}{HuKNSV}

\bibitem[ASFH]{ASFH}  Aizenman, M.,  Schenker, J., Friedrich, R., 
Hundertmark, D.: Finite volume fractional-moment criteria for 
Anderson localization.
Comm. Math. Phys. \textbf{224}, 219-253 (2001)  


\bibitem[AW]{AWmp}  Aizenman, M., Warzel, S.: {Localization bounds for multiparticle systems}. Commun. Math. Phys. {\bf 290}, 903-934 (2009)


\bibitem[CBS]{CBS}  V. Chulaevsky, A. Boutet de Monvel, Y. Suhov: {Dynamical localization for a multi-particle model with an alloy-type external random potential}. Nonlinearity, \textbf{24} , 1451-1472 (2011)

\bibitem[CS1]{CS1}  V. Chulaevsky, Y. Suhov: {Wegner bounds for a two particle tight binding model}. Commun. Math. Phys. {\bf 283}, 479-489 (2008)


\bibitem[CS2]{CS2}  V. Chulaevsky, Y. Suhov: {Eigenfunctions in a two-particle Anderson tight binding model}. Commun. Math. Phys. \textbf{289}, 701-723 (2009) 

\bibitem[CS3]{CS3}  V. Chulaevsky, Y. Suhov: {Multi-particle Anderson Localization: Induction on the number of particles}. Math. Phys. Anal. Geom., no. 2, 117-139 (2009)



\bibitem[DK]{vDK} von Dreifus, H.,  Klein, A.: {A new proof of localization in
the Anderson tight binding model}.  Comm. Math. Phys. {\bf 124},
285-299 (1989). %\doi{10.1007/BF01219198}

\bibitem[GK1]{GK1}  Germinet, F., Klein, A.: {Bootstrap multiscale Analysis
and Localization in Random Media}. Commun. Math. Phys. \textbf{222}, 415-448 (2001)



\bibitem[GK2]{GKjsp}  Germinet, F., Klein, A.: {New Characterization of the region of coplete localization for random Schr\"odinger operators}. J. Stats. Phys. \textbf{122}, 73-94 (2006)

\bibitem[GK3]{GKber}  Germinet, F., Klein, A.: {A comprehensive proof of localization for continuous Anderson models with singular random potentials}. J. Eur. Math. Soc. \textbf{15}, 53-143 (2013)


\bibitem[KKS]{KKS} Klein, A.,  Koines, A., Seifert,  M.: {Generalized
eigenfunctions for waves in inhomogeneous media}. 
J. Funct. Anal. \textbf{190}, 255-291 (2002)

\bibitem[Ki]{Ki} Kirsch, W.: {An invitation to random Schr\"odinger operators. In Random Schr\"odinger-Operators. Panoramas et Syntheses \textbf{25}}. 
Societe Mathematique de France, Paris, 1-119 (2008)

\bibitem[Kl]{Kl} Klein, A.: {multiscale analysis and localization of random operators. In Random Schr\"odinger Operators. Panoramas et Syntheses \textbf{25}}. 
Societe Mathematique de France, Paris, 121-159 (2008)

\end{thebibliography}
\end{document}